\documentclass[qe,nameyear,final]{econsocart}
\RequirePackage[colorlinks,citecolor=blue,linkcolor=blue,urlcolor=blue,pagebackref]{hyperref}

\startlocaldefs

\theoremstyle{plain}
\newtheorem{assumption}{Assumption}
\newtheorem{theorem}{Theorem}
\newtheorem{proposition}{Proposition}
\newtheorem{lemma}{Lemma}
\theoremstyle{remark}

\newtheorem{example}{Example}
\newtheorem{definition}{Definition}
\newtheorem{remark}{Remark}

\usepackage{mathrsfs}
\usepackage{bbm}
\usepackage{enumitem}
\usepackage{mathtools}
\usepackage{threeparttable}
\usepackage{pgf,tikz,pgfplots}
\pgfplotsset{compat=1.16}
\usetikzlibrary{arrows}

\newcommand{\indicator}{\mathbbm{1}}
\newcommand{\norm}[1]{\left\Vert #1\right\Vert }
\newcommand{\real}{\mathbb{R}}
\newcommand{\perturb}[2]{#1_{#2}}

\DeclareMathOperator*{\argmax}{arg\,max}

\endlocaldefs

\begin{document}

\begin{frontmatter}

\title{Discordant Relaxations of Misspecified Models}
\runtitle{Discordant Relaxations of Misspecified Models}

\begin{aug}
%
%
%
\author[id=au1,addressref={add1}]{\fnms{Lixiong }~\snm{Li}\ead[label=e1]{lixiong.li@jhu.edu}}
\author[id=au2,addressref={add2}]{\fnms{D\'esir\'e}~\snm{K\'edagni}\ead[label=e2]{dkedagni@unc.edu}}
\author[id=au3,addressref={add31,add32,add33}]{\fnms{Isma\"el}~\snm{Mourifi\'e}\ead[label=e3]{ismaelm@wustl.edu}}
\address[id=add1]{%
\orgdiv{Department of Economics},
\orgname{Johns Hopkins University}}

\address[id=add2]{%
\orgdiv{Department of Economics},
\orgname{The University of North Carolina at Chapel Hill}}

\address[id=add31]{%
\orgdiv{Department of Economics},
\orgname{University of Toronto}}

\address[id=add32]{%
\orgdiv{Department of Economics},
\orgname{Washington University in St. Louis}}

\address[id=add33]{%
\orgname{NBER}}

\end{aug}

\support{The authors thank three referees for very detailed and insightful reports, that were instrumental to significant improvements in the paper (the usual disclaimer applies). The authors also acknowledge the  helpful comments from Marc Henry, Hiro Kaido, Francesca Molinari, and Aureo de Paula, and seminar audiences at Boston University, and UCL. Mourifi\'e thanks the support from Connaught  and SSHRC Insight Grants \# 435-2018-1273.
All errors are ours.}

\coeditor{\fnm{[Name} \snm{Surname}; will be inserted later]}

\begin{abstract}
In many set-identified models, it is difficult to obtain a tractable characterization of the identified set. Therefore, researchers often rely on non-sharp identification conditions, and empirical results are often based on an outer set of the identified set. This practice is often viewed as conservative yet valid because an outer set is always a superset of the identified set. However, this paper shows that when the model is refuted by the data, two sets of non-sharp identification conditions derived from the same model could lead to disjoint outer sets and conflicting empirical results. We provide a sufficient condition for the existence of such discordancy, which covers models characterized by conditional moment inequalities and the \cite{Artstein1983} inequalities. We also derive sufficient conditions for the non-existence of discordant submodels, therefore providing a class of models for which constructing outer sets cannot lead to misleading interpretations. In the case of discordancy, we follow \cite{Masten2020} by developing a method to salvage misspecified models, but unlike them, we focus on discrete relaxations. We consider all minimum relaxations of a refuted model that restores data-consistency. We find that the union of the identified sets of these minimum relaxations is robust to detectable misspecifications and has an intuitive empirical interpretation.
\end{abstract}

\begin{keyword}
\kwd{Partial identification}
\kwd{identified/outer set}
\kwd{misspecification} 
\kwd{nonconflicting hypothesis}
\end{keyword}

\begin{keyword}[class=JEL] 
\kwd{C12}
\kwd{C21}
\kwd{C26}
\end{keyword}

\end{frontmatter}

\section{Introduction}

A central challenge in the structural estimation of economic models is that the hypothesized structure often fails to identify a single generating process for the data, either due to multiple equilibria or data observability constraints. In such a context, the econometrics of partially identified models have been trying to obtain a tractable characterization of parameters compatible with the available data and maintained assumptions (hereafter \textit{identified set}). A question of particular relevance in applied work is that it is often very difficult to find a tractable characterization of the identified set and then to obtain a valid confidence region for it. To avoid this difficulty, a large part of the literature has been trying to provide a confidence region for an \textit{outer set}, i.e., a collection of values for the parameter of interest that contains the identified set but may also contain additional values.\footnote{See \cite{Molinari2020} for a detailed  discussion.} 
Because of its tractability, constructing a confidence region for an outer set has been entertained in various topics of studies where the parameters of interest are only partially identified, see for instance \cite{Blundell2007}, \cite{Ciliberto2009}, \cite{Aucejo2017}, \cite{Sheng2020}, \cite{dePaula2018}, \cite{Dickstein2018}, \cite{Honore2020}, \cite{Chesher2020}, \cite{Gualdani2021}, and \cite{BerryCompiani2022},  among many others. 

In most empirical studies, obtaining a tight outer set is very often interpreted as evidence for a small and informative identified set.\footnote{A tight outer set here refers to an outer set that is very small and informative.} This is because, under correct specification, any outer set contains the identified set. In this paper, we examine the implications of using outer sets for models that could be misspecified. We say a model is misspecified if the identified set of the model parameters is empty. We use refutation and misspecification interchangeably in this paper.

The first main contribution of this paper is to characterize a class of models for which outer sets based on non-sharp identification conditions may be discordant. For this class of models, as long as the model is misspecified, there always exist two sets of non-sharp identification conditions that fail to detect the violation of the model and at the same time yield outer sets that are disjoint with each other. Our result covers a large class of models studied in the partial identification literature, including models whose identified set is characterized by intersection bounds, conditional moment inequalities, or the \cite{Artstein1983} inequalities. The discordancy that we find is a negative property because the result provided by an outer set could entirely be driven by the set of non-sharp identification conditions chosen by the researcher, and that we could always consider an alternative choice that provides a result that conflicts with the initial one.

Discordant outer sets only exist when the model is misspecified. In theory, a researcher could run a model specification test before using an outer set. However, in practice, although it is possible to construct a \emph{non-sharp} specification test that only checks the sufficient conditions for model misspecification, constructing a \emph{sharp} specification test, which checks the necessary and sufficient conditions for the emptiness of the identified set, is as challenging as obtaining a sharp characterization of the identified set. The non-tractability of the latter is often the motivation to use outer sets in the first place.\footnote{See for instance, \cite{Sheng2020}, \cite{Gualdani2021}, and the empirical application in \citet[Section 6, footnote 42]{BerryCompiani2022}.} 

Therefore, our result shows that the usage of outer sets based on non-sharp identification conditions is an unreliable compromise. It suggests that looking for the sharp characterization of the identified set must not only be viewed as a theoretical exercise but also has important empirical relevance. The identified set not only exhausts all the identification restrictions in the model structure and assumptions but is also immune to the possible misleading conclusions of discordant submodels.

Our warning against the usage of outer sets based on non-sharp identification conditions should not discourage researchers from relaxing stringent primitive model assumptions and replacing them with weaker ones. Although the identified set of a weaker set of primitive model assumptions is necessarily an outer set of the stringent original model, it is different from an outer set based on non-sharp identification conditions of the original model. The identified set, derived from weaker primitive assumptions, has a clear and precise interpretation akin to the empirical content inherent in weaker primitive assumptions due to its sharpness. In contrast, the choice of non-sharp identification conditions is often driven by analytical or computational tractability. Researchers often lack primitive interpretations for the outer set based on these non-sharp identification conditions. These outer sets are only relevant to the empirical analysis when viewed as conservative bounds for the sharp results. Yet, this view of conservative bounds could be misleading because, as we show in this paper, the conclusion drawn from these outer sets can be driven entirely by the choice of non-sharp identification conditions instead of the empirical content of the model.

However, discordant outer sets do not exist in all refuted models, especially when they are based on weaker primitive assumptions instead of non-sharp identification conditions. We then derive sufficient conditions for the non-existence of discordant submodels. This second result characterizes a class of models for which constructing outer sets cannot lead to misleading interpretations. In this case, outer sets would be conservative but always robust.

Prior to our work, various papers have been concerned about misspecification in partially identified models. An important focus has been dedicated to analyzing the impact of model misspecification on standard confidence regions used for set-identified models.  \cite{Bugni_etal2012} analyze the behavior of usual inferential methods for moment inequality models under local model misspecification.  \cite{Ponomareva2011}  and \cite{Kaido2013}  consider the impact of misspecification on semiparametric partially identified models, respectively, in the linear regression model with an interval-valued outcome and in a framework where some nonparametric moment inequalities are correctly specified and misspecification is due to a parametric functional form. See also \cite{Allen2020} who propose a method for statistical inference on the minimum approximation error needed to explain aggregate data in quasilinear utility models.
It is worth noting that if one tries to find a confidence region for an outer set, none of the inference methods, including those developed in the previously cited papers, can resolve the specific issue we are raising here. This is because two non-empty outer sets derived from the same underlying model structure can lead to discordant results. Adopting one of these outer sets without checking the validity of the underlying model could lead to misleading conclusions. Therefore, we need to suggest a more primitive approach to deal with these discordant results in this paper.

This objective leads to our second main contribution, which consists of providing a method to salvage models that are possibly misspecified because of the existence of discordant misspecified submodels or discordant nonempty outer sets. The main intuition is to construct some minimum relaxation of the full model by removing discordant submodels until all remaining submodels are compatible. Because there could be multiple ways to relax a model to restore data consistency, we take the union of the identified set of all these relaxed models. By doing so, we construct what we call the \textit{misspecification robust bound}. We provide general sufficient conditions under which our misspecification robust bound exists and also provide an intuitive empirical interpretation for it. Intuitively, we will say that a hypothesis is robust to misspecification if the hypothesis is compatible with all relaxed models that are data consistent and is implied by at least one of those data-consistent relaxed models.

The  misspecification robust bound concept is related to the minimally relaxed identified set introduced in  \cite{Andrews2019}, and to the falsification adaptive set concept introduced in \cite{Masten2020}. The primary departure from \cite{Masten2020} lies in our emphasis on discrete relaxations, while \cite{Masten2020} focused exclusively on relaxing assumptions in a continuous manner. In general, the use of discrete or continuous relaxation depends on the empirical application under scrutiny. We explore various features of discrete relaxations beyond its formal definition.

It is worth noting that discrete relaxations of misspecified models have been entertained in various existing papers, see for instance, \cite{Manski2000, Manski2009}, \cite{Blundell2007}, \cite{Kreideral2012}, \cite{Chen2018}, \cite{Kedagni2021}, \cite{Mourifie2020},  among many others. In these papers, when the initial model is too stringent, they suggested alternative weaker assumptions that are believed to be more compatible with the empirical application under scrutiny and for which the identified/outer set is not empty. However, some alternative reasonable relaxations may generate results that are discordant with what they suggested. To mitigate this issue, our misspecification robust bound approach suggests to collect information from all reasonable discordant minimum relaxations of the initial model.

We organize the rest of the paper as follows. Section \ref{sec:intro} introduces  two simple leading  examples that will illustrate our main contributions. Section \ref{sec:misleading_general} presents our general setting and main results on the characterization of discordant submodels. Section \ref{sec:robust_submodel} discusses a class of models for which constructing outer sets do not lead to misleading interpretations. Section \ref{sec:MRB} introduces the misspecification robust bound used to salvage misspecified models. Section \ref{App} provides a numerical illustration of the discordancy issue by visiting  the widely used entry game model studied in \cite{Ciliberto2009}, and also illustrates our  misspecification robust bounds in a return to education example. 
The last section concludes, and additional results and proofs are relegated to the appendix. 

 \section{Introductory examples}\label{sec:intro}
Although the main idea of this paper can be applied to general models, we begin with these two straightforward examples to illustrate our main contributions.

 \subsection{First leading  example: Intersection bounds}\label{Lead1}
Let us consider a special case of the intersection bounds in \cite{CLR2013} in which a parameter $\theta$ is bounded by the conditional mean of
an upper and lower bounds, 
\begin{equation}\label{eq:intersection_bounds}
  E[\underline{Y}|Z=z] \le \theta \le E[\overline{Y}|Z=z]\quad \text{ almost surely,}
\end{equation}
where $\underline{Y}$ and $\overline{Y}$ are two observable random bounds and  $Z$  is a vector of instrumental variables. Let $\mathcal{Z}$ be the support of $Z$, and define\footnote{The $\sup$ and $\inf$ operators in \eqref{eq:definition_of_inf_sup} should be understood as essential supremum and essential infimum respectively. } 
\begin{equation}\label{eq:definition_of_inf_sup}
  \underline{\gamma} \equiv \sup_{z\in \mathcal{Z}}E[\underline{Y}|Z=z] \quad\text{and} \quad \overline{\gamma} \equiv
\inf_{z\in \mathcal{Z}}E[\overline{Y}|Z=z].
\end{equation}
The identified set of $\theta$ is the interval $[\underline{\gamma}, \overline{\gamma}]$ when $\underline{\gamma} \leq \overline{\gamma}$. 
We assume the following regularity condition holds in this example. 

\begin{assumption}\label{assu:reg}
Assume $E|\underline{Y}| < \infty$ and $E|\overline{Y}|<\infty$. In addition,
assume that the conditional expectations $E[\underline{Y}|Z]$, and $E[\overline{Y}|Z]$ exist and $E[\underline{Y}|Z] \le
E[\overline{Y}|Z]$ almost surely.
\end{assumption}
 
This simple framework encompasses some important treatment effect models, including discrete and continuous treatment models, see for instance \cite{Manski1990, Manski1994} and \cite{Kim2018} among many others.\footnote{For the sake of conciseness those examples are discussed in more details in Appendix \ref{intersec}.} In practice, model \eqref{eq:intersection_bounds} is sometimes implemented by solving its unconditional version,
\begin{equation}\label{eq:uncond_implementation}
  E\big[h(Z)(\theta - \underline{Y}) \big] \ge 0 \hspace{0.2em} \text{ and }\hspace{0.2em}
  E\big[h(Z)(\overline{Y} - \theta) \big] \ge 0,
\end{equation}
where $h$ is some nonnegative function mapping its input to $\mathbb{R}_+^m$ with $m < \infty$, and the inequalities in 
\eqref{eq:uncond_implementation} are vector inequalities.  The inference for
\eqref{eq:uncond_implementation} is typically much simpler than the inference for the original model
\eqref{eq:intersection_bounds}, especially when $Z$ is multi-dimensional. 
Let $\widetilde{\Theta}(h)$ be the identified set for $\theta$ in model \eqref{eq:uncond_implementation}. As made
explicit in the notation, $\widetilde{\Theta}(h)$ depends on the choice of instrumental function $h$. 
However, since \eqref{eq:intersection_bounds} implies \eqref{eq:uncond_implementation}, we know that for every choice of
$h$, $\widetilde{\Theta}(h)$ is always an outer set  of the interval $[\underline{\gamma}, \overline{\gamma}]$,
the identified set for $\theta$ in model \eqref{eq:intersection_bounds}, i.e., $[\underline{\gamma}, \overline{\gamma}] \subseteq \widetilde{\Theta}(h)$. This inclusion relation is often used as a justification for using model \eqref{eq:uncond_implementation}. Its identification result $\widetilde{\Theta}(h)$ is often viewed as a conservative bound for $[\underline{\gamma}, \overline{\gamma}]$, the identified set for model  \eqref{eq:intersection_bounds}.

Our first observation is that the result based on $\widetilde{\Theta}(h)$ is not always reliable. Later in Section \ref{sec:intersectio_IV}, we show that when the identified set of \eqref{eq:intersection_bounds} is empty, i.e., $\underline{\gamma} > \overline{\gamma}$, there always exist two $h$ and $h'$ such that both $\widetilde{\Theta}(h)$ and $\widetilde{\Theta}(h')$ are nonempty but $\widetilde{\Theta}(h)\cap \widetilde{\Theta}(h')$ is empty.  Thus, two researchers could apply the same model on the same data set and yet draw completely different conclusions from the outer sets by choosing different $h$ functions. For example, if one observes $\widetilde{\Theta}(h) \subseteq (0, +\infty)$ for some $h$, one should not jump directly to the conclusion that the sign of $\theta$ is positive without verifying the non-emptiness of the identified set, since in the case of emptiness, there are circumstances under which  another researcher may choose an alternative $h'$ such that  $\widetilde{\Theta}(h) \subseteq (-\infty,0)$. 

This caveat of outer sets is somewhat overlooked in the literature. As we listed some papers in the introduction, it is common for researchers to construct a confidence interval for an outer set and draw conclusions based solely on its result. If the model is refutable and a researcher only studies an outer set in the empirical analysis without knowing whether the identified set is empty or not, results based on an outer set could be misleading in the intersection-bound model. In section \ref{sec:misleading_general}, we show that this caveat is indeed a concern for some widely used partial identification models, and in Section \ref{App}, we illustrate this discordancy issue using the entry game model studied in \cite{Ciliberto2009}. 

On the other hand, there exist models for which constructing outer sets does not lead to misleading interpretations. In the following, we will introduce a model that belongs to this category as our second leading example. We study this class of models in more detail later in Section \ref{sec:robust_submodel}.

 \subsection{Second leading  example: Adaptive Monotone IV (AMIV)}\label{Lead2}
  Consider the following potential outcome model: $Y = \sum_{z\in \mathcal{Z}}\indicator(Z = z)[Y_{1z}D + Y_{0z}(1-D)]$,
where the treatment $D$ is binary and the support $\mathcal{Z}$ of instrument $Z$ is discrete and finite.  $Y_{dz}$ is
the potential outcome when the treatment and the  instrument are externally set to $d$ and $z$, respectively. Assume $Z$ is one-dimensional and assume, without loss of
generality that $\mathcal{Z} = \{1, ..., k\}$.  We are interested in the average potential outcomes $\theta_d
= \sum_{z}P(Z = z)E Y_{dz}$ for $d \in \{0,1\}$, and then average treatment effect (ATE), i.e., $\theta_1-\theta_0$.
In this framework, the seminal work of \cite{Manski1990}  derived sharp bounds on the ATE under three main assumptions: \ref{enu:support} the bounded support of the potential outcomes,  \ref{enu:MI} the exclusion restriction, i.e., $E Y_{dz}=E Y_{dz'}$ for $z\neq z'$, and \ref{enu:monotone} the mean independence assumption, i.e., $E[Y_{dz}|Z] = E[Y_{dz}]$ for all $z \in \mathcal Z$.
However, there are some empirical evidences ---for instance \cite{Ginther2000}, where the identified for the ATE  proposed by Manski is empty ---suggesting a violation of the set of these assumptions. To be able to say something meaningful on the ATE in such a context, we introduce the following assumptions, which is one way to relax  \citeauthor{Manski1990}'s (\citeyear{Manski1990}) assumptions. 

For any $z\in \{1, ..., k\}$, define assumption $a_z$ to be the collection of the following assumptions:

\begin{enumerate}[label=\emph{E.\arabic*}]
  \item\label{enu:support} for each $d\in\{0,1\}$ and any $t\in \{1, ..., k\}$, $P(Y_{dt}\in [\ \underline{y}_d\ ,\  \overline{y}_d\ ]) = 1$. 
\item\label{enu:MI} for each $d\in \{0,1\}$ and any $t\in \{1, ..., k\}, E[Y_{dt}|Z] = E[Y_{dt}]$ almost surely.
\item\label{enu:monotone}for each $d\in\{0,1\}$,  $Y_{dt} \le Y_{dt'}$ for all $t\le t'$, and $Y_{dt} = Y_{dz}$ for
  all $t \ge z$.
\end{enumerate}

Each assumption $a_z$ has three parts. $\ref{enu:support}$ requires the potential outcomes to have a bounded support. \ref{enu:MI} is a mean independence assumption associated to the potential outcome $Y_{dz}$. The novelty here is \ref{enu:monotone}, which is an adaptive  relaxation of the exclusion restriction. 
Indeed, in the  extreme case when $z=1$, \ref{enu:monotone}  is equivalent to the full exclusion restriction, i.e., 
$Y_{dz} = Y_{dz'}$ for all $d$, $z$ and $z'$, then \ref{enu:MI} and \ref{enu:monotone} are equivalent to $E[Y_{d}|Z] = E[Y_{d}]$, which is the restriction under which \cite{Manski1990} derived bounds on the ATE.
On the other extreme, when $z=k$,
\ref{enu:MI}  and \ref{enu:monotone} imply the MIV assumption introduced in \cite{Manski2000}, i.e., $z_1 < z_2 \Rightarrow E[Y_d\vert Z=z_1] \leq E[Y_d\vert Z=z_2]$.
However, when  $1<z< k$, we are in a middle-ground situation where the exclusion restriction is relaxed in such a way that $Y_{dz'}$ is monotone in $z'$, but remains flat for $z'\ge z$.  See Figure \ref{fig:illustration_Ydz} for an illustration of how $Y_{dz}$ depends on $z$ under \ref{enu:monotone}.

Because the cut-off point $z$ would be set based on the data through techniques elaborated in Section \ref{sec:MRB}, we refer to this assumption as the Adaptive Monotone IV (AMIV) assumption. The economic rationality of the AMIV is that, even if $Z$ is not a valid IV because it could positively affect the potential outcome, in some empirical contexts, it could be reasonable to consider that the marginal effect of the IV on the potential outcome becomes null after a certain cut-off point.

For example, let us consider the well-studied topic of measuring the returns to a college degree. In this empirical context, researchers have often used parental education as an instrumental variable for college education. Nonetheless, some argue that this instrument may not be valid. They point out that parents with higher education tend to have children with better unobserved skills that significantly impact potential earnings, despite the possibility that the marginal benefit of one extra year of a parent's education may be diminishing.
The AMIV assumption introduced in this subsection combines these two features. In terms of our notation, $Y$ is observed wage, $D$ is an indicator for college attendance, $Z$ is parental education, and $Y_{dz}$ is the potential wage when college attendance $D$ and parental education $Z$ are externally set to $d$ and $z$, respectively. In the context here, the AMIV assumption acknowledges that parental education may positively influence children, while also accommodating the potential for its marginal impact to vanish beyond a certain threshold. We elaborate on this application in Section \ref{Emp}.

\begin{figure}[!ht]
  \begin{center}
\scalebox{0.9}{\input{"illustate_Ydz.pgf"}}
  \end{center}
  \caption{Illustration of restriction \ref{enu:monotone} when $z=3$ and $k=5$.}\label{fig:illustration_Ydz}
\end{figure}
Notice that by construction, for all $z=1,...,k-1$, $a_{z}$ implies $a_{z+1}$. Therefore, we have $\Theta\left(\{a_{z}\}\right) \subseteq \Theta\left(\{a_{z+1}\}\right)$ for all $z \in \{1,...,k-1\}$. In section \ref{sec:robust_submodel}, we will show that this nested structure of the assumptions may lead to a situation where we cannot have discordant submodels. More generally, we will derive  sufficient conditions for the non-existence of discordant submodels. 

When comparing these two leading examples, it is important to recognize that the motivation that leads to the construction of the outer sets or the relaxation of the initial model is different. In the first example, the outer sets are generated from non-sharp restrictions and are only used as a device aiming to provide useful information about the identified set. This situation usually occurs when one does not know how to conduct inference directly for the identified set or when the inference for the identified set is computationally intractable to implement.
In Section \ref{sec:misleading_general} below, we will explore why such an approach could lead to unreliable results. In the second example, the outer sets are introduced in a more constructive manner aiming to relax a model refuted by the data. They are based on more primitive assumptions on the latent variables. In such a case, it is often possible to provide a series of outer sets that are not discordant with others. We will discuss this case in Section~\ref{sec:robust_submodel}.


%
%

\section{Misleading Submodels}\label{sec:misleading_general}

We view a model as a collection of constraints on the latent, observable variables, and the parameters. Throughout the paper, the parameter space $\Theta$ is assumed to be some subset in a metric space, which can be of finite or infinite dimensions. Let $A$ be some nonempty collection of these constraints. We consider $A$ as the \textit{full model} (or simply model when there is no confusion) and, $A' \subsetneq A$ as a \textit{submodel}. For any nonempty subset $A'\subseteq A$, we use $\Theta_I(A')$ to denote the set of parameter values that satisfy  all the constraints in $A'$. For each $a\in A$, we abbreviate $\Theta_I(\{a\})$ as $\Theta_I(a)$. Let $\emptyset$ denote the empty set. For our purpose, $\Theta_I(\emptyset)$ can be an arbitrary nonempty subset of $\Theta$ such that $\Theta_I(a)\subseteq \Theta_I(\emptyset)$ for any $a \in A$. For the sake of simplicity, when there is no confusion, we will just use $\Theta$ to refer to $\Theta_I(\emptyset)$.

By definition, $\Theta_I(A') \subseteq \Theta_I(A'')$ if  $A'' \subseteq A'$. As a result, for any $A' \subseteq A$, $\Theta_I(A')$ is an outer set of $\Theta_I(A)$.  Moreover, we say a submodel $A'$ is \emph{data-consistent} if $\Theta_I(A')$ is nonempty, and call it \emph{refuted} or \emph{misspecified} if the reverse is true. Since $\Theta_I(A) \subseteq \Theta_I(A')$ for any
$A'\subseteq A$, if the model $A$ is data-consistent, we know that each
$A'\subseteq A$ is also data consistent, and $\Theta_I(A')\cap \Theta_I(A'')$ is
nonempty for any two submodels $A'$ and $A''$.

In this section, we focus on $A$ that consists of constraints that could be  written in terms of the observable variables and model parameters only. This includes models where the identified sets are entirely characterized by a set of moment (in)equalities, including generalized method of moments (GMM) models,  or \cite{Artstein1983} inequalities involving only observables and the parameter of interest. We refer to this type of constraints as identification conditions. 
These identification conditions are often derived from primitive assumptions involving the latent variables.  In this restricted framework, we view $A$ as the sharp identification conditions, and we call $A'$ nonsharp identification conditions if  $\Theta_I(A) \subsetneq \Theta_I(A')$.  
For instance, in the first leading  example, the role of $A$ is played by the set of moment inequalities 
\eqref{eq:uncond_implementation} indexed by all instrumental functions $h$, with  $\Theta_I(A)=[\underline{\gamma},\overline{\gamma}]$. A submodel  $A' \subseteq A$ could refer to \eqref{eq:uncond_implementation} for a specific function $h$, with $\Theta_I(A')=\widetilde{\Theta}(h)$. If  $[\underline{\gamma}, \overline{\gamma}] \subsetneq \widetilde{\Theta}(h)$, then $A'$ is a nonsharp identification condition. 

When identification conditions are written in terms of the observable variables and model parameters only, they satisfy the following assumption.

\begin{assumption}\label{assu:mis_2}
For any $A'\subseteq A$, $\Theta_I(A') = \cap_{a\in A'}\Theta_I(a)$.
\end{assumption}

For example, all moment (in)equality models satisfy this assumption. This assumption may not hold when some $a \in A$ involves primitive restrictions on latent variables, as is the case in our second leading example.

In empirical works with partially identified models, researchers often use nonsharp identification conditions instead of the sharp ones. This is often motivated by two reasons: (i) sometimes the researchers may not even know the sharp characterization of the identified set, and (ii) the sharp identification conditions might be computationally intractable given the existing inferential methods.\footnote{See, for example, Berry and Compiani (2022, Section 6, footnote 42).}  When empirical results are based on an outer set obtained from nonsharp identification conditions, they are traditionally viewed as conservative yet valid because these outer sets are always supersets of the sharp identified set obtained using the sharp identification conditions.
However, we are going to present a theorem that shows that outer sets obtained from nonsharp identification conditions are not always reliable and could potentially be misleading in the presence of model misspecification. We need the following assumption for this formal result.

\begin{assumption}\label{assu:mis_1}
There exists a collection $\mathscr{C}$ of subsets of $A$ such that 
    \begin{enumerate}
    \item [(1)] $\forall A'\in \mathscr{C}$, $A'$ is data-consistent and consists of finite elements in $A$,
    \item [(2)] $\Theta_I(\cup_{A'\in \mathscr{C}}A') = \Theta_I(A)$,
    \item [(3)] either $\mathscr{C}$ is finite or, for each $A'\in \mathscr{C}$, $\Theta_I(A')$ is compact.  
    \end{enumerate}
\end{assumption}

To clarify this assumption, we start with the simple case where all the conditions can be verified for $\mathscr{C} = \{\{a\}: a\in A\}$. In such a scenario, Assumption \ref{assu:mis_1} breaks down into the following two parts: (\emph{i}) every $a\in A$ is data-consistent, and (\emph{ii}) either $A$ is finite or $\Theta_I(a)$ is compact for each $a\in A$. Part (\emph{i}) ensures that $A$ will not be refuted because a specific $a\in A$ is refuted. When $A$ is finite, this is enough to ensure that identification condtions are not mutually compatible when $A$ is refuted. This reasoning can also apply to cases when $A$ has infinite elements, given the compactness condition in part (\emph{ii}) is true. In Assumption \ref{assu:mis_1}, we permit $\mathscr{C}$ to be formulated in ways beyond $\mathscr{C} = \{\{a\}: a\in A\}$. By doing so, $\Theta_I(a)$ need not necessarily be compact for each $a\in A$ in infinite cases. We only require the identified set for some finite combinations of $a\in A$ is compact. This flexibility is helpful, for example, when some identification conditions determine the lower bound for parameters, while others define the upper bound. 

We would verify Assumption \ref{assu:mis_1} with $\mathscr{C} = \{\{a\}: a\in A\}$ for the first leading example in the next subsection. For an example of $\mathscr{C}$ being constructed in other ways, see  Appendix~\ref{sec:cond_moment_ineq} where Assumption \ref{assu:mis_1} is verified for conditional moment inequalities models. We are now ready to state the formal result.

\begin{theorem}\label{thm:misleading}
	Under Assumptions \ref{assu:mis_2} and \ref{assu:mis_1}, $\Theta_I(A)= \emptyset$ if and only if there exist two finite subsets $A', A''\subseteq A$ such that both $A'$ and $A''$ are data-consistent and $\Theta_I(A')\cap \Theta_I(A'') = \emptyset$.

Moreover, when $\Theta_I(A) = \emptyset$, for any data-consistent $B\subseteq A$, there exists two finite subset $B', B''\subseteq A$ such that both $B\cup B'$ and $B''$ are data-consistent and $\Theta_I(B\cup B')\cap \Theta_I(B'') = \emptyset$.
\end{theorem}

Theorem \ref{thm:misleading} tells us that when the model is refuted, an outer set derived from one set of nonsharp identification conditions could be completely different from an outer set obtained from a different set of nonsharp identification conditions. In such a case, the information delivered by an outer set depends mainly on which nonsharp identification conditions the researcher decides to use. Therefore, for this class of models, applied researchers must be very careful in interpreting outer sets based on nonsharp identification conditions. 
We will illustrate this point further with the first leading example in the next subsection. Later, in Section \ref{App:CT}, we provide a numerical illustration of Theorem \ref{thm:misleading} when applied to an entry game example.

Moreover, Theorem \ref{thm:misleading} shows that, for any data-consistent nonsharp identification condition, there always exists another set of nonsharp identification conditions that are discordant with some of its strengthened versions. This suggests that the issue of discordancy is not confined to certain pairs of outer sets. When the model is refuted, we can start with any data-consistent outer set and further tighten its bounds by applying more restrictions. After incorporating just a finite number of additional restrictions, it will inevitably be in conflict with another data-consistent outer set, even if the cardinalities of $A$ and $\mathscr{C}$ are uncountably infinite.

Although Theorem \ref{thm:misleading} seems to focus essentially on partially identified models, it also applies to point-identified models. Suppose that each outer set derived from each single identification condition is a singleton, i.e., $\Theta_I(a)$ is a singleton for all $a \in A$. Then, Theorem \ref{thm:misleading} says that the model is misspecified if and only if there exist two different identification conditions $a, a' \in A$ such that $\Theta_I(a) \neq \Theta_I(a')$. This is related to a long-existing observation in point-identified models: whenever a model is over-identified, one can test the model specification by comparing the point-estimates obtained from different identification conditions.
In point-identified models, the issue of over-identification or misspecification is a direct concern for the researchers. In partially identified models, however, applied researchers tend to believe that their results are more credible and thus less sensitive to misspecification. Therefore, they often interpret the tightness of an outer set as a signal of an informative identified set, as discussed in \cite{Molinari2020} regarding the ``usefulness" of outer sets. 
Theorem \ref{thm:misleading} provides a different perspective that, for a certain class of models, outer sets could be misleading. Therefore, applied researchers should be very careful in interpreting outer sets based on nonsharp identification conditions, even if the model is only partially identified. In the following, we analyze the implications of Theorem \ref{thm:misleading} on our first leading example.


\subsection{Intersection bounds example continued}\label{sec:intersectio_IV}
To begin, let us construct $A$ in this example. Define $\mathcal{H}^+_m$ to be the space of all nonnegative instrumental functions with dimension $m$. More formally, let $\mathcal{H}^+_m \equiv \{h: \mathcal{Z}\mapsto \mathbb R^m_+ \text{ such that } E\norm{h(Z)} < \infty$, $E\norm{\underline{Y}h(Z)} < \infty$, $E\norm{\overline{Y}h(Z)} < \infty \text{ and } E[h_i(Z)] > 0, \; \forall \; i=1,...,m \}$.
 Let $A$ be the set of all identification conditions \eqref{eq:uncond_implementation}  indexed by $h\in \mathcal{H}^+_1$.     
   The set $A$ constructed here represents sharp identification conditions, because \eqref{eq:intersection_bounds} holds if and only if \eqref{eq:uncond_implementation} holds for all $h\in \mathcal{H}^+_1$. 

Next, we verify Assumptions \ref{assu:mis_2} and \ref{assu:mis_1} hold. Note that, for each $a\in A$, $\Theta_I(a)$ is nonempty and compact. This is because each $a\in A$ corresponds to an $h\in \mathcal{H}^+_1$, and because
  for all $h\in \mathcal{H}^+_1$, the identified set $\widetilde{\Theta}(h)$ for model \eqref{eq:uncond_implementation} is equal to the following interval
  \begin{equation*}
    \left[ \frac{E[h(Z)\underline{Y}]}{E[h(Z)]},\  \frac{E[h(Z)\overline{Y}]}{E[h(Z)]}\right],
  \end{equation*}
  which is nonempty and compact under Assumption \ref{assu:reg}. We can let $\mathscr{C} = \{\{a\}: a\in A\}$. Then, $\Theta_I(A')$ is nonempty and compact for all $A'\in \mathscr{C}$. Moreover, $\Theta_I(A) = \Theta_I(\cup_{A'\in \mathscr{C}}A')$ by the construction of $\mathscr{C}$. Thus, Assumption \ref{assu:mis_1} is satisfied. Since \eqref{eq:uncond_implementation} are moment inequalities which only
  depend on observables and the parameter, Assumption \ref{assu:mis_2} is also satisfied.

  Note that, in this example, if $B\subseteq A$ and $B$ consists of $m$ assumptions, then $B$ refers to the submodel that
  \eqref{eq:uncond_implementation} holds for some $h\in \mathcal{H}^+_m$. 
  As a result, Theorem \ref{thm:misleading} implies that when \eqref{eq:intersection_bounds} is refuted, there must exist
  some $h_1 \in \mathcal{H}^+_{m_1}$ and $h_2\in \mathcal{H}^+_{m_2}$, such that $\widetilde{\Theta}(h_1)\neq \emptyset$
  and $\widetilde{\Theta}(h_2)\neq \emptyset$ but $\widetilde{\Theta}(h_1)\cap \widetilde{\Theta}(h_2)$ is empty. 
  In fact, because of the specific structures of  this example, we can even obtain the following  stronger result.
  
%
%
%
%
\begin{proposition}\label{thm:uncond_minothold}
  Suppose Assumption \ref{assu:reg} holds.  If the restriction in \eqref{eq:intersection_bounds} is
  refuted, i.e., $\underline{\gamma} > \overline{\gamma}$, then, for any $\theta$ in 
  $(\overline{\gamma}, \underline{\gamma})$, there exists some $h\in \mathcal{H}^+_{2}$ such that
  $\widetilde{\Theta}(h) = \{\theta\}$. Conversely, if there exists some integer $m$ and some $h\in \mathcal{H}^+_m$
  such that $\widetilde{\Theta}(h) = \{\theta\}$, then $\theta\in [\overline{\gamma}, \underline{\gamma}]$. 
\end{proposition}

When \eqref{eq:intersection_bounds} is refuted, Proposition \ref{thm:uncond_minothold} shows that the
unconditional moment restrictions can point identify any element in the crossed bound $(\overline{\gamma},
\underline{\gamma})$ with a properly chosen instrumental function. 
The width of $(\overline{\gamma}, \underline{\gamma})$ depends on the extent of the model violation: the worse the violation is, the wider this interval would be. In the extreme case where the mean independence condition is significantly violated such that $[E\underline{Y}, E\overline{Y}]\subseteq (\overline{\gamma}, \underline{\gamma})$, it implies that any point in the Manski worst-case bounds can be selected as the point identification result by a suitable choice of $h$. 

Proposition \ref{thm:uncond_minothold} also sheds some light on the implementation of the inference procedure in \cite{AS2013}, which is one of the most popular inference procedures used for conditional moment inequalities. \cite{AS2013} converts the conditional moment inequalities into unconditional moment inequalities in the same way as we transformed, in the first leading example, \eqref{eq:intersection_bounds} into \eqref{eq:uncond_implementation}. A notable distinction in \citeauthor{AS2013}'s~(\citeyear{AS2013})  approach is that the number of instruments, i.e., the dimension of $h$ in our notation, increases to infinity as the sample size grows. Our results show that the usage of infinite number of instruments in the limit is crucial to achieve reliable inference results. In contrast, if researchers implement their inference but choose an instrumental function with a finite and fixed dimension, their results could be spuriously informative and misleading, as underscored in Proposition \ref{thm:uncond_minothold}. See also the formal results in Appendix \ref{sec:cond_moment_ineq}.

It is worth noting that Theorem \ref{thm:misleading} applies to much more general frameworks. More precisely, in Appendices \ref{sec:cond_moment_ineq} and \ref{sec:random_set_capacity}, we provide sufficient conditions under which Theorem \ref{thm:misleading} applies to two widely used classes of partially identified models. The first is a class of models where the identified set is characterized by the following type of conditional moment inequalities:
\begin{equation}\label{GMI}
E[m(X,Z;\theta)|Z] \le 0 \text{ almost surely}.
\end{equation}
Here, $X \in \real^{k_1}$ and $Z\in \real^{k_2}$ are observable random variables, and $m(\cdot, \cdot\ ;\theta) \in \real$ is some known integrable function for each $\theta$. If researchers construct outer sets by transforming (\ref{GMI}) into finite-dimensional unconditional moment inequalities, then a similar discordancy issue may happen, as the one we have seen in the first leading example. See Appendix \ref{sec:cond_moment_ineq} for more details.

The second class of models that Theorem \ref{thm:misleading} could be applied to is the class for which the identified set can be characterized by \citeauthor{Artstein1983}'s (\citeyear{Artstein1983}) inequalities. In her recent survey, \cite{Molinari2020} shows that this class of models includes simultaneous-move finite games with multiple equilibria, auction models with independent private values, network formation models, treatment effect models, etc. In many of those cases, the number of \citeauthor{Artstein1983}'s (\citeyear{Artstein1983}) inequalities that characterize the (sharp) identified set is extremely high (very often much higher than the sample size of the data under use). In practice, for the sake of computational feasibility, empirical researchers often pre-select a finite collection of Artstein's inequalities to obtain an outer set. As examples, we could cite \cite{Ciliberto2009}, \cite{Haile2003}, \cite{Sheng2020}, \cite{Chesher2020}, and \cite{BerryCompiani2022}, among many others. In Appendix \ref{sec:random_set_capacity}, we show that those pre-selected nonsharp moment inequalities suffer the same issue pointed out in Theorem \ref{thm:misleading}.

In Section \ref{App}, we provide a numerical illustration of the discordancy issue by visiting the widely used entry game model studied in \cite{Ciliberto2009}. We will explore in more detail the consequences of this pre-selection procedure when the original model might be refuted and illustrate the implication of our Theorem \ref{thm:misleading} in this widely used framework.

\section{Compatible Submodels and Minimum Data-Consistent Relaxation}\label{sec:robust_submodel}

As discussed in the previous section, there could be discordant submodels when the full model is refuted. However, the falsification of the full model does not necessarily lead to discordance of the submodels. Unlike in the previous section, we now consider $A$ that consists of primitive assumptions on latent variables in addition to those on observable variables and the parameter of interest. In this wider class of models, we will present a sufficient condition that ensures that all data-consistent submodels are always compatible with each other. In this section, for the sake of simplicity, we focus on the case where $A$ is finite. For the case where $A$ is infinite, similar results could be derived under additional conditions. We elaborate on those conditions and results in Appendix \ref{sec:existence_of_minimum_relaxation}.



To state our result, we need to introduce a new concept. When the full model is refuted, we can obtain a data-consistent submodel by dropping or relaxing some of the assumptions. We say that a data-consistent submodel is a minimum relaxation if we relax the minimum number of assumptions needed to restore data consistency.

\begin{definition}\label{def:min_relax}
Let $\widetilde{{A}}$ be a subset of ${A}$. We say $\widetilde{{A}}$ is a \emph{minimum
data-consistent relaxation} of ${A}$ if $\Theta_I(\widetilde{{A}})$ is nonempty and for any $a\in
{A}\backslash \widetilde{{A}}$, $\Theta_I(\widetilde{A} \cup {a})$ is empty.
\end{definition}

It is worth noting that the concept of a minimum data-consistent relaxation depends on how the researcher defines each of the simple assumptions $a$ that constitute $A$. Therefore, all the subsequent results relying on the minimum data-consistent relaxation depend on the way the researcher constructs $A$. We will return to this point in Section \ref{sec:nonconflicting_statement}. Furthermore, a minimum data-consistent relaxation always exists when $A$ is finite; its existence requires additional conditions when $A$ is infinite. See Appendix \ref{sec:existence_of_minimum_relaxation}.
 To illustrate this concept, let us consider a simple example
where $A = \{a_1, a_2, a_3\}$. The identified sets of each $a_i$ are all closed intervals in $\real$ as shown in Figure \ref{fig:three_interval} with $\Theta_I(a_1) = [b,c]$, $\Theta_I(a_2) = [d, e]$, $\Theta_I(a_3) = [f, g]$ and $f\le b \le c < d \le e \le g$. Assume also for the purpose of illustration that $\Theta_I(\{a, a'\}) = \Theta_I(a)\cap\Theta_I(a')$ for $a,a'\in \{a_1, a_2, a_3\}$.

\begin{figure}[!ht]
  \caption{The three-interval example}
  \label{fig:three_interval}
\centering
\begin{tikzpicture}[line cap=round,line join=round,>=triangle 45,x=3.5cm,y=2.1cm, scale=1, every node/.style={scale=1.2}]
\clip(0.75,-0.4) rectangle (3.75,0.8);
\draw [line width=2.pt] (1.,0.5)-- (2.,0.5);
\draw [line width=2.pt] (2.5,0.5)-- (3.5,0.5);
\draw [line width=2.pt] (0.8,0.15)-- (3.6,0.15);

\begin{scriptsize}
\draw [fill=black] (1.,0.5) circle (1.0pt);
\draw[color=black] (1.0,0.4) node {$b$};
\draw [fill=black] (2.,0.5) circle (1.0pt);
\draw[color=black] (2.0,0.4) node {$c$};
\draw[color=black] (1.5,0.60) node {$\Theta_I(a_1)$};
\draw [fill=black] (2.5,0.5) circle (1.0pt);
\draw[color=black] (2.5,0.4) node {$d$};
\draw [fill=black] (3.5,0.5) circle (1.0pt);
\draw[color=black] (3.5,0.4) node {$e$};
\draw[color=black] (3, 0.60) node {$\Theta_I(a_2)$};
\draw [fill=black] (2.3,0.15) circle (1.0pt);
\draw[color=black] (0.8,0.05) node {$f$};
\draw [fill=black] (3.8,0.15) circle (1.0pt);
\draw[color=black] (3.6,0.05) node {$g$};
\draw[color=black] (2.2,0.25) node {$\Theta_I(a_3)$};
\end{scriptsize}
\end{tikzpicture}
\end{figure}

  In this example, both $\{a_1, a_3\}$ and $\{a_2, a_3\}$ are minimum data-consistent relaxations. And, $\{a_3\}$ is not a minimum data-consistent relaxation, since it will remain data-consistent after including $a_1$ or $a_2$. 
In general, minimum data-consistent relaxations may or may not be unique. We will defer the discussion of multiple minimum data-consistent relaxations to the next section. In this section, we focus on the situation where there exists a unique minimum data-consistent relaxation. In fact, the uniqueness of the minimum data-consistent relaxation ensures the absence of discordancy issues discussed in the previous section. More precisely, in Appendix \ref{DMDCR}, we show that the existence of two discordant submodels is equivalent to the existence of two minimum data-consistent relaxations that are also discordant each other under some regularity conditions. Therefore, the uniqueness of the minimum data-consistent relaxation ensures the absence of discordant sub-models.

The following result describes the conditions under which there exists a unique minimum data-consistent relaxation.

\begin{theorem}\label{thm:compatible_submodel} Suppose $A$ is finite. Then the following statements are equivalent:
\begin{enumerate}[label={(T\ref{thm:compatible_submodel}.C\arabic*)}]
\item\label{enu:ind_assumptions} for any $A'\subseteq A$, $A'$ is data-consistent if and only if all $a\in A'$ are data-consistent.
\item \label{enu:unique_MDR} There exists a unique minimum data-consistent relaxation $A^*$.
\end{enumerate}
\end{theorem}

Theorem \ref{thm:compatible_submodel} through its condition \ref{enu:ind_assumptions} provides a way to check whether there exists a unique minimum data-consistent relaxation $A^*$. 
Note that Condition \ref{enu:ind_assumptions} does not hold in every model: in a general model, each $a\in A'$ being data-consistent is necessary but not sufficient for $A'$ to be data-consistent, because individual data-consistent assumptions might not be mutually compatible when combined.
Condition \ref{enu:ind_assumptions} can be verified by investigating when each $A' \subseteq A$ is data-consistent, which could be done even without seeing the data. In terms of interpretation, condition \ref{enu:ind_assumptions} implies that all data-consistent submodels are compatible with each other. It also implies that the set of data-consistent submodels, i.e., $\{A'\subseteq A: \Theta_I(A')\neq \emptyset\}$, is closed under the union operation: the union of data-consistent submodels remains data-consistent.

When the full model $A$ is data-consistent, the unique minimum data-consistent relaxation $A^*$ is just equal to $A$. When $A$ is refuted, $A^*$ can be viewed as the model learned from the data by removing all refuted assumptions in $A$ while keeping all the data-consistent ones. Indeed, condition \ref{enu:ind_assumptions} suggests that $A^*=\{a \in A:\Theta_I(a)\neq \emptyset\}$. The interpretation of $A^*$ and its role as the unique minimum data-consistent relaxation will be studied further in Appendix \ref{nc}. One way to illustrate condition \ref{enu:ind_assumptions} is to consider our second leading example.

\subsection{AMIV  example continued} \label{sec:AMIV}

Recall that by construction, for all $z=1,...,k-1$, $a_{z}$ implies $a_{z+1}$. In addition, define $a^\dagger$ as the collection of
\ref{enu:support} and \ref{enu:MI}. 
Let $A = \{a_1,...,a_k, a^\dagger\}$ be the collection of all assumptions. Then, the full model $A$ is the classic mean independence assumption considered in \cite{Manski1990}.
Within this second  leading  example, all assumptions are nested. That is, for any two $a,a'\in A$, either $a$ implies $a'$ or $a'$ implies $a$. Therefore, the data-consistency of a set of assumptions is equal to the data-consistency of the strongest assumption in that set, which implies the validity of \ref{enu:ind_assumptions}.\footnote{However, it is worth noting this nested structure is not necessary for  condition \ref{enu:ind_assumptions} to hold.}
 Therefore, \ref{enu:ind_assumptions} holds in this example. Theorem~\ref{thm:compatible_submodel} then implies that all data-consistent submodels will be compatible with each other and there exists a unique minimum data-consistent relaxation $A^*$.
The following result characterizes the identified set of $A^*$. To state the result, we use the following notations: $\underline{Y}_d = Y\indicator(D = d) + \underline{y}_d \indicator(D \neq d) $, 
$\overline{Y}_d = Y\indicator(D = d) + \overline{y}_d \indicator(D \neq d)$,
$\underline{q}_{dz}=E[\underline{Y}_d|Z = z]$, and $\overline{q}_{dz} = E[\overline{Y}_{d}|Z = z]$. 

\begin{proposition}\label{prop:Monotone_Exclusion_ATE}
  Assume that $P(Y\in [\underline{y}_d, \overline{y}_d]\big|D = d) = 1$ for any $d\in \{0, 1\}$. 
   Let $\theta = (\theta_1,\theta_0)$ be the parameter of interest. 
  Then, model $A$ always has a unique minimum data-consistent relaxation $A^*$, and $A^*$ always contains $a^\dagger$. 
  In addition, for any $z=1,...,k$,  $a_z\in A^*$ if and only if the following two conditions hold for each $d\in \{0,
  1\}$:
  \begin{equation}\label{Test1}
    \forall z' < z,\quad \max(\underline{q}_{dt}: t \le z') \le \min(\overline{q}_{dt}: t \ge z')
  \end{equation}
  and 
  \begin{equation} \label{Test2}
  \max(\underline{q}_{dt}: t = 1,...,k ) \le \min(\overline{q}_{dt}: t \ge z)
  \end{equation}
  Hence, $a_z\in A^*$ implies that $a_{z'} \in A^*$ for all $z' > z$. Moreover, if $\{z: a_z \in
  A^*\}$ is nonempty, define $z^* = \min\{z: a_z \in A^*\}$ and
  \begin{multline}\label{eq:idset_AMIV}
    \Gamma_{d,z^*} = \left[\sum_{z < z^*} P(Z
    = z)\max(\underline{q}_{dt}, t \le z)  + \sum_{z\ge z^*} P(Z
    = z) \max(\underline{q}_{dt}: t = 1,...,k),\right. \\
    \left. \sum_{z < z^*}P(Z = z)\min(\overline{q}_{dt}:t \ge z) + \sum_{z\ge z^*} P(Z = z)
    \min(\overline{q}_{dt}:  t\ge z^*) \right].
  \end{multline}
  Then, $\Theta_I(A^*) = \Gamma_{1,z^*}\times \Gamma_{0, z^*}$.
  If $\{z: a_z \in A^*\}$ is empty, then $\Theta_I(A^*) = \left[E[\underline{Y}_1], E[\overline{Y}_1]\right] \times
  \left[E[\underline{Y}_0], E[\overline{Y}_0]\right] $. 
\end{proposition}

\begin{remark}\label{remark:AMIV}
It is worth noting that, for simplicity, we impose the cut-off $z^*$ to be the same for all potential outcomes in \ref{enu:monotone}; however, we do not need to do so. We can let the data determine the cut-offs for each potential outcome separately.
\end{remark}


%
%

\section{Misspecification Robust Bounds}\label{sec:MRB}

In this section, we consider cases where there are multiple data-consistent relaxations. According to Theorem \ref{thm:compatible_submodel}, the multiplicity of minimum data-consistent relaxations is a necessary condition for the existence of discordant submodels. Indeed, whenever there are two mutually incompatible data-consistent submodels, there are at least two minimum data-consistent relaxations. If there is no reason to favor one submodel over another ex-ante, it is reasonable to consider all of these relaxations.

\begin{definition}
  Let $\mathscr{A}_R$ be the set of all minimum data-consistent relaxations. The \emph{misspecification robust
  bound} $\Theta^*_I$ is defined as $\Theta^*_I \equiv \cup_{\widetilde{A}\in \mathscr{A}_R}\Theta_I(
  \widetilde{A})$.
\end{definition}

The misspecification robust bound concept is similar to the falsification adaptive set concept introduced in \cite{Masten2020}. However, a distinctive feature of this section is that we focus on discrete relaxations, where an assumption is either dropped or kept, while \cite{Masten2020} focuses exclusively on relaxing assumptions in a continuous way. In general, the type of relaxation depends on the empirical question under study. In the following, we derive the misspecification robust bound for our two leading examples.

\subsection{Intersection bounds example continued}
For  the model \eqref{eq:intersection_bounds}, the misspecification robust bound is given in the following result.


  \begin{proposition}\label{prop:intersection_bound}
  Suppose Assumption \ref{assu:reg} holds, then:
  \begin{equation}\label{eq:introductory_example_mfaoiw}
  \Theta^*_I = \begin{cases}
    [\underline{\gamma}, \overline{\gamma}]& \text{if }   \underline{\gamma} \leq \overline{\gamma},\\
    [\overline{\gamma}, \underline{\gamma}]& \text{if } \overline{\gamma} < \underline{\gamma}, \;\;\; P(E[\underline{Y}|Z]\le \overline{\gamma}]) > 0 \text{ and }
    P(E[\overline{Y}|Z]\ge \underline{\gamma}]) > 0,  \\
    (\overline{\gamma}, \underline{\gamma}]& \text{if } \overline{\gamma} < \underline{\gamma}, \;\;\;  P(E[\underline{Y}|Z]\le \overline{\gamma}]) = 0 \text{ and }
    P(E[\overline{Y}|Z]\ge \underline{\gamma}]) > 0,\\
    [\overline{\gamma}, \underline{\gamma})& \text{if } \overline{\gamma} < \underline{\gamma}, \;\;\;  P(E[\underline{Y}|Z]\le \overline{\gamma}]) > 0 \text{ and }
    P(E[\overline{Y}|Z]\ge \underline{\gamma}]) = 0,\\
    (\overline{\gamma}, \underline{\gamma})& \text{if } \overline{\gamma} < \underline{\gamma}, \;\;\;  P(E[\underline{Y}|Z]\le \overline{\gamma}]) = 0 \text{ and }
    P(E[\overline{Y}|Z]\ge \underline{\gamma}]) = 0.\\
  \end{cases}
  \end{equation}
\end{proposition}
A direct implication of  Proposition \ref{prop:intersection_bound} is that if $P(E[\underline{Y}|Z]\le \overline{\gamma}) > 0$  and 
$P(E[\overline{Y}|Z]\ge \underline{\gamma}) > 0$ hold, which are mild technical requirements,  the misspecification robust bound simplifies to  $\Theta^*_I=[\min(\underline{\gamma},\overline{\gamma}), \max (\underline{\gamma},\overline{\gamma})]$ whether or not the full model is refuted.

\subsection{AMIV  example continued}
A direct implication of  Proposition \ref{prop:Monotone_Exclusion_ATE} is that the AMIV model  has a unique minimum data-consistent relaxation $A^*$ which can be summarized as follows:
$A^*=\{a^\dagger\}\cup\{a_z:   \text{ Equations } (\ref{Test1}) \text{ } \& \text{ } (\ref{Test2}) \text{ hold}\}$. Therefore,  the misspecification robust bound  for the AMIV model is:

\begin{equation}
\Theta^*_I = \begin{cases}
\Gamma_{1,z^*}\times \Gamma_{0, z^*} & \text{if } A^*\neq \{a^\dagger\},\\
\left[E[\underline{Y}_1], E[\overline{Y}_1]\right] \times
\left[E[\underline{Y}_0], E[\overline{Y}_0]\right] & \text{if } A^*=\{a^\dagger\},
\end{cases}
\end{equation}

where $z^* = \min\{z: a_z \in A^*\}$.

\subsection{Empirical interpretation of the misspecification robust bound}\label{sec:nonconflicting_statement}

As we pointed out earlier, $\Theta^*_I$ will depend on how the researcher decides to define the assumptions $a$ that constitute $A$. Different constructions of $A$ correspond to different ways to relax a refuted model. It is inevitable that there almost always exist multiple ways to relax a stringent and refuted model, and different ways of relaxations would lead to different results. Instead of drawing a general conclusion about which relaxation approach is superior, we believe it is important to offer an empirical interpretation of $\Theta^*_I$ for a given $A$. In this way, even if different researchers may construct $A$ based on their own economic interpretations of the model, they would have a clear interpretation of their results.

In Theorem \ref{thm:interpretation} in Appendix \ref{nc}, we show that the misspecification robust bound $\Theta^*_I$ is both \emph{rationalizable} and \emph{nonconflicting} in the following sense:

 \begin{itemize}
    \item (Rationalizable) The statement that $\Theta^*_I$ contains the true parameter is implied by some data-consistent submodel. That is, there exists some data-consistent submodel $A'\subseteq A$ such that $\Theta_I( {A}')\subseteq \Theta^*_I$. 
    \item (Nonconflicting) The statement that $\Theta^*_I$ contains the true parameter is not rejected by any data-consistent submodel. That is, there does not exist a data-consistent submodel $A'\subseteq A$  such that $\Theta_I(A')\cap \Theta^*_I = \emptyset$.
  \end{itemize}
  When the full model is refuted, different data-consistent submodels can imply different and potentially discordant statements on $\theta$. Among all possible statements on $\theta$, we think that being rationalizable and nonconflicting is a minimum requirement for a statement to be robust to model misspecification. If a statement fails to be rationalizable, then it is not implied by any of the data-consistent submodels. If a statement is not nonconflicting, then it is rejected by some data-consistent submodels.
  
The fact that $\Theta^*_I$ is both rationalizable and nonconflicting gives it an interesting empirical interpretation. Consider the simple case where $\theta$ is a scalar. Suppose we are interested in the sign of $\theta$. And, suppose $\Theta^*_I$ turns out to be within the positive real line, i.e., $\Theta^*_I\subseteq \real{++}$. Then, it means that some submodels identify the sign of $\theta$ to be positive, and whenever the sign of $\theta$ can be identified by a submodel, the sign of $\theta$ is always positive.

In some cases, $\Theta^*_I$ is the \textit{smallest} set that is both rationalizable and nonconflicting. That is, for any $\widetilde{\Theta}\subseteq \Theta$, $\widetilde{\Theta}$ is both rationalizable and nonconflicting if and only if $\Theta^*_I \subseteq \widetilde{\Theta}$. In this case, $\Theta^*_I$ could have richer interpretations. Consider the previous simple example again. Suppose it turns out that $\Theta^*_I\cap \real_{++} \neq \emptyset$ and $\Theta^*_I\cap \real_{--} \neq \emptyset$ so that $\theta\in \Theta^*_I$ does not imply the sign of $\theta$. If we know $\Theta^*_I$ is the smallest rationalizable and nonconflicting set, then we have the following conclusion: neither $\theta$ is positive nor $\theta$ is negative are rationalizable and nonconflicting statements. In other words, the value of $\Theta^*_I$ in this case implies that the data and the model cannot provide a clear statement on the sign of $\theta$. Finally, in Theorem \ref{thm:smallest_if} in Appendix \ref{nc}, we show that $\Theta^*_I$ would be the smallest rationalizable and nonconflicting set if there exists a unique minimum data-consistent relaxation or the identified set for each minimum data-consistent relaxation is a singleton.

\subsection{Discrete Relaxation versus Continuous Relaxation}\label{sec:comparison_discrete_continuous}

As can be seen, the misspecification robust bound relaxes a refuted model in a discrete way: an assumption is either fully kept or dropped during the relaxation. There are many other ways to relax and salvage a refuted model. One can also relax assumptions continuously as in \cite{Masten2020}. In general, different relaxations will lead to different results, and it is hard to compare all the possible approaches. However, there does exist a special case where discrete relaxation always leads to more informative results than any other ways of relaxations.

In order to make an adequate comparison, we need to introduce the terminology used in \cite{Masten2020}.
For any $\epsilon \in [0, 1]$ and any $a\in A$, let $\perturb{a}{\epsilon}$ denote the assumption after relaxing  assumption $a$.  The degree of relaxation is measured by $\epsilon$: when $\epsilon = 0$, $\perturb{a}{\epsilon} = a$; when $\epsilon\in (0, 1)$, the assumption $a$ is partially relaxed but the exact form of $\perturb{a}{\epsilon}$ would
depend on the specific way of relaxation chosen by the researcher; when $\epsilon = 1$, the assumption $a$ is completely relaxed and $\perturb{a}{\epsilon}$ is a null assumption which does not impose any restriction. Assume the relaxation is monotone: if $\epsilon_1 \le \epsilon_2$, $\perturb{a}{\epsilon_1}$ is stronger than  $\perturb{a}{\epsilon_2}$, in the sense that $\perturb{a}{\epsilon_1}$ implies $\perturb{a}{\epsilon_2}$.
For any $\delta: A\to [0, 1]$, define $A(\delta) \equiv \{\perturb{a}{\delta(a)}: a\in
A\}$ as the perturbed full model. For any two $\delta_1: A\to[0, 1]$ and $\delta_2: A\to[0, 1]$, we write $\delta_1
< \delta_2$ if $\delta_1(a) \le \delta_2(a)$ for all $a\in A$ and $\delta_1(a) < \delta_2(a)$ for some $a\in A$. 
Then, the \emph{falsification frontier} (FF) in \cite{Masten2020} can be defined as $FF = \{\delta: A\to [0, 1]\
  : \Theta_I(A(\delta))\neq \emptyset\text{ and there does not exist }\delta'\text{ such that
} \Theta_I(A(\delta'))\neq \emptyset, \Theta_I(A(\delta')) \subsetneq \Theta_I(A(\delta)) \text{ and } \delta'
< \delta\}$. We slightly modified the definition of the falsification frontier of \cite{Masten2020} to ensure the nonemptiness of $FF$ in some special cases.\footnote{
  The original definition in \cite{Masten2020}, written in our notation, is $FF = \{\delta: A\to [0, 1]\
  : \Theta_I(A(\delta))\neq \emptyset\text{ and there does not exist }\delta'\text{ such that
} \Theta_I(A(\delta'))\neq \emptyset \text{ and } \delta'
< \delta\}$. With our modified definition, we do not need to worry about the possibility that there is a sequence of $\{\delta_i:i \geq 1\}$ such that $\delta_n \rightarrow \delta^{\ast}$, $\Theta_I(A( \delta^{\ast}))= \emptyset$, $\Theta_I(A( \delta_n))=\Theta_I(A( \delta_1)) \neq \emptyset$ and  $\delta_{n+1} < \delta_{n}$ for all $n\geq 1$.}
Then, the  \emph{falsification adaptive set} $\Theta^\dagger_I$ is defined as
$\Theta^\dagger_I = \cup_{\delta\in FF}\Theta_I(A(\delta))$. 

Note that $\Theta^\dagger_I$ depends on the specific way that one chooses to relax the assumptions.
If one chooses to relax them discretely, i.e., if $a_{\epsilon} = a_{\indicator(\epsilon > 0)}$ for any $\epsilon$ and $a$, then $\Theta^\dagger_I$ is equal to the minimum data-consistent relaxation $\Theta^*_I$. If one chooses a different way of relaxation, the $\Theta^\dagger_I$ is generally different.  In some special cases, however, $\Theta^*_I$ is always included in $\Theta^\dagger_I$ no matter which way of relaxation is chosen. More precisely, 
whenever  for any minimum data-consistent relaxation $\widetilde{A}$, $\Theta_I(\widetilde{A})$ is a singleton, it can be shown that 
 $\Theta^*_I\subseteq \Theta^\dagger_I$ for any type of relaxation chosen by the researcher.
 We formally state and prove  this result, respectively, in Theorem  \ref{thm:compared_to_perturbation_corrected} in Appendix \ref{Disc}.

\section{Numerical and Empirical Illustrations}\label{App}

In this section, through numerical exercises, we illustrate two of our main theoretical results. In Section \ref{App:CT}, we consider the entry game model studied in \citet[CT]{Ciliberto2009}. We simulate a misspecified entry game model such that the sharp identification conditions deliver an empty identified set. In such a context, we show that it is possible to generate multiple non-empty conflicting outer sets by just selecting different sets of nonsharp identification conditions. This illustrates the discordancy issue raised in Theorem \ref{thm:misleading}. In Section \ref{Emp}, we revisit a return to college application and report the estimated identified set of the minimum data-consistent relaxation derived under the assumption that parental education satisfies the AMIV assumption.


%
%
%
%
%


\subsection{Numerical illustration of discordant outer sets in an entry game model}\label{App:CT}

Consider an entry game model with $m$ players. Each player $i$ chooses $Y_i\in \{0, 1\}$ to maximize its payoff:
\begin{equation}\label{eq:mis_model_entry_game}
	\pi_{i} = Y_i \left(\alpha_i + X_i\beta - \sum_{j\neq i}\delta^{i}_j Y_j  + \epsilon_{i}\right)
\end{equation}
where $\delta^{i}_j$ is player $j$'s competition impact on player $i$, $\epsilon = (\epsilon_1, ..., \epsilon_m)\sim N(0, I_m)$, and $I_m$ is the identity matrix. 
Denote $\alpha = (\alpha_i: i=1,...,m)$, $\delta = (\delta^{i}_j: i\neq j)$,  and, let $\theta = (\alpha, \beta, \delta)$ collect all the parameters. We model the players' behaviour as pure-strategy Nash equilibrium. Therefore, we restrict the parameter space $\Theta$ to be the set of parameters where pure-strategy equilibria exist with probability $1$. 
This is the same empirical model used in \cite{Ciliberto2009}, except for two simplifications: 
(i) we assume that $\beta$ is the same for all players; (ii) we assume that the distribution of $\epsilon$ is known. 

Let $\mathcal{K}$ be the collection of all subsets of $\mathcal{Y} \equiv \{0, 1\}^{m}$. Let $F$ denote the joint distribution of $(Y, X)$ in the data. Define $\Gamma(x, \epsilon;\theta)$ as  the set of  all pure-strategy Nash equilibria given $(x, \epsilon)$, the parameter $\theta$ and the payoff function in (\ref{eq:mis_model_entry_game}).
As shown in Galichon and Henry (2011) and equivalently in Beresteanu, Molchanov, and Molinari (2011), $\theta$ belongs to the identified set $\Theta_I(F)$ if and only if 
\begin{equation}\label{Arstein:CT}
 \forall K\in \mathcal{K},\  \mathbb{P}_{F}(Y\in K|X) \le \mathbb{P}(\Gamma(X, \epsilon;\theta) \cap K\neq \emptyset|X), \, X-a.s.
\end{equation}

Equation (\ref{Arstein:CT}) characterized the \cite{Artstein1983} inequalities associated to the entry game model under study. For each fix covariate $x$, we have $2^{2^m}$ inequalities to be checked. In this case, our full model $A$ is defined by the whole set of \citeauthor{Artstein1983}'s inequalities, and then its cardinality is $2^{2^m} \times Card(X)$. 
This becomes easily non-tractable even for a relatively small number of firms, i.e., for instance, for a fixed $x$, and 5 firms we have $2^{2^5}=4, 294, 967, 296$ inequalities to be checked. Therefore, in practice, outer sets are almost always used. Let $A'$ be a subset of $\mathcal{K}$. Then, the outer set associated with $A'$ is 
\begin{equation*}
\Theta_I(F, A') \equiv \big\{\theta\in \Theta: \forall K\in A',\  \mathbb{P}_{F}(Y\in K|X) \le \mathbb{P}(\Gamma(X, \epsilon;\theta) \cap K\neq \emptyset|X), \, X-a.s.\big\}
\end{equation*}
$A'$ is a nonsharp identification condition whenever  $\Theta_I(F, A')\neq \Theta_I(F, \mathcal{K})\equiv \Theta_I(F)$. 
In CT, for each fixed $x$, they considered the outer set associated with $A_{ct}$ defined as $A_{ct} = \{\{y\}: y\in \mathcal{Y}\}\cup  \{\{y\}^c: y\in \mathcal{Y}\}$
where $\{y\}^c$ stands for the complement set of $\{y\}$ in $\mathcal{Y}$. It is worth noting that whenever $m>2$, $A_{ct}$ is a nonsharp identification condition.

\subsubsection{Data generating process (DGP)} In order to illustrate the issue of non-reliability of the outer sets in presence of misspecification, we generate  a joint distribution $F$ from a game that might be different from the model \eqref{eq:mis_model_entry_game}. Assume that, in the data generating process, each player $i$ chooses $Y_i\in \{0, 1\}$ to maximize the following payoff instead of (\ref{eq:mis_model_entry_game}):

\begin{equation}\label{eq:true_model_entry_game}
	\pi_{i} = Y_i \left(\alpha_i + X_i\beta - \sum_{j\neq i}\delta^{i}_j Y_j - \sum_{j_1, j_2\neq i} \gamma^i_{j1, j2}Y_{j1}Y_{j2} + \epsilon_{i}\right).
\end{equation}

The extra vector of parameter $\gamma\equiv \{\gamma^{i}_{j_1, j_2}: \text{where }(i,j_1, j_2) \text{ are mutually different}\}$ captures the second-order competition effect. When $\gamma = 0$, the model in \eqref{eq:mis_model_entry_game} is correctly specified. If $\gamma\neq 0$, the model in \eqref{eq:mis_model_entry_game} is misspecified.
To complete the model, we assume that, whenever there are multiple pure-strategy Nash equilibria, players will choose each equilibrium with the same probability. We assume that the support of $X$ is a bounded interval $\mathcal{X}$ in $\real^{\dim(X)}$. Without loss of generality, we assume $X$ is distributed uniformly in its support.

In the simulation, we focus on the simple case where $\gamma^1_{j_1,j_2} = \gamma^*$ for all $(j_1, j_2)$, and $\gamma^i_{j_1,j_2} = 0$ for all $(i\neq 1, j_1, j_2)$. Then, the joint distribution of $(Y, X)$ generated from this data generating process is indexed by $(\theta, \gamma^*)$, which we write as $F_{\theta, \gamma^*}$. Note that $\gamma^*$ measures the degree of misspecification: the larger the value of $\gamma^*$ is, the larger the degree of misspecification of model in \eqref{eq:mis_model_entry_game} is.
In the simulation design we impose that $m=3$, $\dim X=1$, $X\sim U[-1,1]$, $\beta=0.1$, $\alpha=[1,1,1]$, and  $\delta^{i}_j =1$ for all $(i,j)$. In the following, we construct outer sets both for the parameter $\delta$ and also for a counterfactual outcome in a counterfactual experiment.
In both cases, we will show that we will be able to generate three outer sets (including the CT outer set) that are discordant with each other  and  none of them contains the true value.

%
%


\subsubsection{Discordant nonsharp identification conditions for $\delta$}
Our first objective is to illustrate the existence of discordant nonsharp identification conditions. Here, we focus on the (projected) identified set for $\delta^1_2$. Given the DGP with $(\theta, \gamma^*)$, the outer set for $\delta^1_2$ associated with the nonsharp identification conditions in $A'$ is characterized as follows: $\Lambda(F_{\theta, \gamma^*}, A') = \{t: \exists (\alpha, \beta, \delta) \in \Theta_I(F_{\theta, \gamma^*}, A') \text{ s.t. }\delta^{1}_2 = t\}$. Now, let us denote by $\overline{\gamma}^*_{ct}$ the maximum degree of misspecification that could not be detected by the submodel $A_{ct}$. More precisely, we define $\overline{\gamma}^*_{ct} \coloneqq \sup\{\gamma^*: \Lambda(F_{\theta, \gamma^*}, A_{ct}) \neq \emptyset\}$, therefore, $\Lambda(F_{\theta, \gamma^*}, A_{ct})$ is a nonempty outer set when $\gamma^* \leq  \overline{\gamma}^*_{ct}$, and it becomes empty whenever $\gamma^* > \overline{\gamma}^*_{ct}$.
%
%
%
%
%
In Figure \ref{fig:delta_outer_set_ct0}, we plot in blue the CT outer set, i.e., $\Lambda(F_{\theta, \gamma^*}, A_{ct})$, at various degree of model misspecification, for $\gamma^* \in [0,\overline{\gamma}^*_{ct}]$, where $0$ corresponds to no misspecification and $\overline{\gamma}^*_{ct}$ corresponds to the maximum degree of misspecification that is not detectable with $A_{ct}$.


\begin{figure}[h!]
	\centering
	\includegraphics[width=\linewidth]{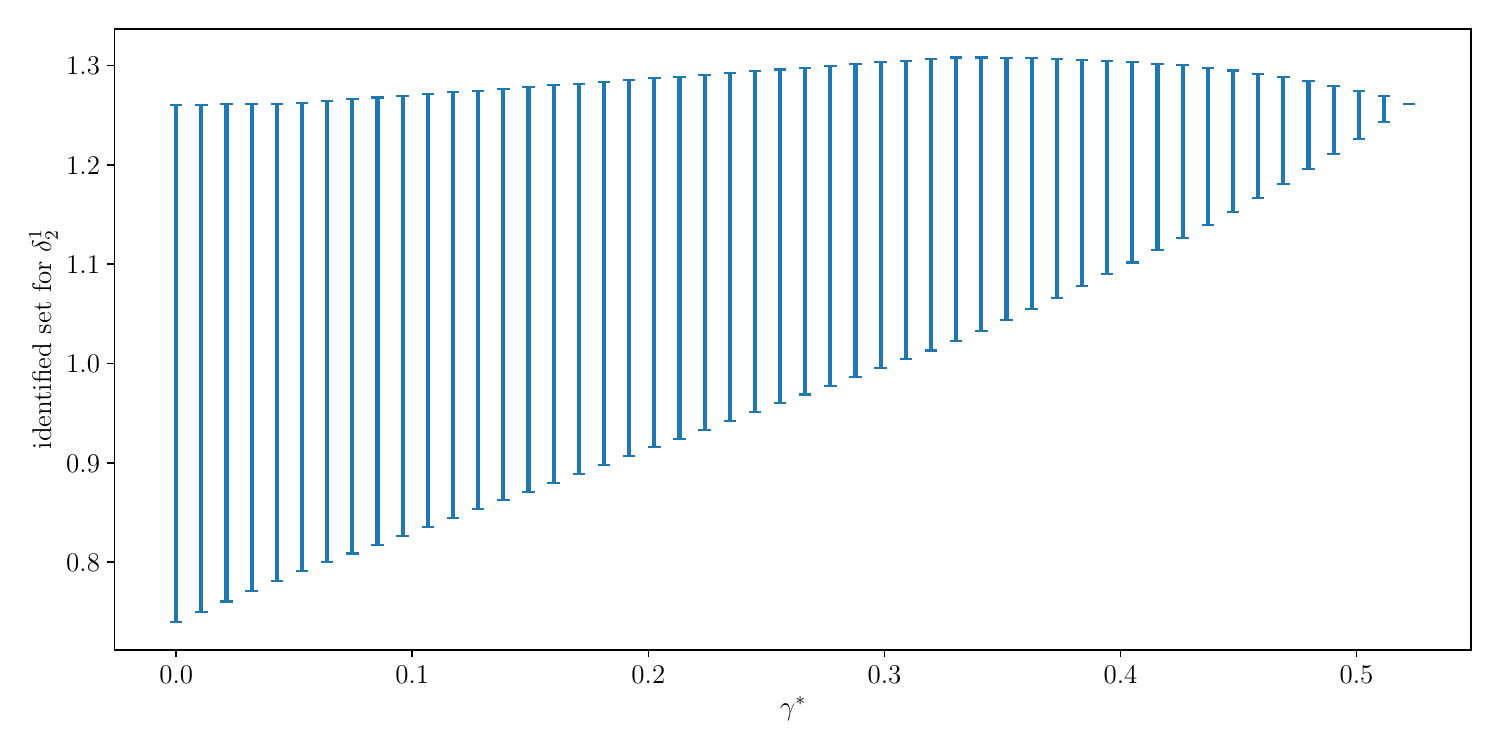}
	\caption{$\Lambda(F_{\theta, \gamma^*}, A_{ct})$ at various values of $\gamma^*$.}%
	\label{fig:delta_outer_set_ct0}
\end{figure}
A first remark is that the CT outer set shrinks when the degree of misspecification increases and at some point it no longer contains the true value. This illustrates that  the tightness of the outer set should not systematically be interpreted as a signal of an informative  identified set but it could just signal a presence of misspecification. 

In Appendix \ref{sec:random_set_capacity}, we explain why the findings of  Theorem \ref{thm:misleading} apply to the entry game example. So, according to Theorem \ref{thm:misleading}, there must exist some other $A'$ nonsharp identification condition which is discordant  with $A_{ct}$. Indeed, we are able to find two sets of nonsharp identification conditions, denoted as $A_1$ and $A_2$, which result in discordant identification results with $A_{ct}$. In  Figure \ref{fig:delta_outer_set_ct}, we plot $\Lambda(F_{\theta, \gamma^*}, A_1)$ in orange and $\Lambda(F_{\theta, \gamma^*}, A_2)$ in green. $\Lambda(F_{\theta, \gamma^*}, A_1)$ suggests values for $\delta^{1}_2$ that are higher than those suggested by  $\Lambda(F_{\theta, \gamma^*}, A_{ct})$ while $\Lambda(F_{\theta, \gamma^*}, A_2)$ suggests values that are lower.  When the degree of misspecification is higher than $0.4$, these three outer sets have no overlap.



\begin{figure}[h!]
	\centering
	\includegraphics[width=\linewidth]{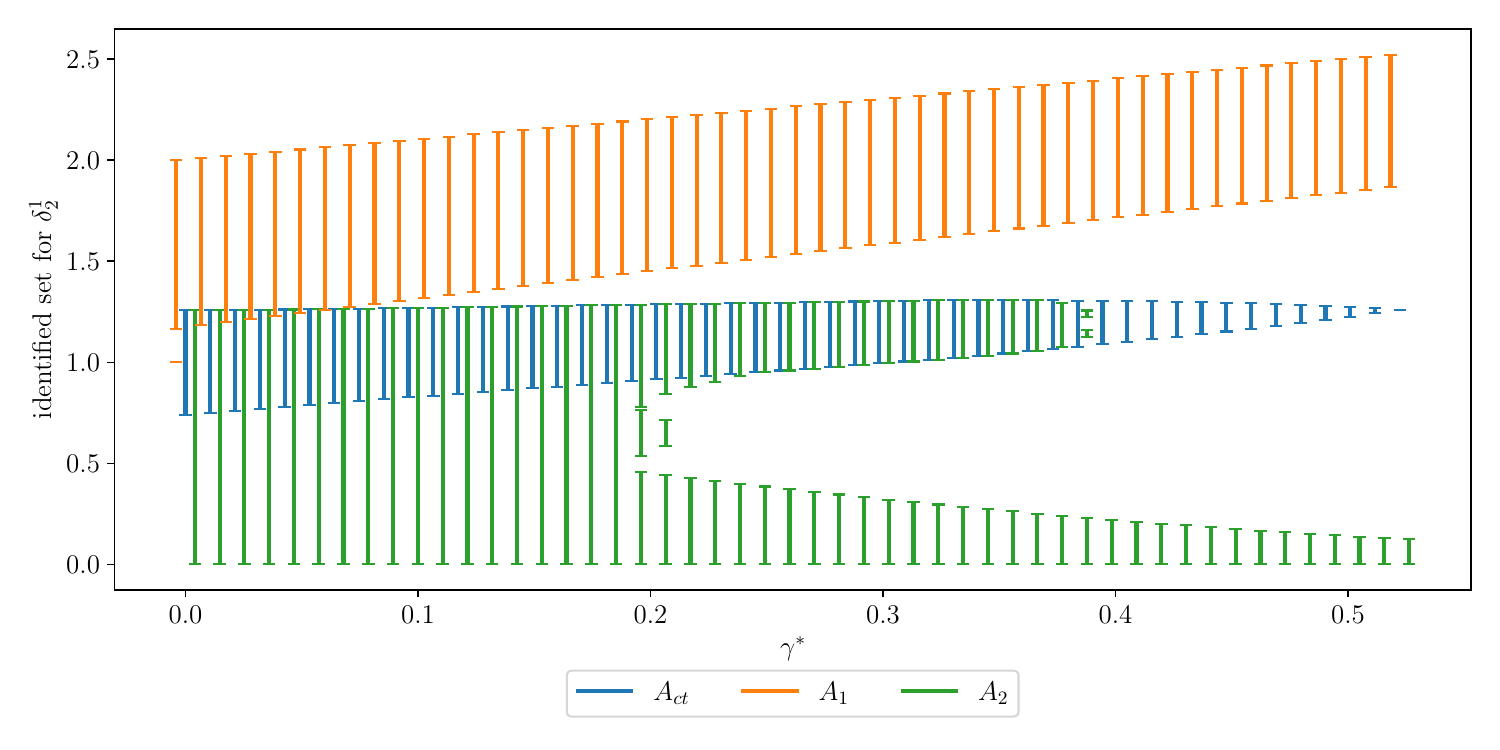}
	\caption{$\Lambda(F_{\theta, \gamma^*}, A_{1})$  and $\Lambda(F_{\theta, \gamma^*}, A_{2})$ at various values of  $\gamma^*$.}%
	\label{fig:delta_outer_set_ct}
\end{figure}



\subsubsection{Discordant counterfactual predictions}
In many of the empirical games applications, applied researchers are very often interested in implementing counterfactual analyses. Interestingly, we observe not only discordant results for the parameters, but we also observe this phenomenon for counterfactuals. 
Below, we illustrate a scenario where different outer sets lead to discordant counterfactual outcomes when the full model is misspecified. Therefore, we illustrate the fact that the discordancy issue also applies to counterfactual outcomes. 
Let us consider a counterfactual where firm $3$ is no longer a potential entrant of the market. This type of counterfactuals would arise, for example, when firm $3$ is a foreign firm and is banned from the home market due to a trade policy, or when firm $3$ is merged with other firms. 

Figure \ref{fig:counterfactual_outer_set_other} plots the different counterfactual predictions of submodels $A_{ct}$, $A_1$, and $A_2$ for the probability that only one firm enters the market, i.e., the probability of the presence of a monopoly in a market with characteristics $x_0 = (0, 0, 0)$.
As we can see clearly, three submodels give counterfactual predictions that are discordant with each other.  

\begin{figure}[h!]
	\centering
	\includegraphics[width=\linewidth]{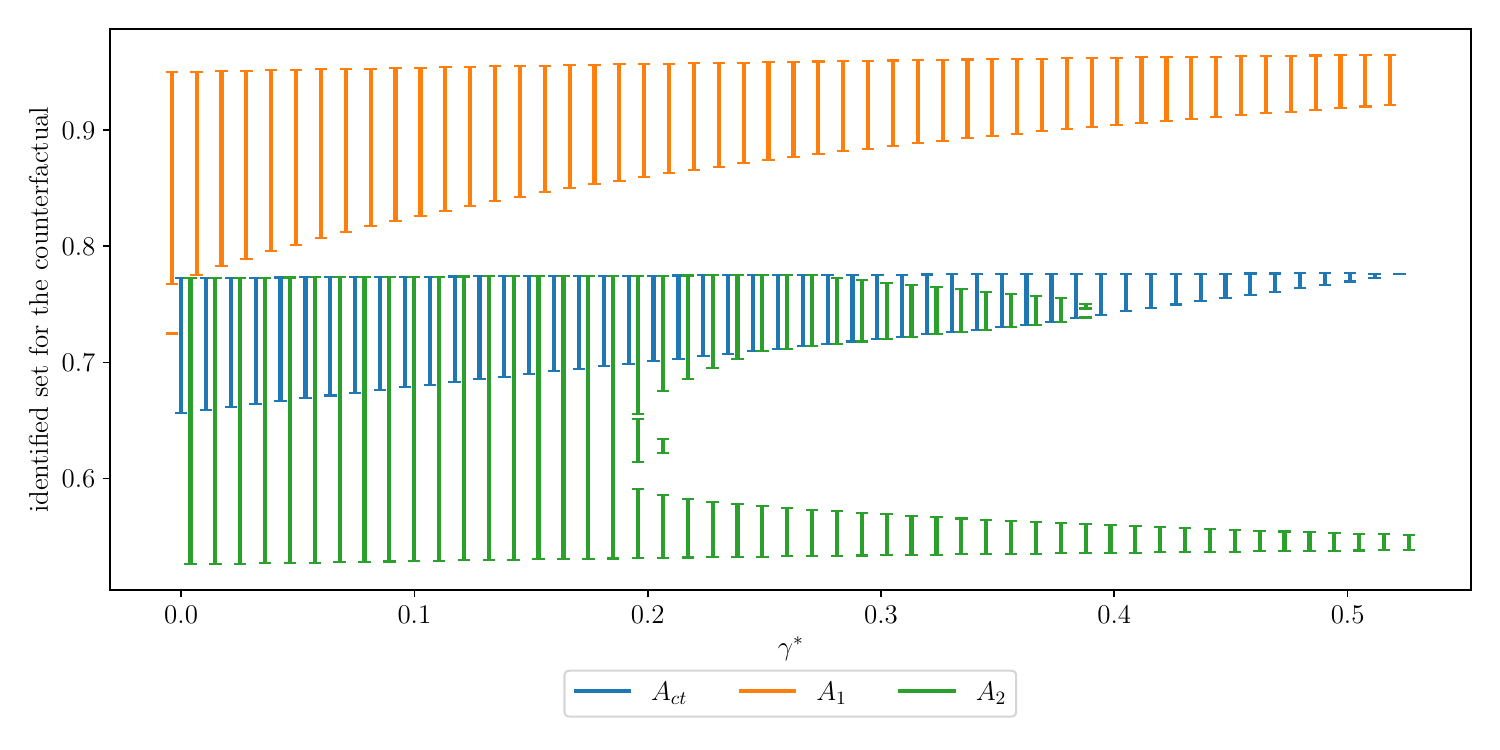}
	\caption{ $\mathbb{P}(Y_1 + Y_2 = 1 | X = x_0)$ at various $\gamma^*$ }%
	\label{fig:counterfactual_outer_set_other}
\end{figure}

\subsection{An empirical illustration for compatible submodels}\label{Emp}

\subsubsection{Context and Data} 
Estimating the causal impact of college education on later earnings has always been troublesome for economists because of the endogeneity of the level of education. To evaluate the returns to schooling, different approaches have been proposed, and most of them rely on the validity of instruments such as parental education, tuition fees, quarter of birth, distance to college, etc. The validity of all these IVs has been widely criticized because of their potential correlation with the children's unobserved skills.

In order to accommodate potentially invalid instruments, \cite{Manski2000, Manski2009} introduced the monotone IV (MIV) that does not require the IV to be valid but only imposes a positive dependence relationship between the IV and potential earnings. For instance, parental education may not be independent of potential wages, but plausibly does not negatively affect future earnings. In such a context, bounds on the average return to education can be derived.

In this application, we will consider the AMIV assumption introduced in Section \ref{sec:AMIV}. We consider that parental education can have a positive effect on children's future earnings, but this marginal positive effect could plausibly become null after some cut-off. The particularity of our method is to let this cut-off be determined by the data using our misspecification robust bounds.

We consider the data used in \citet[HTV]{Heckman2001}. The data consist of a sample of 1,230 white males taken from the National Longitudinal Survey of Youth of 1979 (NSLY79)\footnote{The NLSY79 survey is sponsored and directed by the U.S. Bureau of Labor Statistics, and managed by the Center for Human Resource Research (CHRR) at The Ohio State University. Interviews are conducted by the National Opinion Research Center (NORC) at the University of Chicago. See \cite{bureau_of_labor_statistics_national_2019} for more details.}. The data contain information on the log weekly wage, college education, father's education, mother's education, among many other variables. Following HTV, we consider the college enrollment indicator as the treatment: it is equal to 1 if the individual has completed at least 13 years of education and 0 otherwise. In this empirical exercise, we use the maximum of parental education as the candidate instrumental variable. Some summary statistics are reported in Table \ref{Table:statdes}.

\begin{table}[h] 
\caption{Summary Statistics}
  \begin{tabular}{l|cccccccc} \label{Table:statdes}\\
\\
 &  & Total  \\\hline\\
Observations &  & 1,230 \\ \\
log wage &  & 2.4138 (0.5937) \\
college &  & 0.4325 (0.4956) \\
father's education &  & 12.44715 (3.2638) \\
mother's education &  & 12.1781 (2.2781)\\
max(father's education, mother's education) &  & 13.1699 (2.7123)\\
\hline

\end{tabular}
\begin{center}
\footnotesize Average  and standard deviation (in the parentheses)
\end{center}
\end{table}
\subsubsection{Methodology and results}

We start by  constructing  the 95\% confidence region for the  identified sets of the average structural functions $E[Y_d]$, $d  \in \{0,1\}$ and the average treatment effect $E[Y_1-Y_0]$ under the \cite{Manski1990} mean independence assumption, denoted as $\Theta_I(MI)$, and under the MIV assumption, denoted as $\Theta_I(MIV)$. In addition, we construct an estimate of our misspecification robust bounds under the AMIV assumption, denoted as $\Theta^*_I(AMIV)$, using the following steps:

\begin{enumerate}
\item  \label{step1}

The support of our instrument is $\mathcal Z=\{0,1, \ldots, 20\}$. For each $z \in \{0,1, \ldots, 20\}$, we test the implications \eqref{Test1} and \eqref{Test2} using the intersection bounds method of \cite{CLR2013}, implemented in \cite{CKLRstata}. For each $d\in \{0,1\}$, we set  $z^*_d$ as the smallest $z$ for which we do not reject \eqref{Test1} and \eqref{Test2}. 
 The use of  \citeauthor{CKLRstata}'s (\citeyear{CKLRstata}) Stata package yields a 95\% confidence set for
$[\max(\underline{q}_{dt}: t \le z), \min(\overline{q}_{dt}: t \ge z)]$ for each $z<z^*_d$, and $[\max(\underline{q}_{dt}: t = 1,...,k ), \min(\overline{q}_{dt}: t \ge z^*_d)]$. 

\item \label{step2} 

We then plug the 95\% confidence bounds obtained from step (\ref{step1}) into the bounds in Equation \eqref{eq:idset_AMIV}, where we replace $\mathbb P(Z=z)$ by its sample analog. This procedure leads to an estimate of the identified set $\Gamma_{d, z^*}$ for each $d\in \{0,1\}$, which yields an estimate for the identified set $\Theta^*_I(AMIV)$.\footnote{Because our primarily focus in this paper is about identification, we do not attempt to study the statistical issues related to the derivation of a valid confidence region  for the misspecification robust bound. We leave this open question  for future research.}

\end{enumerate}
The same procedure is applied to get an estimate for $\Theta_I(MIV)$ except that $z^*_d$ is set to 20 for each $d \in \{0,1\}$.
Finally, since the identified set for the ATE   under the mean independence assumption denoted as $\Theta_I(MI)$ takes the form of standard intersection bounds, we use the \cite{CKLRstata} package to obtain its 95\% confidence bounds.

The results are summarized in Table \ref{Table:cs}. 
Column (1) shows that the 95\% confidence region for $\Theta_I(MI)$ is empty. 
In other words, the data shows clear evidence against the use of parental education as a valid IV. On the other hand,  column (4) shows the result for $\Theta_I(MIV)$.\footnote{We even test the validity of the MIV using the test proposed by \cite{Hsu2019}, we do not reject the MIV assumption even at 10 \% level.} As can be seen, we move from an empty identification region to a wide and non-informative identification region. In contrast, our misspecification robust bounds provide a nonempty yet relatively smaller set estimate for the ATE. 
Column (2) shows estimates of our misspecification robust bounds $\Theta_I^*(AMIV)$ when we allow the cut-offs to differ across potential outcomes as discussed in Remark \ref{remark:AMIV}, while column (3) shows estimates where the cut-offs are restricted to be the same for both potential outcomes as in Proposition \ref{prop:Monotone_Exclusion_ATE}. In the former case, we see that our proposed approach almost identifies the sign of the ATE.

\begin{table}[h] 
\caption{Results}
\begin{threeparttable}
  \begin{tabular}{l|cccccc} \label{Table:cs}
 &  & (1) &  (2)  &   (3)  & (4) \\ \\
Set estimates/   &  &  $\Theta_I(MI)$ &   $\Theta^*_I(AMIV)$  &   $\Theta^*_I(AMIV)$  & $\Theta_I(MIV)$ \\
 95\% Conf. Bounds                   &  	&  & $(z^*_1,z^*_0)=(0,11)$ & $(z^*_1,z^*_0)=(11,11)$  & \\\hline\\
$\theta_1\equiv\mathbb E[Y_1]$ &   &$[2.535,2.815]$  & $[2.535,2.815]$   &$[2.412,2.816]$   &$[0.933,2.815]$     \\
$\theta_0 \equiv \mathbb E[Y_0]$&     & Empty   & $[2.547,2.591]$ & $[2.547,2.591]$   &  $[2.548,2.814]$     \\
$ATE\equiv\mathbb E[Y_1-Y_0]$ &   & Empty  & $[-0.056,0.268]$ & $[-0.179,0.269]$  & $[-1.881,0.267]$      \\
\hline
\end{tabular}
\begin{tablenotes}
   \footnotesize
 \item[1] All values in column (1) are the $95\%$ confidence intervals.
 \item[2] All values in column (2)-(4) are set estimates based on the $95\%$ confidence interval of $\Theta_I(a_z)$.
\end{tablenotes}
\end{threeparttable}
\end{table}

\section{Discussion}

In this paper, we demonstrate the existence of discordant submodels in a wide range of models in the presence of model misspecification. This provides another reason why one should use the sharp characterization of the identified set whenever possible. The identified set not only exhausts all the identification restrictions in the model structure and assumptions but also is immune to the possible misleading conclusions of discordant submodels. Unlike an outer set, the identified set will be empty when the model is refuted by the data.

In empirical applications where a sharp characterization of the identified set is not tractable, our results suggest that empirical researchers should exercise caution when working with nonsharp identification conditions, especially when the bounds they obtain are very tight. For example, as a robustness check, one could construct the outer sets in different ways and check for any discordance between them.

Salvaging a refuted model is usually a challenging task, as it often involves some arbitrariness in how the model gets relaxed, and it could sometimes be computationally intractable. However, things get much easier when the minimum data-consistent relaxation is unique. In this case, it is apparent which assumptions are consistent with the data and which assumptions are not because all the data-consistent assumptions are compatible with each other (Theorems \ref{thm:compatible_submodel} and \ref{thm:compatible_submodel_advanced}). Moreover, the identified set of any data-consistent submodel can be viewed as a conservative bound for the misspecification robust bound in this case, making the computation a lot easier.

When the uniqueness of the minimum data-consistent relaxation is beyond reach, one can still choose to find the misspecification robust bound we proposed in this paper. It always leads to rationalizable and non-conflicting statements (Theorem \ref{thm:interpretation}), and it is sometimes the most informative rationalizable and non-conflicting statement (Theorem \ref{thm:smallest_if}). We work out the misspecification robust bound in some simple examples, but its exact solution could be too complicated to solve when the underlying model involves many structures. In those challenging cases, it might be possible to construct an outer set that always covers the misspecification robust bound proposed in this paper. This type of outer sets will be immune to the issue raised in this paper. It remains unclear how to construct such outer sets, but this could be one reasonable step beyond the findings in this paper.

\newpage

\bibliographystyle{jpe}
\bibliography{mybib}

\clearpage

\begin{appendix}

\section{Additional Results}\label{App:as}
This appendix collects some additional theoretical results. We put all the proofs in Appendix \ref{sec:additional_result_proof} except for very short ones.

\subsection{Example of intersection bounds}\label{intersec}
\begin{example}[Discrete treatment model]\label{ex:1}
  Consider a setting where $\mathcal{X} \equiv \{ x_1, ..., x_K \}$ is the set of all possible treatments.
  Let $Y_k$ be the potential outcome when the treatment is externally set to  $x_k$. The observed outcome $Y$ is defined as follows: $Y = \sum_k \indicator(X = x_k)Y_k$. Let us define $\theta_k\equiv E[Y_k]$ and assume that $Y_k$ has a bounded support $[\underline{y}_k, \overline{y}_k]$. The random bound for
  $Y_k$ can be constructed as follows $\underline{Y}_k \equiv Y\indicator(X = x_k)+ \underline{y}_k\indicator(X\neq x_k)$ and
  $\overline{Y}_k \equiv Y\indicator(X=x_k)+ \overline{y}_k\indicator(X\neq x_k)$. If we assume  the mean independence assumption $E[Y_k|Z] = E[Y_k]$ we obtain a special case of  \eqref{eq:intersection_bounds}.
\end{example}
Discrete treatment models with bounded potential outcomes are  usually considered in Manski's work. See for instance \cite{Manski1990, Manski1994} among many others.

\begin{example}[Smooth Treatment Model] \label{ex:2}
	Consider a smooth treatment model as in \cite{Kim2018}. When the treatment is $x$, the potential
  outcome is $Y(x) = g(x, \epsilon)$ where $g$ is an unknown function, and $\epsilon$ is individual
  heterogeneous characterization. Assume $g(x,\epsilon)$ is Lipschitz continuous in $x$ with Lipschitz
  constant equal to $\tau$.  Suppose we are interested in $\theta_x = E[Y(x)]$. The lower and upper bounds can be constructed as
  $\underline{Y}(x) = Y - \norm{X - x}\tau$ and $\overline{Y}(x) = Y + \norm{X - x}\tau$. As in the discrete treatment case,  if we assume  $E[Y(x)|Z] = E[Y(x)]$, we obtain model \eqref{eq:intersection_bounds}.  As a special case, one can also consider a linear model with heterogeneous coefficient, $Y = X'\beta + \epsilon$ where
  $\beta$ is a vector of an unobserved random coefficient. Suppose the coefficient space for $\beta$ is
  $[\underline{\beta}, \overline{\beta}]$. 
    Then, $\underline{Y}(x) = Y + \sum_{i}\min\left\{(x_{i} - X_{i})\underline{\beta}_{i}, (x_{i}
    - X_{i})\overline{\beta}_{i}\right\}$ where the subscript $i$ stands for the $i^{th}$ dimension of the corresponding
    variables. Similarly, $\overline{Y}(x) = Y + \sum_{i}\max\left\{(x_{i} - X_{i})\underline{\beta}_{i}, (x_{i}
    - X_{i})\overline{\beta}_{i}\right\}$.
\end{example}

 \subsection{Conditional Moment Inequalities}\label{sec:cond_moment_ineq}
Let us now consider a more general setting than the introductory example. Assume the full model is a conditional moment
inequality, 
\begin{equation}\label{eq:cond_moment_ineq}
  E[m(X;\theta)|Z] \le 0 \text{ almost surely} 
\end{equation}
where $X \in \real^{k_1}$ and $Z\in \real^{k_2}$ are observable random variables and $m(\cdot, \cdot\ ;\theta)$ is some known integrable function with $E\norm{m(X;\theta)} < \infty$ for each $\theta$. We focus on the case where $Z$ are continuous random variables. Random variables $X$ and $Z$ could have overlaps. In practice, empirical researchers sometimes use the following unconditional
model instead:
\begin{equation}\label{eq:uncond_moment_ineq}
  E[w(Z) m(X;\theta)] \le 0,
\end{equation}
where $w(\cdot)$ is some nonnegative weighting function. We want to understand what would happen when one conduct
empirical analysis based on \eqref{eq:uncond_moment_ineq} when \eqref{eq:cond_moment_ineq} happens to be refuted.

To answer this question, define $\mathcal{W}^+_m$ to be the set of all $m$-dimenstional nonnegative function $w$ which satisfies $0 < E\norm{w(Z)}^2 < \infty$ and $E \norm{w(Z) m(X;\theta) } < \infty$ for all $\theta\in \Theta$. Define $A$ as the collection of condition \eqref{eq:uncond_moment_ineq} for all $w\in \mathcal{W}^+_1$, i.e.
\begin{equation*}
A \coloneqq \{ \eqref{eq:uncond_moment_ineq} \text{ with }w: w \in \mathcal{W}^+_1\}
\end{equation*}
With this definition, any subset $B$ of $A$ with $m$ elements corresponds to the condition which \eqref{eq:uncond_moment_ineq} hold for some $w\in \mathcal{W}^+_m$. By the construction of $A$, Assumption \ref{assu:mis_2} are satisfied.

To verify Assumption \ref{assu:mis_1}, we need to construct a $\mathscr{C}$. Let $\mathcal{Z}$ be the support of $Z$. For any $z\in \mathcal{Z}$ and any $\epsilon > 0$, define function $h_{z,\epsilon}$ as $h_{z,\epsilon}(Z) = \indicator(\norm{Z - z} < \epsilon)$. Suppose that, for any $z \in \mathcal{Z}$, there exists some $\theta\in \Theta$ and some $\delta(z) > 0$ such that $E[m(X;\theta)|Z] \le 0$ for almost every $Z$ with $\norm{Z - z}\le \delta(z)$. Then, for each $z\in \mathcal{Z}$ and each $\epsilon\in (0, \delta(z))$, there exists some $\theta\in \Theta$ such that $E[h_{z,\epsilon}(Z)m(X,\theta)] \le 0$. Define the collection of functions $\mathcal{W}^*$ as $\mathcal{W}^* \coloneqq \{h_{z,\epsilon}: \epsilon\in (0, \delta(z)), z\in \mathcal{Z}\}$. Then, we can construct $\mathscr{C}$ as
\begin{equation}\label{eq:mathscrC_in_moment_ineq}
\mathscr{C} \coloneqq \{\{a\}: a\in A^*\},\text{ where } A^* \coloneqq \{\eqref{eq:uncond_moment_ineq} \text{ with }w: w \in \mathcal{W}^*\}.
\end{equation}
The following proposition shows that this $\mathscr{C}$ satisfies Assumption \ref{assu:mis_1} under some regularity conditions.
\begin{proposition}\label{prop:moment_ineq_misleading}
Assume that
  \begin{enumerate}[label={(\alph*)}]
    \item there exists a function $g(z;\theta)$ such that (\emph{i}) for every $\theta\in \Theta$, $E[m(X,Z;\theta)|Z] = g(Z;\theta)$ almost surely; (\emph{ii}) $g(z;\theta)$ is continuous in $z$ for any given $\theta$; (\emph{iii}) $g(z;\theta)$ is continuous in $\theta$ for any given $z$ 
    \item for any $z$ in the support of $Z$, there exists some $\delta(z) > 0$ and some $\theta\in \Theta$ such that $E[m(X;\theta)|Z] \le 0$ for almost every $Z$ satisfying $\norm{Z - z} \le \delta(z)$. 
    \item there exists some function $\gamma(\cdot)$ such that $ \sup_{\theta\in \Theta} \norm{E[m(X;\theta)|Z]} \le \gamma(Z)$ almost surely and $E|\gamma(Z)|^2 < \infty$.
    \item $\Theta$ is compact.
  \end{enumerate}
Then, Assumption \ref{assu:mis_1} is satisfied for $\mathscr{C}$ constructed in \eqref{eq:mathscrC_in_moment_ineq}. 
\end{proposition}
As a result, Theorem \ref{thm:misleading} can be applied here. In the context of moment inequalities, the result in Theorem \ref{thm:misleading} means that \eqref{eq:cond_moment_ineq} is refuted if and only if there exists $w_1 \in \mathcal{W}^+_{m_1}$ and $w_2 \in \mathcal{W}^+_{m_2}$ such that both \eqref{eq:uncond_moment_ineq} with $w=w_1$ and \eqref{eq:uncond_moment_ineq} with $w=w_2$ are not refuted but the identified sets of these two sets of identification conditions have empty intersection. Moreover, whenever \eqref{eq:cond_moment_ineq} is refuted, for any $\tilde{w}\in \mathcal{W}^+_{\tilde{m}}$ with which \eqref{eq:uncond_moment_ineq} is data-consistent, there exists some $w_1\in \mathcal{W}^+_{m_1}$ and $w_2\in \mathcal{W}^+_{m_2}$ such that both the \eqref{eq:uncond_moment_ineq} with $w=(\tilde{w}, w_1)$ and \eqref{eq:uncond_moment_ineq} with $w=w_2$ are data-consistent but their identified sets have empty intersection. 

This result  complements the findings in \cite{AS2013}. In \cite{AS2013}, they propose an inference procedure for models like \eqref{eq:cond_moment_ineq}. Their inference transform \eqref{eq:cond_moment_ineq} into \eqref{eq:uncond_moment_ineq} by selecting $w$ in a sub-family of $\mathcal{W}^+_m$ and letting $m\to\infty$ as the sample size increases. Our result shows that increasing $m$ to infinity is crutial to ensure the robustness of the result if \eqref{eq:cond_moment_ineq} could be misspecified. If the dimension of $w$ is fixed, then the empirical result for \eqref{eq:uncond_moment_ineq} could be misleading even if the inference controls the size uniformly.

\subsection{Random Sets and Choquet Capacity}\label{sec:random_set_capacity}
In this section, we consider models whose identified set can be described with random sets and choquet capacity functions. Let $Y$ be a vector of endogenous random variables, and let $X$ be a vector of exogenous observable covariates. Let $\mathcal{Y}$ and $\mathcal{X}$ denote the support of $Y$ and $X$ respectively. Here, the parameter $\theta$ of interest could be of infinite dimensions, and its parameter space $\Theta$ need not be compact.

Let $\Gamma(\theta)$ be some random closed set in $\mathcal{Y}$ which could depend on $\theta$, $X$ and some latent random variables. Assume $P(Y\in \Gamma(\theta)) = 1$. \cite{Artstein1983} shows that the conditional distribution of $Y$ given $X$ equals $F_{Y|X}$ almost surely if and only if for any compact subset $K$ of $\mathcal{Y}$, the following inequality holds:
\begin{equation}\label{eq:artstein_ineq}
  P_{F}(Y\in K |X) \le L(K, X;\theta)\text{ almost surely},\text{ where }L(K, X;\theta)\coloneqq P(\Gamma(\theta)\cap K \neq \emptyset|X)
\end{equation}
where $P_F$ refers to the probability measure corresponding to $F_{Y|X}$. The $L(\cdot, X;\theta)$ is often known as the \emph{Choquet capacity function}. This type of models plays an important role in the partial identification literature. We refer to \cite{Molinari2020} for more background introductions. Often in practice, either $P_F(Y\in K|X)$ or $L(K, X;\theta)$ can be identified from the data, and the other one can typically be derived or simulated from some additional assumptions. For the purpose of illustration, we consider the case where $Y$ and $X$ are observable so that $P_F(Y\in K|X)$ can be identified from the data,  and assume $L(K, X;\theta)$ is a known function of $K$ and $X$ given $\theta$. 

In general, one needs to check \eqref{eq:artstein_ineq} for all compact sets of $\mathcal{Y}$ in order to ensure this collection of moment inequalities is a sharp identification condition. In some circumstances, checking the inequalities for all compact sets is equivalent to checking the inequalities only for a subcollection of compact sets,  in which case, this subcollection is called the \textit{core determining class} in the language of \cite{henry2011}. However, in practice, researchers often pre-select some finite collection $\mathcal{K}$ of compact sets that are not core-determining, and they only check \eqref{eq:artstein_ineq} for compact sets in this $\mathcal{K}$. For instance,  in the treatment effect literature, the well-known \cite{Manski1994} bounds on the potential outcome distributions implemented in various applications such as in \cite{Blundell2007}, or  \cite{Peterson1976} bounds on competing risk, use only a finite and not sufficient  collection  of Artstein inequalities. See, respectively,  \cite{Molinari2020}, and  \cite{Mourifie2020} for a detailed discussion. In empirical games, auction and  network applications we can also cite \cite{Ciliberto2009},  \cite{Haile2003},  \cite{Sheng2020},  \cite{Chesher2020}, among many others who also  focused on a finite and not sufficient  collection of Artstein inequalities.\footnote{See \cite{Molinari2020} for a detailed discussion.}

We want to explore the consequences of this pre-selection procedure when the original model might be refuted by the data. For simplicity, we focus on the case where $Y$ only takes a finite number of possible values. In this case, the support $\mathcal{Y}$ of $Y$ is a finite set and the collection of all compact sets in $\mathcal{Y}$ is simply the power set of $\mathcal{Y}$, i.e. the collection of all subsets of $\mathcal{Y}$. 

To fit this model into the general framework in Section \ref{sec:misleading_general}, define $A$ as the collection of all Artsein's inequality, i.e.
\begin{equation*}
A \coloneqq \{\eqref{eq:artstein_ineq} \text{ with }K: K \subseteq \mathcal{Y}\}.
\end{equation*}
With this definition, any subset $A'$ of $A$ corresponds to testing \eqref{eq:artstein_ineq} only for a pre-selected collection of compact sets. In order to apply Theorem \ref{thm:misleading}, we need to verify Assumptions \ref{assu:mis_2} and \ref{assu:mis_1}.  By the construction of $A$, Assumption \ref{assu:mis_2} is satisfied. To verify Assumption \ref{assu:mis_1}, let us construct $\mathscr{C}$ as 
\begin{equation}\label{eq:mathscrC_in_random_set}
\mathscr{C} \coloneqq \{\{a\}: a\in A\}.
\end{equation}
The following proposition provides a sufficient condition under which Assumption \ref{assu:mis_1} is satisfied with this choice of $\mathscr{C}$.

\begin{proposition}\label{prop:capcity_ineq_prop}
Let $\mathcal{X}$ and $\mathcal{Y}$ be the supports of $X$ and $Y$, respectively. Suppose $\mathcal{Y}$ is a finite set. The parameter space $\Theta$ may or may not be compact. Suppose that, for each $y\in \mathcal{Y}$, the following assumptions hold:
\begin{enumerate}[label={(L\ref{prop:capcity_ineq_prop}.C\arabic*)}]
\item\label{enu:finite_y_1}$\inf_{x\in \mathcal{X}} P(Y = y|X=x) > 0$,
\item\label{enu:finite_y_2}there exists a sequence $\theta_1, \theta_2,...$ in $\Theta$ such
that $ \inf_{x\in \mathcal{X}} L(\{y\}, x;\theta_k)\to 1$ as $k\to\infty$, 
\end{enumerate}
where the $\inf$ in the above two conditions refers to the essential infimum. Then, Assumption \ref{assu:mis_1} holds for $\mathscr{C}$ defined in \eqref{eq:mathscrC_in_random_set}.
\end{proposition}

As a result, Theorem \ref{thm:misleading} could be applied here. For any pre-selected collection $\mathcal{K}$ of compact subsets, define $\Theta_I(\mathcal{K})$ as the set of parameters which satisfy \eqref{eq:artstein_ineq} for all $K\in \mathcal{K}$. In this context, the result in Theorem \ref{thm:misleading} means that the model is refuted if and only if there exists two $\mathcal{K}_1$ and $\mathcal{K}_2$ such that $\Theta_I(\mathcal{K}_1)\neq \emptyset$, $\Theta_I(\mathcal{K}_2)\neq \emptyset$ and $\Theta_I(\mathcal{K}_1)\cap\Theta_I(\mathcal{K}_2)= \emptyset$. Moreover, for any pre-selected collection $\mathcal{K}$ of compact sets with $\Theta_I(\mathcal{K})\neq \emptyset$, there always exist two finite collections $\mathcal{K}_1$ and $\mathcal{K}_2$ such that $\Theta_I(\mathcal{K}\cup \mathcal{K}_1)\neq \emptyset$, $\Theta_I(\mathcal{K}_2)\neq \emptyset$ and $\Theta_I(\mathcal{K}\cup \mathcal{K}_1)\cap\Theta_I(\mathcal{K}_2) = \emptyset$.

We conclude this subsection with the entry game model as an example. We are going to verify all the conditions in Proposition \ref{prop:capcity_ineq_prop} for this example. 

\begin{example}[Entry game]\label{ex:entry_game}
  Consider an $m$-player complete information entry game as in Ciliberto and
  Tamer (2009). Assume there are $m$ players, where player $i$'s payoff function is specified as
  \begin{equation*}
    \pi_i = Y_i\left(\gamma_i +  X_i'\beta_i - \sum_{j\neq i}\delta^{i}_j Y_j + \epsilon_i\right)
  \end{equation*}
  where the $X_i$s are some covariates which might be player $i$ specific, $Y_i\in \{0, 1\}$ stands for player $i$'s entry decision, and $Y_j$ stands for the decision of player $j$. Here, $\gamma_i$ and $\beta_i$ are player-specific parameter coefficient, and $\delta^i_j > 0$ is the parameter that describes the strategic interaction between player $i$ and $j$. We assume that $Y = (Y_i: i=1,...,m)$ is always a pure-strategy Nash equilibirum. 

  Assume $\epsilon = (\epsilon_1,..., \epsilon_m)$ is independent of $X$ and $\epsilon$ follows the normal distribution $N(0, \Sigma)$. Let $\gamma=(\gamma_1, ..., \gamma_m)$, $\beta = (\beta_1,...,\beta_m)$, and $\delta = (\delta^{i}_j: i\neq j)$. Let $\theta = (\gamma, \beta, \delta, \Sigma)$ be the vector of all parameters.  Let $\mathcal{Y} = \{y=(y_1,...,y_m): y_i\in \{0, 1\},i=1,...,m\}$ be the set of all possible
  entry decisions. For any $K\subseteq \mathcal{Y}$, define $L(K, X,\theta)$ to be the probability that at least one $y\in K$ is a pure-strategy Nash
  equilibrium given $X$ and $\theta$. In practice, $L(K,X, \theta)$ can often be solved from numerical simulations.
  
  In \cite{henry2011}, the identified set of this model is shown to be the set of all $\theta$ which satisfies \eqref{eq:artstein_ineq} for every subset $K$ of $\mathcal{Y}$. The number of these inequalities increases with $m$ very quickly in the order of $2^{2^m}$. \cite{henry2011} provide some methods to reduce the number of inequalities by removing redundant inequalities in \eqref{eq:artstein_ineq}, but, in general, sharp characterization of the identified set involves a large number of inequalities. In practice for the sake of computational feasibility, emprical researchers often pre-select a finite collection $\mathcal{K}$ of subsets and only check \eqref{eq:artstein_ineq} for each $K\in \mathcal{K}$. See, for example, \cite{Ciliberto2009}, and \cite{Ciliberto2020}. 

  Let us now check conditions in Proposition \ref{prop:capcity_ineq_prop}. Let $\Theta = \mathbb{R}^{\dim(\theta)}$. Condition \ref{enu:finite_y_1} in Proposition \ref{prop:capcity_ineq_prop} is a low-level condition that can be directly verified by the data. In theory, this condition would hold, for example, if the true data generating process has the following properties: (\emph{i}) the support of $\epsilon$ is $\real^m$ conditinal on almost every $X$, and (\emph{ii}) for each $i$, player $i$'s payoff function is 
  \begin{equation*}
    \pi_i = Y_i\left(g(X_i, Y_{-i}) + \epsilon_i\right)
  \end{equation*}
where $Y_{-i} = \{Y_j: j\neq i)$ and $g$ can be an arbitrary function of $(X_i, Y_{-i})$ that is bounded in their support. This class of data generating processes nests the model that we imposed above, but the true data generating process need not be the same as our model. That is, Condition \ref{enu:finite_y_1} in Proposition \ref{prop:capcity_ineq_prop} holds even if the model is misspecified. Condition \ref{enu:finite_y_2} in Proposition \ref{prop:capcity_ineq_prop} also holds, because for each $y \in \mathcal Y$, one can have  $L(\{y\}, x;\theta_k)\to 1$ by simply fixing $\beta = 0$, $\delta = 0$ and let $\gamma\to \gamma^*$ where $\gamma^*_i = \infty$ if $y_i = 1$ and $\gamma^*_i = -\infty$ if $y_i = 0$.
\end{example}

\subsection{Existence and Uniqueness of Minimum Data-consitent relaxation}\label{sec:existence_of_minimum_relaxation}

\begin{theorem}\label{thm:existence_MDR}
  Suppose one of the following two conditions is satisfied,
  \begin{enumerate}[label={(T\ref{thm:existence_MDR}.C\arabic*)}]
    \item\label{enu:T3_1} $A$ is a finite set.
    \item\label{enu:T3_2} For any $a\in A$, $\Theta_I(a)$ is compact. Moreover, for any $B\subseteq A$, $\Theta_I(B)=\cap_{a\in B}\Theta_I(a)$.
  \end{enumerate}
  Then, there exists some minimum data-consistent relaxation of $A$. Moreover, for any data-consistent $A'\subseteq A$, there exists some minimum data-consistent relaxation $\widetilde{A}$ such that $A'\subseteq \widetilde{A}$.
\end{theorem}

Theorem \ref{thm:existence_MDR} not only establishes the existence of a minimum data-consistent relaxation, but also shows that any data-consistent subset $A' \subseteq A$ can be further strengthened into a minimum data-consistent relaxation by including additional assumptions. It is worth noting that, when $A$ is a finite set, the result of Theorem \ref{thm:existence_MDR} does not require any additional conditions. When $A$ is an infinite set, we need $\Theta_I(a)$ to be compact for each $a\in A$. In addition, we need that, for any $B\subseteq A$, $\Theta_I(B)=\cap_{a\in B}\Theta_I(a)$, which would hold, for example, if $\theta$ fully describes the distribution of both observed and latent random variables as in a maximum likelihood setting.

The following theorem studies the uniqueness of the minimum data-consistent relaxation, and is a generalized version of Theorem \ref{thm:compatible_submodel} in the main text. 
\begin{theorem}\label{thm:compatible_submodel_advanced}
Statement \ref{enu:ind_assumptions} implies \ref{enu:unique_MDR}. If either \ref{enu:T3_1} or \ref{enu:T3_2} holds, then \ref{enu:unique_MDR} implies \ref{enu:ind_assumptions}.
\end{theorem}

\subsection{Discordancy and multiplicity of minimum data-consistent relaxations}\label{DMDCR}
\begin{proposition}\label{prop:DMDCR}
Whenever Condition \ref{enu:T3_1}  or \ref{enu:T3_2} hold we have the following result: There exists two data-consistent $A_1, A_2 \subseteq A$ with $\Theta_I(A_1)\cap\Theta_I(A_2)=\emptyset$ if and only if there exists two minimum data-consistent relaxations $\widetilde A_1$ and $\widetilde A_2$ such that  $\Theta_I(\widetilde A_1)\cap\Theta_I(\widetilde A_2)=\emptyset$.
\end{proposition}
\begin{proof}[Proof for Proposition \ref{prop:DMDCR}]
We first prove the {\bf if} part. Suppose there exists two minimum data-consistent relaxations $\widetilde A_1$ and $\widetilde A_2$ such that  $\Theta_I(\widetilde A_1)\cap\Theta_I(\widetilde A_2)=\emptyset$. By the definition of minimum data-consistent relaxation, both $\widetilde{A}_1$ and $\widetilde{A}_2$ are data-consistent subsets of $A$. Thus, this proves the existence of two data-consistent $A_1, A_2 \subseteq A$ with $\Theta_I(A_1)\cap\Theta_I(A_2)=\emptyset$. 

Next, we want to prove the {\bf only if} part. Suppose that there exists two data-consistent $A_1, A_2 \subseteq A$ with $\Theta_I(A_1)\cap\Theta_I(A_2)=\emptyset$. By Theorem \ref{thm:existence_MDR}, there exists two minimum data-consistent relaxations $\widetilde{A}_1$ and $\widetilde{A}_2$ such that $A_1 \subseteq \widetilde{A}_1$ and $A_2 \subseteq \widetilde{A}_2$. Because $\widetilde{A}_1 \subseteq A_1$, we have $\Theta_I(\widetilde{A}_1) \subseteq \Theta_I(A_1)$. Similarly, we have $\Theta_I(\widetilde{A}_2) \subseteq \Theta_I(A_2)$. Because $\Theta_I(A_1) \cap \Theta_I(A_2) = \emptyset$, we must have $\Theta_I(\widetilde{A}_1) \cap \Theta_I(\widetilde{A}_2)  = \emptyset$.
\end{proof}

\subsection{Empirical Interpretation of Misspecification Robust Bound}\label{nc}
As discussed in Section \ref{sec:nonconflicting_statement}, a rationalizable and nonconflicting set has rich interpretations. The following theorem shows that the misspecification robust bound $\Theta^*_I$ is both rationalizable and nonconflicting. 

\begin{theorem}\label{thm:interpretation}
    Suppose either \ref{enu:T3_1} or \ref{enu:T3_2} holds. Then, $\Theta^*_I$ is both rationalizable and nonconflicting. That is, 
 \begin{itemize}
    \item(rationalizable) there exists some submodel $A'\subseteq A$ such that $\Theta_I( {A}')\subseteq \Theta^*_I$ and $\Theta_I({A}')\neq \emptyset$. 
    \item(nonconflicting) there does not exist a submodel $A'\subseteq A$ with $\Theta_I(A')\neq \emptyset$ such that $\Theta_I(A')\cap \Theta^*_I = \emptyset$.
  \end{itemize}
\end{theorem}

As discussed in Section \ref{sec:nonconflicting_statement}, $\Theta^*_I$ have even richer explanations when it is the smallest rationalizable and nonconflicting set. Recall that a set $S^*$ is the smallest rationalizable and nonconflicting set, if $S^*$ is rationalizable and nonconflicting, and $S^*\subseteq S$ for every rationalizable and nonconflicting set $S$. The smallest rationalizable and nonconflicting set does not always exists. However, the following theorem shows that, under mild conditions, whenever the smallest rationalizable and nonconflicting set exists, it is equal to $\Theta^*_I$.

\begin{theorem}\label{thm:smallest_iff_simple}
Suppose either \ref{enu:T3_1} or \ref{enu:T3_2} holds.  Assume 
    \begin{enumerate}[label=(T\ref{thm:smallest_iff_simple}.C\arabic*)]
      \item \label{enu:no_nested_MDR} there does not exist two different minimum data-consistent relaxation $\tilde{A}_1$ and $\tilde{A}_2$ such that $\Theta_I(\tilde{A}_1)\subsetneq \Theta_I(\tilde{A}_2)$.
    \end{enumerate}
    Then, whenever the smallest rationalizable and nonconflicting set exists, it is equal to $\Theta^*_I$.
\end{theorem}

Condition \ref{enu:no_nested_MDR} could be verified from the data. Note that, for any two different minimum data-consistent relaxations $\tilde{A}_1$ and $\tilde{A}_2$, we always have $\Theta_I(\tilde{A}_1\cup \tilde{A}_2) = \emptyset$, i.e. $\tilde{A}_1$ and $\tilde{A}_2$ are not compatible with each other. As it is unlikely that for two sets $\tilde{A}_1$ and $\tilde{A}_2$ to satisfy $\Theta_I(\tilde{A}_1)\subsetneq \Theta_I(\tilde{A}_2)$ while being incompatible with each other, we consider \ref{enu:no_nested_MDR} as a mild condition. Finally, Condition \ref{enu:no_nested_MDR} would hold, if there is a unique minimum data-consistent relaxation, or if the identified set of every minimum data-consistent relaxation is a singleton set. In fact, these two conditions are also sufficient conditions for $\Theta^*_I$ being the smallest rationalizable and nonconflicting set, as shown in the following theorem.

\begin{theorem}\label{thm:smallest_if}
Suppose either \ref{enu:T3_1} or \ref{enu:T3_2} holds. Assume one of the following cnoditions is satisfied:
        \begin{enumerate}[label=(T\ref{thm:smallest_if}.C\arabic*), itemindent=16pt, start=1]
          \item\label{enu:S6_3_1} there exists a unique minimum data-consistent relaxation,
          \item\label{enu:S6_3_2} for any minimum data-consistent relaxation $\widetilde{A}$, $\Theta_I(\widetilde{A})$ is a singleton.
        \end{enumerate}
Then, $\Theta^*_I$ is the smallest set that are both rationalizable and nonconflicting.
\end{theorem}

\subsection{Discrete Relaxation versus Continuous Relaxation}\label{Disc}
%

\begin{theorem}\label{thm:compared_to_perturbation_corrected}
  Suppose \ref{enu:S6_3_2} holds. Then, $\Theta^*_I\subseteq \Theta^\dagger_I$ for any type of relaxation chosen by the researcher. 
\end{theorem}

%
%

%

\section{Proof of the main results}\label{App:proof}
\subsection{Proof of Theorem \ref{thm:misleading}}

First of all, note that Assumption \ref{assu:mis_2} implies that for any $A', A''\subseteq A$, $\Theta_I(A'\cup A'') = \Theta_I(A') \cap \Theta_I(A'')$.

If there exists $A', A'' \subseteq A$ such that $\Theta_I(A') \neq \emptyset$, $\Theta_I(A'')\neq \emptyset$ and $\Theta_I(A')\cap \Theta_I(A'') = \emptyset$, then, $\Theta_I(A) \subseteq \Theta_I(A'\cup A'') = \Theta_I(A')\cap \Theta_I(A'') = \emptyset$. Hence, $\Theta_I(A) = \emptyset$ if there exists $A', A'' \subseteq A$ such that $\Theta_I(A') \neq \emptyset$, $\Theta_I(A'')\neq \emptyset$ and $\Theta_I(A')\cap \Theta_I(A'') = \emptyset$.

Reversely, if $\Theta_I(A) = \emptyset$, we want to show that there exists two finite subsets $A', A'' \subseteq A$ such that $\Theta_I(A') \neq \emptyset$, $\Theta_I(A'')\neq \emptyset$ and $\Theta_I(A')\cap \Theta_I(A'') = \emptyset$. More specifically, we are going to show the following statement:

\begin{equation}\label{eq:statement1_in_proof}
\begin{array}{c}
\textnormal{when }\Theta_I(A)=\emptyset, \textnormal{ there exists }A'\in \mathscr{C}\textnormal{ and } \{A_1,...,A_n\}\subseteq \mathscr{C}\textnormal{ for some finite }n \\
\textnormal{ such that both }A' \text{ and } A'' = \cup_{i=1}^n A_i\textnormal{ are data-consistent, but }\Theta_I(A')\cap \Theta_I(A'') = \emptyset.
\end{array}
\end{equation}
To show \eqref{eq:statement1_in_proof}, we consider two cases.

\noindent \emph{\bf Case 1: the $\mathscr{C}$ in Assumption \ref{assu:mis_1} has infinite elements}.
In this case, $\Theta_I(A')$ is compact for all $A'\in \mathscr{C}$. Define $\mathscr{D} \coloneqq \{B: \Theta_I(B)\neq \emptyset, \text{ and }\exists \mathscr{C}'\subseteq \mathscr{C}, B = \cup_{A'\in \mathscr{C}'}A'\}$. Because $\Theta_I(A')\neq \emptyset$ for all $A'\in \mathscr{C}$,  $\mathscr{C}\subseteq \mathscr{D}$. Hence, $\mathscr{D}$ is nonempty. Moreover, because intersection of compact sets is compact, we know $\Theta_I(A')$ is compact for any $A'\in \mathscr{D}$. 

Note that $\subseteq$ can be viewed as a partial order for elements within $\mathscr{D}$. We are going to show $\mathscr{D}$ has a maximal element in terms of $\subseteq$, i.e. there exists some $A'\in \mathscr{C}$ such that you cannot find an $A''\in \mathscr{D}$ satisfying $A' \subseteq A''$ and $A' \neq A''$.

To show that $\mathscr{D}$ has a maximal element in terms of $\subseteq$, we are going to invoke the Zorn's lemma. Let $\mathscr{Z}$ be an arbitrary nonempty chain in $\mathscr{D}$. That is, $\mathscr{Z}\neq \emptyset$, $\mathscr{Z}\subseteq \mathscr{D}$ and, for any $A', A''\in \mathscr{Z}$, either $A' \subseteq A''$ or $A'' \subseteq A'$. Define $A^\dagger \coloneqq \cup_{A'\in \mathscr{Z}}A'$. Then, $\Theta_I(A^\dagger) = \cap_{A'\in \mathscr{Z}}\Theta_I(A')$. Because $\mathscr{Z}$ is a chain, $\{\Theta_I(A'): A'\in \mathscr{Z}\}$ is also a chain in terms of $\subseteq$. Because $\mathscr{Z}\subseteq \mathscr{D}$,  $\Theta_I(A')$ is nonempty and compact for any $A'\in \mathscr{Z}$. Hence, Lemma \ref{lem:chain_compact} (stated and proved below) implies that $\Theta_I(A^\dagger)$ is nonempty. As a result, $A^\dagger \in \mathscr{D}$. Moreover, for any $A'\in \mathscr{Z}$, $A' \subseteq A^\dagger$. Thus, $\mathscr{D}$, as a partially ordered set in terms of $\subseteq$, has the following property: every nonempty chain $\mathscr{Z}$ in $\mathscr{D}$ has an upper bound $A^\dagger$ in $\mathscr{D}$. By Zorn's lemma, this implies that $\mathscr{D}$ has a maximal element in terms of $\subseteq$.

Let $A^*$ be a maximal element of $\mathscr{D}$ in terms of $\subseteq$. Because $A^*\in \mathscr{D}$, $\Theta_I(A^*) \neq \emptyset$. Because we have $\cap_{\widetilde{A}\in \mathscr{C}}\Theta_I(\widetilde{A}) = \Theta_I\left(\cap_{\widetilde{A}\in \mathscr{C}}\widetilde{A}\right) = \emptyset$ when $\Theta_I(A) = \emptyset$, $A^* \neq \cup_{\widetilde{A}\in \mathscr{C}}\widetilde{A}$. Therefore, there must exist some $A'\in \mathscr{C}$ such that $A'$ is not a subset of $A^*$. Moreover, because $\mathscr{C}\subseteq \mathscr{D}$, and because $A^*$ is a maximal element of $\mathscr{D}$ in terms of $\subseteq$, $A^*\cup A' \notin \mathscr{D}$. This implies that $\Theta_I(A^*\cup A') = \emptyset$. Because $\Theta_I(A^*\cup A')=\Theta_I(A^*)\cap \Theta_I(A')$, $\Theta_I(A^*)\cap \Theta_I(A') = \emptyset$. Because $A^*\in \mathscr{D}$, there exists $\mathscr{C}'\subseteq \mathscr{C}$ such that $A^* = \cup_{\widetilde{A}\in \mathscr{C}'}\widetilde{A}$ and, hence, $\Theta_I(A^*) = \cap_{\widetilde{A}\in \mathscr{C}'}\Theta_I(\widetilde{A})$. Because $\Theta_I(\widetilde{A})$ is compact for each $\widetilde{A}\in \mathscr{C}'$, Lemma \ref{lem:compact_disjoint} (shown and proved below) implies that there exists $\{A_1,...,A_n\}\subseteq \mathscr{C}'\subseteq \mathscr{C}$ for some finite $n$ such that $\Theta_I(\cup_{i=1}^nA_i)\cap \Theta_I(A') = \emptyset$ and $\Theta_I(\cup_{i=1}^nA_i) \neq \emptyset$. This proves \eqref{eq:statement1_in_proof}.

\noindent \emph{\bf Case 2: the $\mathscr{C}$ in Assumption \ref{assu:mis_1} has finite elements}. Enumerate $\mathscr{C}$ as $\mathscr{C} = \{A_1, ..., A_K\}$. For any $k \in \{1, ..., K\}$, define $B_k = \cup_{i=1}^k A_i$. For any $k$, $\Theta_I(B_k) = \cap_{i=1}^k \Theta_I(A_i)$. Define $\mathcal{N}= \{k\in \{1,...,K\}: \Theta_I(B_k) \neq \emptyset\}$. Because $\Theta_I(A_1)\neq \emptyset$, $\mathcal{N}$ is nonempty and $1\in \mathcal{N}$. Let $n$ be the largest element in $\mathcal{N}$. By construction, $\Theta_I(B_{n}) \neq \emptyset$. By Assumption \ref{assu:mis_1}, $\Theta_I(B_K) = \Theta_I(\cup_{i=1}^KA_i) = \emptyset$. Therefore, we know $n<K$ and $\Theta_I(B_{n + 1}) = \emptyset$. Because $B_{n + 1} = B_{n}\cup A_{n + 1}$, we know that $\Theta_I(B_{n}\cup A_{n + 1})  = \Theta_I(B_{n})\cap \Theta_I(A_{n + 1})= \emptyset$.
This proves \eqref{eq:statement1_in_proof} with $A' = A_{n+1}$.

We have proven \eqref{eq:statement1_in_proof} in both above cases, which completes the proof for the first result of the theorem.

For the second part of the results, when $\Theta_I(A) = \emptyset$, for any $B\subseteq A$ with $\Theta_I(B)\neq \emptyset$, we want to show there exists two finite subsets $B', B''\subseteq A$ such that $\Theta_I(B\cup B')\neq \emptyset$, $\Theta_I(B'')\neq \emptyset$ and $\Theta_I(B\cup B')\cap \Theta_I(B'') = \emptyset$. Let $A'$ and $A''$ be the two finite subsets of $A$ stated in \eqref{eq:statement1_in_proof}. Consider the following two cases:

\begin{enumerate}
\item Suppose $\Theta_I(A')\cap \Theta_I(B) = \emptyset$. Then, let $B' = \emptyset$, $B'' = A'$. We have that $\Theta_I(B\cup B')\neq \emptyset$, $\Theta_I(B'')\neq \emptyset$ and $\Theta_I(B\cup B')\cap \Theta_I(B'') = \emptyset$. 
\item Suppose $\Theta_I(A')\cap \Theta_I(B) \neq \emptyset$. Then, $\Theta_I(A'\cup B) =  \Theta_I(A')\cap \Theta_I(B) \neq \emptyset$. Moreover, $\Theta_I(A'\cup B) \cap \Theta_I(A'') = \Theta_I(B) \cap (\Theta_I(A') \cap \Theta_I(A'')) = \emptyset$. Let $B' = A'$ and $B'' = A''$. Then, we have that $\Theta_I(B\cup B')\neq \emptyset$, $\Theta_I(B'')\neq \emptyset$ and $\Theta_I(B\cup B')\cap \Theta_I(B'') = \emptyset$. 
\end{enumerate}

This completes the proof of Theorem \ref{thm:misleading}.

\begin{lemma}\label{lem:chain_compact}
Let $\mathscr{B}$ be a collection of nonempty compact sets within metric space $\mathcal{T}$. Moreover, suppose $\mathscr{B}$ is a nonempty chain in terms of $\subseteq$, i.e. for any $B, B'\in \mathscr{B}$, either $B\subseteq B'$ or $B' \subseteq B$. Then, $\cap_{B\in \mathscr{B}}B$ is nonempty.
\end{lemma}

\begin{proof}
For the purpose of contradiction, suppose $\cap_{B\in \mathscr{B}}B$ is empty. For any $B\in \mathscr{B}$, let $B^C$ denote the complement of $B$. Because the complement of $\cap_{B\in \mathscr{B}}B$ is $\cup_{B\in \mathscr{B}}B^C$, empty $\cap_{B\in \mathscr{B}}B$ implies that $\cup_{B\in \mathscr{B}}B^C = \mathcal{T}$. Pick an arbitrary $B'\in \mathscr{B}$. The fact that $\cup_{B\in \mathscr{B}}B^C = \mathcal{T}$ implies that $\{B^C: B\in \mathscr{B}\}$ is an open cover of $B'$. Because $B'$ is compact, there exists a finite $\{B_1, ..., B_n\} \subseteq \mathscr{B}$ with $n< \infty$ such that $B' \subseteq \cup_{i=1}^n B^C_i$. This is equivalent to $B' \cap (\cap_{i=1}^n B_i) = \emptyset$. In other words, we can find $n + 1$ elements in $\mathscr{B}$ whose intersection is empty. This contradicts the assumptiton that $\mathscr{B}$ is a chain in terms of $\subseteq$ and that every set in $\mathscr{B}$ is nonempty.
\end{proof}

\begin{lemma}\label{lem:compact_disjoint}
Let $\mathscr{B}$ be a collection of nonempty compact sets within metric space $\mathcal{T}$ and $\cap_{B\in \mathscr{B}}B \neq \emptyset$. Let $B'$ be a nonempty compact set in $\mathcal{T}$ such that $B' \cap (\cap_{B\in \mathscr{B}}B) = \emptyset$. Then, there exists $\{B_1,...,B_n\}\subseteq \mathscr{B}$ for some finite $n$ such that $\cap_{i=1}^n B_i \neq \emptyset$ and $B' \cap (\cap_{i=1}^n B_i) = \emptyset$.
\end{lemma}
\begin{proof}
Because $B' \cap (\cap_{B\in \mathscr{B}}B) = \emptyset$, we know $\{B^C: B\in \mathscr{B}\}$ is an open cover of $B'$. Because $B'$ is compact, there must exist a finite subcover $\{B^C_1,..., B^C_n\}\subseteq \{B^C: B\in \mathscr{B}\}$ for $B'$. This implies that $B'\cap (\cap_{i=1}^n B_i) = \emptyset$. Finally, because $\cap_{i=1}^n B_i \supseteq (\cap_{B\in \mathscr{B}}B)$, $\cap_{i=1}^n B_i\neq \emptyset$.
\end{proof}

  \subsection{Proof of Proposition \ref{thm:uncond_minothold}}
  Proposition \ref{thm:uncond_minothold} is an immediate result of the following two lemmas.

  \begin{lemma}\label{lem:lemma1}
    Suppose Assumption \ref{assu:reg} hold and $\overline{\gamma} < \underline{\gamma}$. Define the interval $\mathcal{W}$
    as the following:
    \begin{equation}\label{eq:W_def}
    \mathcal{W} \equiv
      \begin{cases}
        [\overline{\gamma}, \underline{\gamma}]&\text{ if } P(E[\overline{Y}|Z] = \overline{\gamma}) > 0 \text{ and }
          P(E[\underline{Y}|Z] = \underline{\gamma}) > 0\\
        [\overline{\gamma}, \underline{\gamma})&\text{ if } P(E[\overline{Y}|Z] = \overline{\gamma}) > 0 \text{ and }
          P(E[\underline{Y}|Z] = \underline{\gamma}) = 0\\
        (\overline{\gamma}, \underline{\gamma}]&\text{ if } P(E[\overline{Y}|Z] = \overline{\gamma}) = 0 \text{ and }
          P(E[\underline{Y}|Z] = \underline{\gamma}) > 0\\
        (\overline{\gamma}, \underline{\gamma})&\text{ if } P(E[\overline{Y}|Z] = \overline{\gamma}) = 0 \text{ and }
          P(E[\underline{Y}|Z] = \underline{\gamma}) = 0
      \end{cases}
    \end{equation}
    For any integer $m$ and any $h\in \mathcal{H}_m^+$, if $\widetilde{\Theta}(h)$ is nonempty, then
    $\widetilde{\Theta}(h)\cap \mathcal{W}$ is nonempty. 
\end{lemma}

\begin{proof}[Proof of Lemma \ref{lem:lemma1}]
  Since $h$ has $m$ dimensions, we can write $h=(h_1, ..., h_m)$. Then, $\widetilde{\Theta}(h)$ can be characterized as 
  $\widetilde{\Theta}(h) = [\underline{\theta}, \overline{\theta}]$, where 
  \begin{equation*}
    \underline{\theta}  =  \max_i \frac{E[h_i(Z)\underline{Y}]}{E[h_i(Z)]} \quad \text{ and }\quad
    \overline{\theta}  =  \min_i \frac{E[h_i(Z)\overline{Y}]}{E[h_i(Z)]}.
  \end{equation*}

  Let us first prove $\underline{\theta} \le \underline{\gamma}$. Suppose, for the purpose of contradiction,  $\delta \equiv \underline{\theta} - \underline{\gamma} > 0$.  Let $i' \in \argmax_i E[h_i(Z)\underline{Y}] / E[h_i(Z)]$. Then, we
  have
  \begin{equation}\label{eq:temp_nawoihx81}
    E[h_{i'}(Z)(E[\underline{Y}|Z] - \underline{\theta}]) = 0
  \end{equation}
  Because $\delta \equiv \underline{\theta} - \underline{\gamma} > 0$, $E[\underline{Y}|Z] - \underline{\theta} \le -\delta$. In addition, because $h_{i'}$ is nonnegative, we have $E[h_{i'}]\delta \le
  0$, which contradicts to the fact that $\delta > 0$ and $E[h_{i'}(Z)] > 0$. Moreover, if
  $P(E[\underline{Y}|Z] = \underline{\gamma}) = 0$, then $E[h_i(Z)\underline{Y}] < \underline{\gamma}\cdot E[h_i(Z)]$
  for all $i$ so that $\underline{\theta} < \underline{\gamma}$.

  Similarly, we can show $\overline{\theta} \ge \overline{\gamma}$, and that $\overline{\theta} > \overline{\gamma}$ if
  $P(E[\overline{Y}|Z] = \overline{\gamma}) = 0$. These result then implies that $\widetilde{\Theta}(h)\cap \mathcal{W}
  \neq \emptyset$ whenever $\widetilde{\Theta}(h) \neq \emptyset$.
\end{proof}

\begin{lemma}\label{lem:lemma2}
Suppose Assumption \ref{assu:reg} hold and $\overline{\gamma} < \underline{\gamma}$. Let $\mathcal{W}$ be the interval
defined as in \eqref{eq:W_def}. Then, for any $\theta\in \mathcal{W}$, there exists some $h\in \mathcal{H}^+_2$ such
that $\widetilde{\Theta}(h) = \{\theta\}$.
\end{lemma}
\begin{proof}[Proof of Lemma \ref{lem:lemma2}]
  Fix any $\theta\in \mathcal{W}$.
Define $\underline{S}^+ = \left\{z: E[\underline{Y}|Z=z] \ge \theta\right\}$, 
$\underline{S}^- = \left\{z: E[\underline{Y}|Z=z] \le \theta\right\}$,
$\overline{S}^+ = \{z: E[\overline{Y}|Z=z] \ge \theta\}$  and
$\overline{S}^- = \{z: E[\overline{Y}|Z=z] \le \theta\}$.
Note that, for any $\vartheta > \overline{\gamma}$, the definition of $\overline{\gamma}$ implies that $P(\vartheta
\ge E[\overline{Y}|Z]) > 0$. When $P(E[\overline{Y}|Z] = \overline{\gamma}) > 0$, we also have $P(\vartheta \ge E[\overline{Y}|Z]) > 0$ for any $\vartheta \ge \overline{\gamma}$. Since $\theta\in \mathcal{W}$, we
conclude that $P(Z\in \overline{S}^-) > 0$. Similarly, that $\theta\in \mathcal{W}$ also implies that $P(Z \in
\underline{S}^+) > 0$. 
 Moreover, since $E[\underline{Y}|Z] \le
E[\overline{Y}|Z]$ almost surely, we know $\underline{S}^+ \subseteq \overline{S}^+$ and $\overline{S}^-
\subseteq \underline{S}^-$ almost surely. Therefore, $P(Z\in \underline{S}^-) > 0$ and $P(Z\in \overline{S}^+)
> 0$. 

Next, we show there exists some nonnegative function $h_{1}$ which satisfies
$E[\underline{Y}h_1(Z)] = \theta$ and $E[h_1(Z)] = 1$. Define $h_{1}^{+}(z) = \indicator(z\in
\underline{S}^{+})/P(Z\in \underline{S}^{+})$ and $h_{1}^{-}(z) = \indicator(z\in
\underline{S}^{-})/P(Z\in \underline{S}^{-})$. By construction, $h_{1}^{+}$ and $h_{1}^{-}$ are
nonnegative, and $E[h_{1}^{+}(Z)]= 1$ and $E[h_{1}^{-}(Z)]= 1$. Moreover,
$E[\underline{Y}h_{1}^{+}(Z)]\ge  \theta \ge E[\underline{Y}h_{1}^{-}(Z)]$. Hence, there must exists some
$q\in [0, 1]$ such that $E[\underline{Y} (qh_{1}^{-}(Z) + (1-q)h_{1}^{+}(Z))] = \theta$. Let $h_{1}
= qh_{1}^{-}(Z) + (1-q)h_{1}^{+}(Z)$. Then, such $h_{1}$ satisfies $E[\underline{Y}h_1(Z)] = \theta$ and
$E[h_1(Z)] = 1$. Similarly, there exists some nonnegative function $h_{2}$ which satisfies
$E[\overline{Y}h_2(Z)] = \theta$ and $E[h_2(Z)] = 1$.

Then, $E[h_1(Z)(\tilde{\theta} - \underline{Y})] \ge 0$  is equivalent to $\tilde{\theta} \ge \theta$. To see
this, note that
\begin{eqnarray*}
&& E[h_1(Z)(\tilde{\theta} - \underline{Y})] \ge 0\\
  \Leftrightarrow && E[h_1(Z)] \tilde{\theta} \ge E[h_1(Z)\underline{Y}]  \\
  \Leftrightarrow && \tilde{\theta} \ge \theta
\end{eqnarray*}
where the second equivalence follows from $E[\underline{Y}h_1(Z)] = \theta$ and $E[h_1(Z)] = 1$.
Similarly, we can show $E[h_2(Z)(\overline{Y} - \tilde{\theta})] \ge 0$ is equivalent to $\tilde{\theta} \le
\theta$. 
Let $h=(h_1, h_2)$. These equivalence relation implies that if $\tilde{\theta}\in \widetilde{\Theta}(h)$, then
$\tilde{\theta} = \theta$. 

Moreover, we have
\begin{eqnarray*}
  E[h_2(Z)\theta] & = & \theta \\
  & = & E[h_2(Z) \overline{Y}]\\
  & \ge & E[h_2(Z) \underline{Y}]
\end{eqnarray*}
where the first equality follows from $E[h_2(Z)] = 1$, and the second equality follows from $\theta = E[h_2(Z)
\overline{Y}]$, and the last inequality comes from $E[\underline{Y}|Z] \le E[\overline{Y}|Z]$ almost surely.
Similarly, we can show $E[h_1(Z)\theta] \le E[h_1(Z) \overline{Y}]$. Therefore, $\theta\in
\widetilde{\Theta}(h)$. As a result, $\widetilde{\Theta}(h)= \left\{\theta\right\}$. 
\end{proof}

\subsection{Proof of Theorem \ref{thm:compatible_submodel}} \label{sec:proof-comp-submodels} 
Theorem \ref{thm:compatible_submodel} is a corollary of Theorem \ref{thm:compatible_submodel_advanced} which is  proved below in Section \ref{sec:proof-comp-submodels-advanced}.

\subsection{Proof of Proposition \ref{prop:Monotone_Exclusion_ATE}}

Recall the notation used in this example: $\underline{Y}_d \equiv Y\indicator(D = d) + \underline{y}_d \indicator(D
\neq d)$, $\overline{Y}_d \equiv Y\indicator(D = d) + \overline{y}_d \indicator(D \neq d)$,
$\underline{q}_{dt}\equiv E[\underline{Y}_d|Z = t]$ and $\overline{q}_{dt} \equiv E[\overline{Y}_{d}|Z = t]$. 
Proposition \ref{prop:Monotone_Exclusion_ATE} is an immediate corollary of the following two lemmas.

\begin{lemma}\label{lem:monotone_IV_idset}
In model $Y = \sum_{z\in \mathcal{Z}}\indicator(Z = z)[Y_{1z}D + Y_{0z}(1-D)]$ where $\mathcal{Z}=\left\{ 1,2,...,k
\right\}$. Fix an arbitrary $z^* = 1,...,k$. 
Let $\Theta_{I, z^*}$ be the identified set of $a_{z^*}$, i.e. the identified set of \ref{enu:support}, \ref{enu:MI} and
\ref{enu:monotone} for $z = z^*$. Then,
\begin{enumerate}
  \item $\Theta_{I, z^*} \neq \emptyset$ if and only if the following two conditions hold for each $d\in \{0, 1\}$:
  \begin{equation}\label{eq:criteria_1}
    \forall z < z^*,\quad \max(\underline{q}_{dt}: t \le z) \le \min(\overline{q}_{dt}: t \ge z)
  \end{equation}
  and 
  \begin{equation}\label{eq:criteria_2}
   \max(\underline{q}_{dt}: t = 1,...,k ) \le \min(\overline{q}_{dt}: t \ge z^*)
  \end{equation}
\item if $\Theta_{I, z^*}\neq \emptyset$, then $\Theta_{I, z^*} = \Gamma_{1,z^*}\times \Gamma_{0, z^*}$.
\end{enumerate}
\end{lemma}

\begin{lemma}\label{lem:MI_IV_idset}
In model $Y = \sum_{z\in \mathcal{Z}}\indicator(Z = z)[Y_{1z}D + Y_{0z}(1-D)]$ where $\mathcal{Z}=\left\{ 1,2,...,k
\right\}$.  Let $\Theta_{I}$ be the identified set of $a^\dagger$, i.e. the identified set of \ref{enu:support} and \ref{enu:MI}. Then,
\begin{enumerate}
  \item $\Theta_{I}\neq \emptyset$ if and only if $P\big(Y\in [\underline{y}_d, \overline{y}_d]|D = d\big) = 1$ for any
    $d\in \{0, 1\}$.
  \item when $\Theta_{I}\neq \emptyset$, $\Theta_{I} = \left[E[\underline{Y}_1], E[\overline{Y}_1]\right] \times
  \left[E[\underline{Y}_0], E[\overline{Y}_0]\right]$.
\end{enumerate}
\end{lemma}


\begin{proof}[Proof of Lemma \ref{lem:monotone_IV_idset}]
  The results of this lemma can be divided into the following two parts:
\begin{enumerate}
  \item For any $z^*=1,...,k$, $\Theta_{I, z^*}\neq \emptyset$ only if that \eqref{eq:criteria_1} and
    \eqref{eq:criteria_2} hold for each $d=0,1$.  Moreover, $\Theta_{I, z^*} \subseteq \Gamma_{1,z^*}\times
    \Gamma_{0,z^*}$.
\item if \eqref{eq:criteria_1} and \eqref{eq:criteria_2} hold, then $\Theta_{I, z^*}\neq \emptyset$ and
  $ \Theta_{I, z^*} \supseteq \Gamma_{1,z^*}\times \Gamma_{0,z^*}$. 
 \end{enumerate}
 Let us now prove these two parts one by one.

 \paragraph{\underline{Part 1}}
Fix any $d\in \{0, 1\}$. Suppose assumption $a_{z^*}$ hold, i.e. assumptions \ref{enu:support}, \ref{enu:MI} and
\ref{enu:monotone} hold for $z = z^*$. Assumption
\ref{enu:monotone} implies that for any $z' < {z^*}$ and $t \le z'$ we have $Y_{dt} \le Y_{dz'}$ so that $E[Y_{dt}|Z
= z']\le E[Y_{dz'} | Z = z']$. Due to \ref{enu:MI}, we know $E[Y_{dt}|Z = z'] = E[Y_{dt}|Z = t]$, so that
$E[Y_{dt}|Z = t]\le E[Y_{dz'} | Z = z']$. Since $\underline{q}_{dt} \le E[Y_{dt}|Z = t]$, we conclude that
$\max_{t\le z'}\underline{q}_{dt} \le E[Y_{dz'}|Z = z']$. Similarly, \ref{enu:monotone} implies that for any $z' < {z^*}$
and $t \ge z'$, we have $Y_{dz'} \le Y_{dt}$ so that $E[Y_{dz} | Z = z'] \le E[Y_{dt}|Z = z']$. Because of
\ref{enu:MI}, and because $\overline{q}_{dt} \ge E[Y_{dt}|Z = t]$,  we know that $E[Y_{dz'} | Z = z'] \le \min_{t\ge z}
\overline{q}_{dt}$. Hence, for any $d\in \{0, 1\}$,
\begin{equation}\label{eq:bound_1}
    \forall z' < {z^*},\quad \max(\underline{q}_{dt}: t \le z') \le E[Y_{dz'} | Z = z'] 
    \le \min(\overline{q}_{dt}: t \ge z')
  \end{equation}

  Now, for any $z' \ge {z^*}$, \ref{enu:monotone} implies that $Y_{dt}\le Y_{dz'}$ for any $t\in \{1, ...,k\}$. 
  Hence, $E[Y_{dt}|Z = z'] \le E[Y_{dz'}|Z = z']$ for all $t$. Because \ref{enu:MI} implies that 
  $E[Y_{dt}|Z = t]  =  E[Y_{dt}|Z = z']$, we have $E[Y_{dt}|Z = t] \le E[Y_{dz'}|Z = z']$ for all $t$, 
  so that $\max(\underline{q}_{dt}:  t=1,...,k)
\le E[Y_{dz'}|Z = z']$. For any $z' \ge {z^*}$, assumption \ref{enu:monotone} implies that $Y_{dt}\ge Y_{dz'}$  for all $t\ge z^*$. 
Hence, $E[Y_{dt}|Z = z] \ge E[Y_{dz}|Z = z]$ for all $t \ge {z^*}$. 
Assumption \ref{enu:MI} then implies
 that $E[Y_{dt}|Z = t] \ge E[Y_{dz'}|Z = z']$ for all $t \ge {z^*}$, so that
$\min(\overline{q}_{dt}:  t \ge {z^*}) \ge E[Y_{dz}|Z = z]$. Hence, we conclude that for any $d\in \{0, 1\}$: 
  \begin{equation}\label{eq:bound_2}
  \forall z' \ge {z^*},\quad \max(\underline{q}_{dt}: t = 1,...,k ) \le E[Y_{dz'} | Z = z'] \le \min(\overline{q}_{dt}: t \ge {z^*}).
  \end{equation}

  Combine \eqref{eq:bound_1} and \eqref{eq:bound_2}, we conclude that for any $d$, $\theta_d \in \Gamma_{d,{z^*}}$, so that 
$\Theta_{I,z^*}\subseteq \Gamma_{1,{z^*}}\times \Gamma_{0,{z^*}}$.
  Moreover, because Assumption \ref{enu:support}, \ref{enu:MI} and \ref{enu:monotone} imply \eqref{eq:bound_1} and \eqref{eq:bound_2},
  the violation of \eqref{eq:criteria_1} and \eqref{eq:criteria_2} implies that  $\Theta_{I,z^*} = \emptyset$.
  Equivalently, $\Theta_{I, z^*} \neq \emptyset$ only if \eqref{eq:criteria_1} and \eqref{eq:criteria_2} hold for any $d\in
  \{0, 1\}$.

  \paragraph{\underline{Part 2}} We want to prove that \eqref{eq:criteria_1} and \eqref{eq:criteria_2}
implies that $\Theta_{I, z^*}\neq \emptyset$ and $\Theta_{I, z^*} \supseteq \Gamma_{1,z^*}\times
\Gamma_{0,z^*}$. Fix an arbitrary $d\in \{0, 1\}$. First of all, we are going to prove that one can construct $Y_{dz}$
which achieves the lower bound in $\Gamma_{d,z^*}$, satisfies assumptions \ref{enu:support}-\ref{enu:monotone}, and is compatible with the data at the same time.

Define $\gamma_z$ for each $z=1,...,k$ as follows:
\begin{itemize}
  \item for $z < z^*$, let $\gamma_z$ be the value which solves 
    \begin{equation*}
      \max(\underline{q}_{dt}:t\le z) = E[\indicator(D = d) Y|Z = z] + E[\indicator(D\neq d)Y|Z = z] \gamma_z
    \end{equation*}
    Then, $\gamma_z \in [\underline{y}_d, \overline{y}_d]$ if $\overline{q}_{dz} \ge \max(\underline{q}_{dt}:t\le z)$, which is
    implied by \eqref{eq:criteria_1}. 
  \item for $z\ge z^*$, let $\gamma_z$ be the value which solves
    \begin{equation*}
    \max( \underline{q}_{dt}: t = 1, ..., k) = E[\indicator(D = d)Y|Z = z] + E[\indicator(D\neq d)Y|Z
      = z] \gamma_z
    \end{equation*}
    Then, $\gamma_z \in [\underline{y}_d, \overline{y}_d]$ if $\max( \underline{q}_{dt}: \forall t) \le
    \overline{q}_{dz}$ which is implied by \eqref{eq:criteria_2}. 
\end{itemize}
Define $W_{dz}\equiv\indicator(D = d, Z = z)Y + \indicator(D\neq d, Z = z)\gamma_z$. Then, by construction,
\begin{equation}
  \label{eq:expt_W_dz}
    E[W_{dz}|Z = z] = 
    \begin{cases}
      \max(\underline{q}_{dt}:t\le z) & \text{if } z < z^*\\
      \max( \underline{q}_{dt}: t = 1,...,k)  & \text{if } z \ge z^*
    \end{cases}
\end{equation}
which implies that
\begin{eqnarray}
  \forall z \le t, &&  E[W_{dz}|Z = z] \le E[W_{dt}|Z = t]  \label{eq:W_monotone}\\
  \forall z \ge z^*, && E[W_{dz}|Z = z] = \max( \underline{q}_{dt}: t = 1,...,k) \label{eq:maofwie_2}
\end{eqnarray}
Moreover, because $\gamma_z \in [\underline{y}_d, \overline{y}_d]$ for any $z\in \{1, ..., k\}$, we know $P(W_{dz}\in
[\underline{y}_z, \overline{y}_z]) = 1$ for all $z\in \{1, ..., k\}$. And, $P(W_{dz} = Y | D = d, Z = z) = 1$ for any
$d$ and $z$. 

Now, for any $t\in \{1,..., k\}$, define, $\phi_{dt}(\alpha)\equiv (1-\alpha) W_{dt} + \alpha \overline{y}_d$ and
$\psi_{dt}(\alpha)\equiv(1-\alpha)\underline{y}_d + \alpha W_{dt}$.  We claim that, 
for any $t\neq z$, there exists $\alpha_{tz}\in [0, 1]$ which solves the following equations:
\begin{equation}\label{eq:alpha_equation}
  \begin{array}{ll}
    \forall t < z, & E[W_{dz}|Z = z] = E[\phi_{dt}(\alpha_{tz})|Z = t] \vspace{0.4em} \\
  \forall t > z, & E[W_{dz}|Z = z] = E[\psi_{dt}(\alpha_{tz})|Z = t].
\end{array}
\end{equation}
To see why it is so, note that
\begin{equation}
  \label{eq:mfqiobn23}
  \begin{array}{rl}
    \forall t < z, & E[W_{dt}|Z = t] = E[\phi_{dt}(0)|Z = t] \text{ and } E[\phi_{dt}(1)|Z = t] = \overline{y}_{d},\vspace{0.5em}\\
  \forall t > z, & \underline{y}_d = E[\phi_{dt}(0)|Z = t]  \text{ and } E[\phi_{dt}(1)|Z = t] = E[W_{dt}|Z = t].
  \end{array}
\end{equation}
These results, combined with  \eqref{eq:W_monotone}, imply that 
\begin{eqnarray*}
  \forall t < z, && E[\phi_{dt}(0)|Z = t] \le E[W_{dz}|Z = z] \le E[\phi_{dt}(1)|Z = t],\\
  \forall t > z, && E[\psi_{dt}(0)|Z = t] \le E[W_{dz}|Z = z] \le E[\phi_{dt}(1)|Z = t].
\end{eqnarray*}
which implies the existence of $\alpha_{tz}\in [0, 1]$ satisfying \eqref{eq:alpha_equation} for all $t\neq z$. 

In addition, $(\alpha_{tz}: t\neq z)$ has some extra properties. Because \eqref{eq:W_monotone} holds and  $E[\phi_{dt}(\alpha)|Z = t]$ is an increasing function of $\alpha$, 
\begin{equation}\label{eq:Q1}
\forall t < z< z',\quad \alpha_{tz} \le \alpha_{tz'}.
\end{equation}
Because \eqref{eq:W_monotone} holds and $E[\psi_t(\alpha)|Z = t]$ is an increasing function of $\alpha$,
\begin{equation}\label{eq:Q_2}
\forall z < z' < t,\quad \alpha_{tz} \le \alpha_{tz'}. 
\end{equation}

Construct $Y_{dz} \equiv \sum_{t < z} \phi_{dt}(\alpha_{tz}) + W_{dz} + \sum_{t > z} \psi_{dt}(\alpha_{tz})$.
Because $P(W_{dz}\in [\underline{y}_d, \overline{y}_d]) = 1$, assumption \ref{enu:support} holds for this $Y_{dz}$. 
Because of \eqref{eq:alpha_equation}, assumption \ref{enu:MI} holds for this $Y_{dz}$, 
i.e. $E[Y_{dz}|Z = t] = E[Y_{dz}|Z = z]$ for any $t$, $z$ with $t\neq z$.  

To show assumption \ref{enu:monotone} also holds for this $Y_{dz}$,  note that, for any $z_1$, $z_2$ with $1 \le z_1 < z_2 \le k$, 
\begin{itemize}
  \item If $Z < z_1$, $Y_{dz_1} = \phi_{dZ}(\alpha_{Zz_1}) \le \phi_{dZ}(\alpha_{Z z_2}) = Y_{dz_2}$ because of
    \eqref{eq:Q1} and because $\phi_{dZ}(\alpha)$ is increasing in $\alpha$. 
  \item If $Z = z_1$, $Y_{dz_1} = W_{dz_1} \le \phi_{dZ}(\alpha_{Z z_2}) = Y_{dz_2}$ because
    of the definition of $\phi_{dZ}(\alpha)$.
  \item If $z_1 < Z < z_2$, $Y_{dz_1} = \psi_{dZ}(\alpha_{Zz_1}) \le W_{dZ} \le \phi_{dZ}(\alpha_{Zz_2})
    = Y_{dz_2}$ because of the definition of $\phi_{dZ}(\alpha)$ and $\psi_{dZ}(\alpha)$.
  \item If $Z = z_2$, $Y_{dz_1} = \psi_{dZ}(\alpha_{Zz_1}) \le W_{dZ} = Y_{dz_2}$ because of the definition of
    $\psi_{dZ}(\alpha)$. 
  \item If $z_2 < Z$, $Y_{dz_1} = \psi_{dZ}(\alpha_{Zz_1}) \le \psi_{dZ}(\alpha_{Zz_2}) = Y_{dz_2}$ because of
    \eqref{eq:Q_2} and because $\psi_{dZ}(\alpha)$ is increasing in $\alpha$. 
\end{itemize}
As a result, $Y_{dz_1} \le Y_{dz_2}$ almost surely for any $z_1 \le z_2$. Moreover, because of \eqref{eq:maofwie_2},
$\alpha_{tz} = \alpha_{tz'}$ for any $t$, $z$ and $z'$ with $t < \min(z, z')$ and $z^* \le \min(z,z')$. Because of 
\eqref{eq:maofwie_2} and \eqref{eq:mfqiobn23},  $\alpha_{tz} = 0$ for any $z^* \le t < z$, and $\alpha_{tz} = 1$ for any
$t > z \ge z^*$. Given these results, one can show that for any $z' \ge z^*$,
$Y_{dz'} = \sum_{t < z^*}\phi_{dt}(\alpha_{tz^*}) + \sum_{ t= z^*}^k W_{dt}$.  
This implies that assumption \ref{enu:monotone} also holds. So far, we have shown that $Y_{dz}$ constructed above
satisfies assumption $a_{z^*}$.

Finally, because $E[Y_{dz}] = E[Y_{dz}|Z = z] = E[W_{dz}|Z = z]$ and because of \eqref{eq:expt_W_dz}, we know $\sum_{z}P(Z
= z)E[Y_{dz}]$ achieves the lower bound in $\Gamma_{d,z}$. Moreover, because $P[Y_{dz} = Y|D = d, Z =z] = 1$, this
construction of $Y_{dz}$ is consistent with the data.  
Combine all the above results, for an arbitrary $d\in \{0, 1\}$, we have constructed $Y_{dz}$ which satisfies
assumption $a_{z^*}$ and, at the same time,  $\sum_{z}P(Z = z)E[Y_{dz}]$ achieves the lower bound of
$\Gamma_{d,z}$. 

  Similarly, one can construct $Y_{dz}$ which satisfies
  assumption $a_{z^*}$ and $\sum_{z}P(Z = z)E[Y_{dz}]$ achieves the upper bound of
  $\Gamma_{d,z}$, by defining $\gamma'_z$ as follows:
  \begin{itemize}
    \item for $z < z^*$, let $\gamma'_z$ be the value which solves 
      \begin{equation*}
        \min(\overline{q}_{dt}:t \ge z) = E[\indicator(D = d) Y|Z = z] + E[\indicator(D\neq d)Y|Z = z] \gamma_z
      \end{equation*}
    \item for $z\ge z^*$, let $\gamma'_z$ be the value which solves
      \begin{equation*}
       \min(\overline{q}_{dt}:  t\ge z^*) = E[\indicator(D = d)Y|Z = z] + E[\indicator(D\neq d)Y|Z
        = z] \gamma_z
      \end{equation*}
  \end{itemize}
  Following the same steps as before except replacing $\gamma_{z}$ with $\gamma'_z$, one can show that the constructed
  $Y_{dz}$ satisfies \ref{enu:support}-\ref{enu:monotone} and $\sum_{z}P(Z = z)E[Y_{dz}]$ achieves the upper
  bound of $\Gamma_{d,z}$. 

  Taking convex combinations of the constructions which achieve the upper and lower bound, every point in $\Gamma_{d,z}$
  can be achieved under assumption \ref{enu:support}-\ref{enu:monotone}. This completes the proof.
\end{proof}

\begin{proof}[Proof of Lemma \ref{lem:MI_IV_idset}]
  Suppose \ref{enu:support} and \ref{enu:MI} hold. For any $z\in \{1, 2, ..., k\}$ and any $d\in \{0, 1\}$, we have
  \begin{equation*}
    \indicator(Z = z, D = d)Y + \indicator(Z \neq z \text{ or } D\neq d)\underline{y}_d \le Y_{dz} \le 
\indicator(Z = z, D = d)Y + \indicator(Z \neq z \text{ or } D\neq d)\overline{y}_d
  \end{equation*}
  Therefore, $\underline{q}_{dz}\le E[Y_{dz}|Z = z] \le \overline{q}_{dz}$. Because of \ref{enu:MI}, this implies
  that $\underline{q}_{dz}\le E[Y_{dz}] \le \overline{q}_{dz}$. As a result, $E[\underline{Y}_d] \le
  \sum_{z}P(Z = z) E Y_{dz} \le E[\overline{Y}_d]$, which proves that $\Theta_I \subseteq
  \left[E[\underline{Y}_1], E[\overline{Y}_1]\right] \times \left[E[\underline{Y}_0],
  E[\overline{Y}_0]\right]$. Moreover, when $P\big(Y\in [\underline{y}_d, \overline{y}_d]|D = d\big) = 1$ for
  any $d\in \{0, 1\}$ fails to hold, \ref{enu:support} will fail to hold.  Hence, $\Theta_I \neq \emptyset$ only if
  $P\big(Y\in [\underline{y}_d, \overline{y}_d]|D = d\big) = 1$ for any $d\in \{0, 1\}$. 

  Suppose that $P\big(Y\in [\underline{y}_d, \overline{y}_d]|D = d\big) = 1$ for any
  $d\in \{0, 1\}$ hold.  Then, we know that for each
 $z=1,...,k$ and each $d$, $\underline{y}_d \le \underline{q}_{dz} \le \overline{q}_{dz} \le \overline{y}_d$. Construct
 $Y_{dz}$ as the following for each $z$ and $d$:
 \begin{equation*}
   Y_{dz} = \indicator(Z = z, D = d)Y + \indicator(Z = z, D\neq d) \underline{y}_d + \indicator(Z\neq z)
   \underline{q}_{dz}. 
 \end{equation*}
 By construction, $\theta_d = \sum_{z}P(Z = z) E Y_{dz} = \sum_{z}P(Z = z)\underline{q}_{dz} = E[\underline{Y}_d]$.
 Moreover, one can check that this construction also satisfies assumptions \ref{enu:support} and \ref{enu:MI}. 
 Similarly, for each $d$,  we can construct $Y'_{dz}$ as 
 \begin{equation*}
   Y'_{dz} = \indicator(Z = z, D = d)Y + \indicator(Z = z, D\neq d) \overline{y}_d + \indicator(Z\neq z)
   \overline{q}_{dz}. 
 \end{equation*}
 Again, $Y'_{dz}$ satisfies assumptions \ref{enu:support} and \ref{enu:MI} by construction. In addition,  $\theta_d
 = \sum_{z}P(Z = z) E Y'_{dz}  = E[\overline{Y}_d]$. By considering $(Y_{1z}, Y'_{0z})$, $(Y'_{1z}, Y_{0z})$, $(Y_{1z},
 Y_{0z})$ and $(Y'_{1z}, Y'_{0z})$, we conclude that $\Theta_I$ is nonempty and $\Theta_I = \left[E[\underline{Y}_1],
 E[\overline{Y}_1]\right] \times \left[E[\underline{Y}_0], E[\overline{Y}_0]\right]$.
\end{proof}

\subsection{Proof of Proposition \ref{prop:intersection_bound}}
Recall the $A$ in the introductory example is the set of all condition \eqref{eq:uncond_implementation} indexed by $h\in \mathcal{H}^+_1$.

Let us first show $\Theta^*_I$ is equal the interval specified in \eqref{eq:introductory_example_mfaoiw}. By Lemma \ref{lem:lemma2}, we know that for each
$\theta\in (\overline{\gamma}, \underline{\gamma})$, there exists some $A' \subseteq A$ such that $\Theta_I(A')
= \{\theta\}$. By Theorem \ref{thm:existence_MDR}, there exists some minimum data-consistent relaxation $A^*$ such that
$A'\subseteq A^*$. Since $\Theta_I(A')$ is singleton, we know $\Theta_I(A^*) = \Theta_I(A') = \{\theta\}$. Therefore,
$(\overline{\gamma}, \underline{\gamma})\subseteq \Theta^*_I$. 

We claim that if $P(E[\underline{Y}|Z]\le \overline{\gamma}) > 0$, then $\overline{\gamma}\in \Theta^*_I$ and there
exists some $A'\subseteq A$ with $\Theta_I(A') = \{\overline{\gamma}\}$. To see why it
is so, suppose $P(E[\underline{Y}|Z]\le \overline{\gamma}) > 0$. Then, define $S_1 = \{z: E[\underline{Y}|Z=z] \le
\overline{\gamma} \}$ and $S_2= \{z: E[\underline{Y}|Z=z] \ge \overline{\gamma} \}$. Since $P(E[\underline{Y}|Z]\le
\overline{\gamma}) > 0$, we know $P(Z\in S_1) > 0$. Since $\underline{\gamma} > \overline{\gamma}$, we know $P(Z\in S_2)
>  0 $. Now, define $h_1(z) = \indicator\left( z\in S_1\right) /P(Z\in S_1) $ and $h_2(z) = \indicator(z\in S_2)
/ P(Z\in S_2)$. Then, $Eh_1(Z) = 1$, $Eh_2(Z) = 1$, $E[h_1(Z)\underline{Y}] \le \overline{\gamma}$ and
$E[h_2(Z)\overline{Y}] \ge \overline{\gamma}$. Therefore, there must exists $\underline{h}$ as a convex combination of
$h_1$ and $h_2$ such that $E \underline{h}(Z) = 1$ and $E[\underline{h}(Z) \underline{Y}] = \overline{\gamma}$. Hence,
$\overline{\gamma}\in \widetilde{\Theta}(\underline{h})$ and $\widetilde{\Theta}(\underline{h})\cap (-\infty,
\overline{\gamma}) = \emptyset$. 
Moreover, for each $i=1,2,...$, construct $\overline{h}_i(z)$ as
$\overline{h}_i(z) = \indicator(E[\overline{Y}|Z=z]\in [\overline{\gamma}, \overline{\gamma}
+ 1/i])$. By the definition of $\overline{\gamma}$, we know $E\overline{h}_i(Z) > 0$ for each $i \ge 1$. Note that
the identified set of \eqref{eq:uncond_implementation} of $\overline{h}_i$, $\widetilde{\Theta}(\overline{h}_i)$ is
\begin{equation*}
  \left[ \frac{E[\overline{h}_i(Z)\underline{Y}]}{E\overline{h}_i(Z)},
  \frac{E[\overline{h}_i(Z)\overline{Y}]}{E\overline{h}_i(Z)} \right]. 
\end{equation*}
Because $E[\underline{Y}|Z] \le E[\overline{Y}|Z]$ almost surely, the law of iterated expectation implies that
$\overline{\gamma} \in \widetilde{\Theta}(\overline{h}_i)$. Moreover, by construction, for any $\theta
> \overline{\gamma}$, $\theta\notin \cap_i \widetilde{\Theta}(\overline{h}_i)$. Therefore, if we define $\mathcal{H}'
= \{\overline{h}_i: i\ge 1\}\cup \{\underline{h}\}$, we have $\cap_{h\in \mathcal{H}'}\widetilde{\Theta}(h)
= \{\overline{\gamma}\}$. This implies that $\overline{\gamma}\in \Theta^*_I$ and there
exists some $A'\subseteq A$ with $\Theta_I(A') = \{\overline{\gamma}\}$.

Next, we claim that if $P(E[\underline{Y}|Z] \le \overline{\gamma}) > 0$, then $\Theta^*_I \cap (-\infty,
\overline{\gamma}) = \emptyset$. To see this, note that Lemma \ref{lem:lemma1} implies that for any $a\in A$,
$\Theta_I(a)\cap [\overline{\gamma}, \underline{\gamma}] \neq \emptyset$. Let $A'$ be an arbitrary minimum data-consistent
relaxation $A'$ of $A$.  Because Assumption \ref{assu:mis_2} holds in this example, the preceding result implies that  for
any minimum data-consistent relaxation $A'$ of $A$, we know $\Theta_I(A')\cap [\overline{\gamma}, \underline{\gamma}] \neq
\emptyset$.  Our claim will be verified if we can prove $\Theta_I(A')\cap (-\infty, \overline{\gamma}) = \emptyset$.
Suppose not, i.e suppose $\Theta_I(A') \cap (-\infty, \overline{\gamma}) \neq \emptyset$. Because $\Theta_I(A')$ is a closed
interval, the fact that $\Theta_I(A')\cap [\overline{\gamma}, \underline{\gamma}] \neq \emptyset$ implies that
$\overline{\gamma}\in \Theta_I(A')$. Because we've proven that $\cap_{h\in \mathcal{H}'}\widetilde{\Theta}(h)
= \{\overline{\gamma}\}$, and because $A'$ is a minimum data-consistent relaxation, we know $\Theta_I(A')
= \{\overline{\gamma}\}$ which leads to contradiction.

Next, we claim that if $P(E[\underline{Y}|Z] \le \overline{\gamma}) = 0$, then $\Theta^*_I \cap (-\infty,
\overline{\gamma}] = \emptyset$. To
see this, note that $P(E[\underline{Y}|Z] \le \overline{\gamma}) = 0$ implies $P(E[\underline{Y}|Z] > \overline{\gamma})
= 1$. Therefore, for any $h\in \mathcal{H}^+_1$, $E[h(Z) (\theta - \underline{Y})] \ge 0$ implies that 
\begin{eqnarray*}
&&E[h(Z) (\theta - \underline{Y})] \ge 0\\
&\Rightarrow & E[h(Z) \underline{Y}] \le E[h(Z)] \theta \\
& \Leftrightarrow & E[h(Z)E[\underline{Y}|Z]] \le E[h(Z)] \theta \\
& \Rightarrow & E[h(Z)] \overline{\gamma} < E[h(Z)]\theta
\end{eqnarray*}
where the last inequality follows from the fact that $P(E[\underline{Y}|Z] > \overline{\gamma}) = 1$. Therefore, we know
for any $h\in \mathcal{H}^+_1$, $\widetilde{\Theta}(h)\cap (-\infty, \overline{\gamma}] = \emptyset$. This implies
$\Theta^*_I \cap (-\infty, \overline{\gamma}] = \emptyset$.

Following similar steps as above, we can also prove the following results:
\begin{itemize}
  \item If $P(E[\overline{Y}|Z] \ge \underline{\gamma}) > 0$, then $\underline{\gamma}\in \Theta^*_I$ and there exists
    some $A'\subseteq A$ with $\Theta_I(A') = \{\underline{\gamma}\}$.
  \item If $P(E[\overline{Y}|Z] \ge \underline{\gamma}) > 0$, then $\Theta^*_I \cap (\underline{\gamma}, +\infty)
    = \emptyset$. 
  \item If $P(E[\overline{Y}|Z] \ge \underline{\gamma}) = 0$, then $\Theta^*_I \cap [\underline{\gamma}, +\infty)
    = \emptyset$. 
\end{itemize}

Combining these results and that $(\overline{\gamma}, \underline{\gamma})\subseteq \Theta^*_I$, we conclude that 
$\Theta^*_I$ is equals to the interval specified in \eqref{eq:introductory_example_mfaoiw}. 


\section{Proofs for Additional Results}\label{sec:additional_result_proof}
\subsection{Proof for Proposition \ref{prop:moment_ineq_misleading}}
We need to verify that the $\mathscr{C}$ constructed in \eqref{eq:mathscrC_in_moment_ineq} satisfies all three requirements in Assumption \ref{assu:mis_1}. 

First, we are going to show that $\forall A'\in \mathscr{C}$, $A'$ is data-consistent and consists of finite elements in $A$.  Fix an arbitrary $A'\in \mathscr{C}$.
By the construction of $\mathscr{C}$, $A'$ only contains one element in $A$: \eqref{eq:uncond_moment_ineq} holds for some $w = h_{z,\epsilon}\in \mathcal{W}^*$. By assumption, there exists some $\theta\in \Theta$ such that $E[m(X;\theta)|Z] \le 0$ for all most every $Z$ satisfying $\norm{Z - z} \le \delta(z)$. This implies that this $\theta$ must also satisfy \eqref{eq:uncond_moment_ineq} with $w = h_{z,\epsilon}$ for all $0 < \epsilon < \delta(z)$. This proves $\Theta_I(A') \neq \emptyset$. 

Second, we need to prove that $\Theta_I(\cup_{A'\in \mathscr{C}}A') = \Theta_I(A)$. Because $\mathcal{W}^* \subseteq \mathcal{W}^+_1$, we know $\cup_{A'\in \mathscr{C}}A' \subseteq A$ so that $\Theta_I(\cup_{A'\in \mathscr{C}}A') \supseteq \Theta_I(A)$. Hence, we only need to show $\Theta_I(\cup_{A'\in \mathscr{C}}A') \subseteq \Theta_I(A)$. By assumption, there exists a function $g(z;\theta)$ such that (\emph{i}) for every $\theta\in \Theta$, $E[m(X,Z;\theta)|Z] = g(Z;\theta)$ almost surely; (\emph{ii}) $g(z;\theta)$ is continuous in $z$ for any given $\theta$; (\emph{iii}) $g(z;\theta)$ is continuous in $\theta$ for any given $z$. Because $\theta\in \Theta_I(A)$ if and only if $\theta$ satisfy \eqref{eq:cond_moment_ineq}, we know that 
\begin{equation*}
\theta \notin \Theta_I(A) \text{ if and only if } P(Z\in \mathcal{Z}_\theta) > 0 \text{ where } \mathcal{Z}_\theta \coloneqq \{z \in \mathcal{Z}: g(z;\theta) > 0\}.
\end{equation*}
Fix an arbitrary $\theta\notin \Theta_I(A)$. Because $g(\cdot,\theta)$ is a continuous function of $z$ given $\theta$, $\mathcal{Z}_\theta$ is an open set of $z$. Therefore, there must exists some $z\in \mathcal{Z}$ and $\epsilon \in (0, \delta(z))$ such that $\{z': \norm{z' - z} < \epsilon \} \subseteq \mathcal{Z}_\theta$. As a result, the law of iterated expectation implies that
\begin{equation*}
E[h_{z,\epsilon}(Z) g(Z;\theta)] = E[h_{z,\epsilon}(Z)m(X;\theta)] > 0.
\end{equation*}
This means that $\theta\notin \Theta_I(\cup_{A'\in \mathscr{C}}A')$. Thus, we have proven that $\theta\notin \Theta_I(\cup_{A'\in \mathscr{C}}A')$ if $\theta\notin \Theta_I(A)$, which is equivalent to $\Theta_I(\cup_{A'\in \mathscr{C}}A') \subseteq \Theta_I(A)$.

Finally, we need to show that $\Theta_I(A')$ is compact for any $A'\in \mathscr{C}$. Fix an arbitrary $A' \in \mathscr{C}$. By the construction of $\mathscr{C}$, $A'$ only contains one element in $A$: \eqref{eq:uncond_moment_ineq} holds for some $w  \in \mathcal{W}^*$. Define $\kappa(\theta) \coloneqq E[w(Z)g(Z;\theta)]$. Because $g(z;\theta)$ is continuous in $\theta$ for any given $z$, and because, by assumption, $\sup_{\theta\in \Theta}\norm{g(z;\theta)} = \sup_{\theta\in \Theta}\norm{E[m(X;\theta)|Z = z]} \le \gamma(z)$ and $E|\gamma(Z)| < \infty$, and because $w(\cdot)$ is a bounded function, the dominated convergence theorem implies that $\kappa(\theta)$ is a continuous at any $\theta\in \Theta$. The law of iterated expectations implies that $E[w(Z)m(X;\theta)] = E[w(Z)g(Z;\theta)] \eqqcolon \kappa(\theta)$.  Therefore, $\Theta_I(A') = \{\theta\in \Theta: \kappa(\theta) \le 0\}$. Because $\kappa$ is continuous in $\theta$, we know $\Theta_I(A')$ is a closed set. Because $\Theta$ is compact by assumption, $\Theta_I(A')$ is compact. 

\subsection{Proof of Proposition \ref{prop:capcity_ineq_prop}}
We need to verify that $\mathscr{C}$ constructed in \eqref{eq:mathscrC_in_random_set} satisfies all three requirements in Assumption \ref{assu:mis_1}

First, we are going to show that $\forall A'\in \mathscr{C}$, $A'$ is data-consistent and consists of finite elements in $A$.  Fix an arbitrary $A'\in \mathscr{C}$. By the construct of $\mathscr{C}$ in \eqref{eq:mathscrC_in_random_set}, $A'$ is a singleton set which only contains one element in $A$: \eqref{eq:artstein_ineq} holds for some $K\subseteq \mathcal{Y}$. If $K = \mathcal{Y}$, \eqref{eq:artstein_ineq} holds for any $\theta\in \Theta$ because $L(K, X;\theta) = 1$ almost surely in this case. If $K = \emptyset$, \eqref{eq:artstein_ineq} holds for any $\theta\in \Theta$ because $P_F(Y\in K|X) = 0$ almost surely in this case.

If $K \subsetneq \mathcal{Y}$ and $K \neq \emptyset$, pick an arbitrary $y'\in \mathcal{Y}\backslash K$. By \ref{enu:finite_y_1}, we know $\inf_{x\in \mathcal{X}} P(Y = y'|X=x) > 0$. Therefore, we know 
\begin{equation*}
\sup_{x\in \mathcal{X}} P_F(Y \in K|X = x)= 1 - \inf_{x\in \mathcal{X}} P_F(Y \notin K|X = x)
\le  1 - \inf_{x\in \mathcal{X}} P_F(Y = y'|X = x) < 1
\end{equation*}
On the other hand, pick an arbitrary $y''\in K$, \ref{enu:finite_y_2} implies that there exists some sequence $\theta_k\in \Theta$ such that 
\begin{equation*}
\inf_{x\in \mathcal{X}} L(\{y''\}, x;\theta_k) \to 1 \text{ as }k\to \infty.
\end{equation*}
Therefore, there must exist some $\theta^*\in \Theta$ such that $\inf_{x\in \mathcal{X}} L(\{y''\}, x;\theta^*) \ge \sup_{x\in \mathcal{X}} P_F(Y \in K|X = x)$. This implies that 
\begin{equation*}
\sup_{x\in \mathcal{X}} P_F(Y \in K|X = x) \le \inf_{x\in \mathcal{X}} L(\{y''\}, x;\theta^*) \le \inf_{x\in \mathcal{X}} L(K, x;\theta^*) 
\end{equation*}
the last equality hold because $y''\in K$ implies that $L(\{y''\}, x;\theta^*) \le L(K, x;\theta^*)$ for any $x$. The above inequality implies that $\theta^*\in \Theta_I(A')$, because
\begin{equation*}
P_F(Y \in K | X) \le \sup_{x\in \mathcal{X}} P_F(Y \in K|X = x) \le \inf_{x\in \mathcal{X}} L(K, x;\theta^*)  \le L(K, X;\theta^*)  \text{ almost surely}
\end{equation*}
Hence, $\Theta_I(A')$ is nonempty. That is, we have shown that every $A'\in \mathscr{C}$ is data-consistent. 

Second, we need to prove $\Theta_I(\cup_{A'\in \mathscr{C}}A') = \Theta_I(A)$. This is trivial, because $\cup_{A'\in \mathscr{C}}A' = A$ by the construction of $\mathscr{C}$ in \eqref{eq:mathscrC_in_random_set}. 

Finally, because $\mathcal{Y}$ is a finite set, $A$ is a finite set. By the construction of $\mathscr{C}$ in \eqref{eq:mathscrC_in_random_set}, $\mathscr{C}$ is also a finite set. This completes the proof.

\subsection{Proof for Theorem \ref{thm:existence_MDR}}
Let us start with two trivial cases: (\emph{i}) suppose $\Theta_I(A')= \emptyset$ for any $A'\subseteq A$. Then, $\emptyset$ is the minimum data-consistent relaxation in this case; (\emph{ii}) suppose $\Theta_I(A) \neq \emptyset$. Then, $A$ is the minimum data-consistent relaxation. Next, let us consider the following nontrivial case.

Suppose $\Theta_I(A)= \emptyset$ and there exits some $A_0\subseteq A$ such that $\Theta_I(A_0) \neq \emptyset$. We are going to show that there exists some minimum data-consistent relaxation $\widetilde{A}$ such that $A_0\subseteq \widetilde{A}$. Let's consider two cases:

\paragraph{\bf Case 1: \ref{enu:T3_1} holds.} Because $\Theta_I(A) = \emptyset$ while $\Theta_I(A') \neq \emptyset$, $A_0$ cannot be $A$ so that $A\backslash A_0$ is nonempty. Because of \ref{enu:T3_1}, $A \backslash A_0$ is a finite set. Enumerate it as $A \backslash A_0 = \{a_1, ..., a_k\}$. Construct $A_1, ..., A_k$ iteratively as follows: For any $i=1,...,k$, define $A_i = A_{i-1}\cup\{a_i\}$ if 
$\Theta_I(A_{i-1}\cup \{a_i\})\neq \emptyset$, and define $A_i = A_{i-1}$ if otherwise. By construction, for each $i=1,...,k$, $\Theta_I(A_i) \neq \emptyset$. Moreover, if $a_i\notin A_k$, we must have $\Theta_I(A_k \cup \{a_i\}) = \emptyset$ because $\Theta_I(A_{i}\cup \{a_i\}) = \emptyset$ and $\Theta_I(A_k \cup \{a_i\}) \subseteq \Theta_I(A_{i}\cup \{a_i\})$. Therefore, $A_k$ must be a minimum data-consistent relaxation, and $A_0 \subseteq A_k$ by construction.

\paragraph{\bf Case 2: \ref{enu:T3_2} holds.} 
Define $\mathscr{A} = \{A': A'\subseteq A,\  A_0\subseteq A'\text{ and }\Theta_I(A')\neq
\emptyset\}$. $\mathscr{A}$ is not empty because $A_0\in \mathscr{A}$. We are going to prove that there exists some minimum data-consistent relaxation $\widetilde{A}$ such that $A_0\subseteq \widetilde{A}$, which is equivalent to show the following statement:
\begin{equation}\label{eq:statement2_in_proof}
\begin{aligned}
&\mathscr{A} \textnormal{ has a maximum element }\widetilde{A}\in \mathscr{A} \textnormal{ in terms of partial order }\subseteq, \\
&\textnormal{ i.e. there is no }A'\in \mathscr{A} \textnormal{ such that } \widetilde{A} \subsetneq A'.
\end{aligned}
\end{equation}

To prove \eqref{eq:statement2_in_proof}, we are going to invoke Zorn's lemma. Let $\mathscr{Z}$ be an arbitrary nonempty chain in $\mathscr{A}$ in terms of $\subseteq$. Because $\mathscr{Z}$ is a chain, for any $A'$ and $A''$ in $\mathscr{Z}$, there is either $A'\subseteq A''$ or $A''\subseteq A'$, so that there is either $\Theta_I(A')\subseteq \Theta_I(A'')$ or $\Theta_I(A'')\subseteq \Theta_I(A')$. Define $A^{\dagger} =\cup_{A'\in \mathscr{Z}}A'$. Because of \ref{enu:T3_2}, $\Theta_I(A^{\dagger}) = \cap_{A'\in \mathscr{Z}}\Theta_I(A')$. By Lemma \ref{lem:chain_compact}, $\Theta_I(A^{\dagger})$ is nonempty. This implies that $A^{\dagger}\in \mathscr{A}$. Moreover, by construction, $A'\subseteq A^{\dagger}$ for any $A' \in \mathscr{Z}$. Therefore, we have shown that there exists an upper bound $A^{\dagger}$ in $\mathscr{A}$ in terms of partial order $\subseteq$ for any nonempty chain $\mathscr{Z}$ in $\mathscr{A}$. Zorn's lemma then implies \eqref{eq:statement2_in_proof}.

\subsection{Proof of Theorem \ref{thm:compatible_submodel_advanced}} \label{sec:proof-comp-submodels-advanced} 

We first prove the first part of the theorem. 

\paragraph{\bf \ref{enu:ind_assumptions} $\Rightarrow$ \ref{enu:unique_MDR}:} 

Construct $A^* = \{a\in A: \Theta_I(a) \neq \emptyset\}$. We are going to show that $A^*$ is the only minimum data-consistent relaxation. For any $a\notin A^*$, we have $\Theta_I(a) = \emptyset$ by the construction of $A^*$. Hence, for any $a\notin A^*$, $\Theta_I(A^*\cup\{a\}) = \emptyset$ because $\Theta_I(A^*\cup\{a\}) \subseteq \Theta_I(a)$. Moreover, \ref{enu:ind_assumptions} implies that $\Theta_I(A^*) \neq \emptyset$ because every $a\in A^*$ is data-consistent. Therefore, $A^*$ is a minimum data-consistent relaxation. 

Suppose, for the purpose of contradiction, there exists another minimum data-consistent relaxation $A'$ different from $A^*$. Because $A'$ is a minimum data-consistent relaxation, there is no $A'\subseteq A^*$. Therefore, $A'\backslash A^*$ must be nonempty. Pick an arbitrary $a'\in A'\backslash A^*$. Because $a'\notin A^*$, we know $\Theta_I(a') = \emptyset$ by the construction of $A^*$. Therefore, $\Theta_I(A') = \emptyset$ because $\Theta_I(A') \subseteq \Theta_I(a')$.  This contradicts to the fact that $A'$  is a minimum data-consistent relaxation.

\paragraph{\bf If either \ref{enu:T3_1} or \ref{enu:T3_2} holds, \ref{enu:unique_MDR} $\Rightarrow$ \ref{enu:ind_assumptions}:} Suppose either \ref{enu:T3_1} or \ref{enu:T3_2} hold. Let $A^*$ denote the unique minimum data-consistent relaxation. First of all, we are going to prove the following statement:
\begin{equation}\label{eq:statement3_in_proof}
A^* = \{a\in A: \Theta_I(a) \neq \emptyset\}.
\end{equation}
Because $A^*$ is a minimum data-consistent relaxation, $\Theta_I(A^*) \neq \emptyset$. Because $\Theta_I(A^*) \subseteq \Theta_I(a)$ for every $a\in A^*$, we know $\Theta_I(a)$ is nonempty for every $a\in A^*$. This implies that $A^* \subseteq \{a\in A: \Theta_I(a) \neq \emptyset\}$. To show, $A^* \supseteq \{a\in A: \Theta_I(a) \neq \emptyset\}$, note that
 \begin{itemize}
 \item when $A^* = A$, we have $A^* \supseteq \{a\in A: \Theta_I(a) \neq \emptyset\}$ trivially.
 \item when $A^* \subsetneq A$. Pick an arbitrary $a'\notin A^*$. Suppose, for the purpose of contradiciton, $\Theta_I(a') \neq \emptyset$. Then, Theorem \ref{thm:existence_MDR} implies that there exists some minimum data-consistent relaxation $\widetilde{A}$ with $a'\in \widetilde{A}$. Because $a'\in \widetilde{A}$ and $a'\notin A^*$, we must have $\widetilde{A}\neq A^*$, which contradicts to $A^*$ being the unique minimum data-consistent relaxation. Hence, $\Theta_I(a') = \emptyset$ for any $a'\notin A^*$, which is equivalent to $A^* \supseteq \{a\in A: \Theta_I(a) \neq \emptyset\}$.
 \end{itemize}
This proves \eqref{eq:statement3_in_proof}.
 
Because of  \eqref{eq:statement3_in_proof}, $a\in A$ is data-consistent if and only if $a\in A^*$. Therefore, for any $A'\subseteq A$, all $a\in A'$ are data-consistent if and only if $A' \subseteq A^*$. As a result, we can show \ref{enu:ind_assumptions} if the following statement is true:
\begin{equation}\label{eq:statement4_in_proof}
A' \subseteq A^* \text{if and only if }\Theta_I(A') \neq \emptyset.
\end{equation}
To see why \eqref{eq:statement4_in_proof} is indeed true, note that
\begin{itemize}
\item Because $\Theta_I(A^*) \neq \emptyset$, and because $\Theta_I(A^*) \subseteq \Theta_I(A')$ for any $A'\subseteq A^*$, we know $\Theta_I(A') \neq \emptyset$ if $A' \subseteq A^*$. 
\item Fix an arbitrary set $A'$ with $\Theta_I(A') \neq \emptyset$. Because $\Theta_I(A') \subseteq \Theta_I(a)$ for every $a\in A'$, we know $\Theta_I(a)\neq \emptyset$ for every $a\in A'$. Because of \eqref{eq:statement3_in_proof}, this implies that $a\in A^*$ for every $a\in A'$, i.e. $A' \subseteq A^*$. Thus, we have shown that $A'\subseteq A^*$ if $\Theta_I(A') \neq \emptyset$.
\end{itemize}
This completes the proof.

\subsection{Proof of Theorem \ref{thm:interpretation}}
Recall that $\mathscr{A}_R$ denote the collection of all minimum data-consistent relaxations. Because either \ref{enu:T3_1} or \ref{enu:T3_2} holds, $\mathscr{A}_R$ is nonempty. 

First of all, we are going to prove $\Theta^*_I$ is rationalizable. Because $\Theta^*_I = \cup_{A'\in \mathscr{A}_R}\Theta_I(A')$, for any minimum data-consitent relaxation $\widetilde{A}$, we must have $\Theta_I(\widetilde{A})\subseteq \Theta^*_I$. Becuase $\mathscr{A}_R$ is nonempty, we know $\Theta^*_I$ is rationalizable.

Second, we are going to prove $\Theta^*_I$ is nonconflicting. Fix an arbitrary data-consistent subset $A'$ of $A$.  By Theorem \ref{thm:existence_MDR}, there exists a minimum data-consistent relaxation $\widetilde{A}$ such that $A'\subseteq \widetilde{A}$. Because $A'\subseteq \widetilde{A}$, we know $\Theta_I(\widetilde{A})\subseteq \Theta_I(A')$, so that $\Theta_I(A')\cap \Theta_I(\widetilde{A}) \neq \emptyset$. Because $\widetilde{A}\in \mathscr{A}_R$ and $\Theta^*_I = \cup_{\widetilde{A}'\in \mathscr{A}_R}\Theta_I(\widetilde{A}')$, we know $\Theta_I(A')\cap \Theta^*_I \neq \emptyset$.

\subsection{Proof for Theorem \ref{thm:smallest_iff_simple}}
Suppose the smallest rationalizable and nonconflicting set exists. Denote it as $S^*$. We are going to show that $S^* = \Theta^*_I$ when \ref{enu:no_nested_MDR} is true. By Theorem \ref{thm:interpretation}, $\Theta^*_I$ is both rationalizable and nonconflicting. Therefore, $S^* \subseteq \Theta^*_I$. What left to show is $\Theta^*_I \subseteq S^*$.

Define $\mathscr{A}_R$ to be the collection of all minimum data-consistent relaxation. Define $\mathscr{A}_1 = \{A'\in \mathscr{A}_R: \Theta_I(A')\subseteq S^*\}$ and $\mathscr{A}_2 = \mathscr{A}_R \backslash \mathscr{A}_1$. Because $S^*$ is rationalizable, there exists some data-consistent $A'\subseteq A$ such that $\Theta_I(A') \subseteq S^*$. By Theorem \ref{thm:existence_MDR}, there exists some $\tilde{A}\in \mathscr{A}_R$ such that $A'\subseteq \tilde{A}$. Therefore, $\Theta_I(\tilde{A})\subseteq \Theta_I(A') \subseteq S^*$. Hence, $\mathscr{A}_1$ is not empty. 

Because $\mathscr{A}_R = \mathscr{A}_1\cup \mathscr{A}_2$, we know
\begin{equation*}
\Theta^*_I = \left(\cup_{A'\in \mathscr{A}_1}\Theta_I(A') \right) \cup \left(\cup_{A'\in \mathscr{A}_2} \Theta_I(A')\right). 
\end{equation*}
Therefore, to show $\Theta^*_I \subseteq S^*$,  we only need to show $\mathscr{A}_2$ is an empty set. We discuss two cases:

\begin{itemize}
\item Suppose there exists some $A^*\in \mathscr{A}_1$ such that $\Theta_I(A^*)$ contains at least two different elements. We are going to show that $\mathscr{A}_2 = \emptyset$ in this case. Suppose, for the purpose of contradiction, $\mathscr{A}_2$ is nonempty. Pick an arbitrary $A^\dagger$ within $\mathscr{A}_2$. Fix, also, an arbitrary element $a^*$ in $\Theta_I(A^*)$. Define $S'$ as follows
	\begin{equation*}
		S' = 
		\begin{cases}
			\emptyset & \text{if }\Theta_I(A^*)\cap \Theta_I(A^\dagger) \neq \emptyset\\
			\{a^*\} &  \text{if }\Theta_I(A^*)\cap \Theta_I(A^\dagger) = \emptyset
		\end{cases}
	\end{equation*}
	Define $S^{\dagger}$ as
	\begin{equation*}
	S^\dagger \coloneqq  \Theta_I(A^\dagger) \cup \left( \cup_{A'\in \mathscr{A}_R} \Theta_I(A')\backslash \Theta_I(A^*) \right) \cup S'.
	\end{equation*}
Because $\Theta_I(A^\dagger) \subseteq S^\dagger$, $S^\dagger$ is rationalizable. Because of \ref{enu:no_nested_MDR}, $\Theta_I(A')\backslash \Theta_I(A^*)\neq \emptyset$ for each $A'\in \mathscr{A}_R$ with $\Theta_I(A')\neq \Theta_I(A^*)$, otherwise we would have $\Theta_I(A')\subsetneq \Theta_I(A^*)$ for some $A'\in \mathscr{A}_R$ which violates \ref{enu:no_nested_MDR}. As a result, $S^\dagger\cap \Theta_I(A')\neq \emptyset$ for any $A'\in \mathscr{A}_R$ with $\Theta_I(A')\neq \Theta_I(A^*)$. Next, because $\Theta_I(A^\dagger)\cup S'$ and $\Theta_I(A^*)$ has nonempty intersection by the construction of $S'$, we know that $S^\dagger \cap \Theta_I(A') \neq \emptyset$ 
for any $A'\in \mathscr{A}_R$ with $\Theta_I(A') = \Theta_I(A^*)$. In total, we know $S^\dagger \cap \Theta_I(A') \neq \emptyset$ for any $A'\in \mathscr{A}_R$. Because of this result, and because, by Theorem \ref{thm:existence_MDR}, for any $A$$\subseteq A$, there exists some $A' \in \mathscr{A}_R$ such that $\Theta_I(A')\subseteq \Theta_I(A'')$, we know $S^\dagger$ is nonconflicting. So far, we have shown that $S^\dagger$ is both rationalizable and nonconflicting.

Then, we claim that there is no $\Theta_I(A^*)\subseteq S^\dagger$. By the construction of $S^\dagger$, $ S^\dagger \cap \Theta_I(A^*) = (\Theta_I(A^\dagger) \cup S') \cap \Theta_I(A^*)$. So, this claim is also equivalent to that there is no $\Theta_I(A^*) \subseteq \Theta_I(A^\dagger) \cup S'$. To see why this claim is true, discuss to cases:
\begin{itemize}
	\item when $\Theta_I(A^*)\cap \Theta_I(A^\dagger) \neq \emptyset$, $S' = \emptyset$. Because of \ref{enu:no_nested_MDR}, there is no $\Theta_I(A^*) \subsetneq \Theta_I(A^\dagger)$. Because $A^*\in \mathscr{A}_1$ and $A^\dagger \in \mathscr{A}_2$, there is no $\Theta_I(A^*) = \Theta_I(A^\dagger)$. In total, there is no $\Theta_I(A^*) \subseteq \Theta_I(A^\dagger)$. Because $S' = \emptyset$, we know there is no $\Theta_I(A^*)\subseteq \Theta_I(A^\dagger) \cup S'$ is this case.
	\item when $\Theta_I(A^*)\cap \Theta_I(A^\dagger) = \emptyset$, $(\Theta_I(A^\dagger) \cup S') \cap \Theta_I(A^*) = S'\cap \Theta_I(A^*) = \{a^*\}$. Because $\Theta_I(A^*)$ contains at least two different elements, there is no $\Theta_I(A^*)\subseteq S'$, which implies there is no $\Theta_I(A^*)\subseteq \Theta_I(A^\dagger) \cup S'$ is this case.
\end{itemize}
In total, we have verify the claim that there is no $\Theta_I(A^*)\subseteq S^\dagger$. However, because $S^*$ is the smallest rationalizable and nonconflicting set, and because $S^\dagger$ is both rationalizable and nonconflicting, there must be $S^* \subseteq S^\dagger$, which further implies $\Theta_I(A^*) \subseteq S^\dagger$ because $A^*\in \mathscr{A}_1$. This leads to the contradiction. As a result, $\mathscr{A}_2$ must be empty in this case.

\item Suppose that, for any $A'\in \mathscr{A}_1$,  $\Theta_I(A')$ only contains one element. We are going to show that $\mathscr{A}_2 = \emptyset$ in this case. Suppose, for the purpose of contradiction, that $\mathscr{A}_2$ is nonempty. By the construction of $\mathscr{A}_1$ and $\mathscr{A}_2$, there is no $\Theta_I(A_1) = \Theta_I(A_2)$  for any $A_1\in \mathscr{A}_1$ and $A_2\in \mathscr{A}_2$. Because of \ref{enu:no_nested_MDR}, there is no $\Theta_I(A_1) \subsetneq \Theta_I(A_2)$  for any $A_1\in \mathscr{A}_1$ and $A_2\in \mathscr{A}_2$. In total, there is no $\Theta_I(A_1) \subseteq \Theta_I(A_2)$ for any $A_1\in \mathscr{A}_1$ and $A_2\in \mathscr{A}_2$. Because  $\Theta_I(A')$ only contains one element for any $A'\in \mathscr{A}_1$, for any $S \subseteq \Theta$, $\Theta_I(A') \cap S \neq \emptyset$ would mean $\Theta_I(A')\subseteq S$. Therefore, we must have $\Theta_I(A')\cap \Theta_I(A'') = \emptyset$ for any $A'\in \mathscr{A}_1$ and $A''\in \mathscr{A}_2$. Define $S_1$ as 
\begin{equation*} 
S_1 = \cup_{A'\in \mathscr{A}_1}\Theta_I(A').
\end{equation*} 
And, define $S_2$ as
\begin{equation*}
S_2 = \left( \cup_{A'\in \mathscr{A}_2} \Theta_I(A') \right) \cap S^*.
\end{equation*}
Because $\Theta_I(A')\cap \Theta_I(A'') = \emptyset$ for any $A'\in \mathscr{A}_1$ and $A''\in \mathscr{A}_2$, $S_1\cap S_2 = \emptyset$. Moreover, because $S^*$ is nonconflicting, $S^*\cap \Theta_I(A')\neq \emptyset$ for each $A'\in \mathscr{A}_2$. Therefore, $S_2\cap \Theta_I(A')\neq \emptyset$ for any $A'\in \mathscr{A}_2$. This implies that $S_1\cup S_2$ is nonconflicting. Moreover, because $\mathscr{A}_1$ is nonempty, $S_1$ is rationalizable so that $S_1\cup S_2$ is also rationalizable. As a result, $S_1\cup S_2$ is both nonconflicting and rationalizable. Next, define 
	\begin{equation*}
	S_3 = \left(\cup_{A'\in \mathscr{A}_2} \Theta_I(A') \right) \backslash S^*.
	\end{equation*}
By the definition of $\mathscr{A}_1$, $S_1\subseteq S^*$. Therefore, $S_1\cap S_3 = \emptyset$. Also, we have $S_2 \cap S_3 = \emptyset$ by the construction of $S_2$ and $S_3$. In addition, by the construction of $\mathscr{A}_2$, for each $A'\in \mathscr{A}_2$, $\Theta_I(A')\backslash S^* \neq \emptyset$. Therefore, $S_3\cap \Theta_I(A') \neq \emptyset$ for each $A'\in \mathscr{A}_2$. This implies that $S_1\cup S_3$ is nonconflicting. Moreover, because $\mathscr{A}_1$ is nonempty, $S_1$ is rationalizable so that $S_1\cup S_3$ is also rationalizable. As a result, $S_1\cup S_3$ is both nonconflicting and rationalizable.

So far, we have shown that $S_1\cup S_2$ is both rationalizable and nonconflicting. And, we have shown that $S_1\cup S_3$ is both rationalizable and nonconflicting. Because $S^*$ is the smallest rationalizable and nonconflicting set, we must have $S^* \subseteq S_1\cup S_2$ and $S^* \subseteq S_1\cup S_3$. In other words, 
\begin{equation*}
S^* \subseteq (S_1\cup S_2) \cap (S_1 \cup S_3)
\end{equation*}
However, because $S_1\cap S_2 = \emptyset$, $S_1\cap S_3 = \emptyset$ and $S_2 \cap S_3 = \emptyset$, we know $(S_1\cup S_2) \cap (S_1 \cup S_3) = S_1$.  As a result, we have
\begin{equation}\label{eq:nafwi23}
S^* \subseteq S_1
\end{equation}
We have already shown that we must have $\Theta_I(A')\cap \Theta_I(A'') = \emptyset$ for any $A'\in \mathscr{A}_1$ and $A''\in \mathscr{A}_2$. Therefore, $S_1 \cap  \Theta_I(A'')  = \emptyset$ for any $A''\in \mathscr{A}_2$. Because $\mathscr{A}_2$ is nonempty, this means that $S_1$ is not nonconflicting. Because of \eqref{eq:nafwi23}, this implies that $S^*$ is not nonconflicting, which contradicts to the fact that $S^*$ is both rationalizable and nonconflicting. 
\end{itemize}
We have shown that $\mathscr{A}_2$ must be empty in both of the above cases. This completes the proof.

\subsection{Proof for Theorem \ref{thm:smallest_if}}
By Theorem \ref{thm:interpretation}, $\Theta^*_I$ is both rationalizable and nonconflicting. If we could show $\Theta^*_I\subseteq S$ for an arbitrary set $S$ that is both rationalizable and nonconflicting, then 
 $\Theta^*_I$ would be the smallest rationalizable and nonconflicting set. Fix an arbitrary set $S$ that is both rationalizable and nonconflicting.

We first prove that \ref{enu:S6_3_1} implies $\Theta^*_I\subseteq S$.  Because $S$ is rationalizable, there exists some data-consistent $A'\subseteq A$ such that $\Theta_I(A')\subseteq S$. By Theorem \ref{thm:existence_MDR}, there exits a minimum data-consistent relaxation $\tilde{A}$ such that $A'\subseteq \tilde{A}$. Because $A' \subseteq \tilde{A}$, we know $\Theta_I(\tilde{A})\subseteq \Theta_I(A')$. Because of \ref{enu:S6_3_1}, $\Theta^*_I = \Theta_I(\tilde{A})$. Therefore, $\Theta^*_I = \Theta_I(\tilde{A})\subseteq \Theta_I(A') \subseteq S$. 

Next, we prove that \ref{enu:S6_3_2} implies $\Theta^*_I \subseteq S$. Let $\mathscr{A}_R$ denote the collection of all minimum data-consistent relaxation. Because $S$ is nonconflicting, we must have $\Theta_I(A')\cap S\neq \emptyset$ for each $A'\in \mathscr{A}_R$. For each $A'\in \mathscr{A}_R$, because $\Theta_I(A')$ only contains one element under \ref{enu:S6_3_2}, $\Theta_I(A')\cap S\neq \emptyset$ is equivalent to $\Theta_I(A')\subseteq  S$. As a result, we know $\Theta^*_I \subseteq S$ because $\Theta^*_I = \cup_{A'\in \mathscr{A}_R}\Theta_I(A')$.

\subsection{Proof for Theorem \ref{thm:compared_to_perturbation_corrected}}
Let $\mathscr{A}_R$ be the set of all minimum data-consitent relaxations. Suppose for each $\tilde{A}\in \mathscr{A}_R$, $\Theta_I(\tilde{A})$ is singleton. Then, we want to show that $\Theta^*_I = \cup_{\widetilde{A}\in \mathscr{A}_R}\Theta_I(\widetilde{A})$ is a subset of $\Theta^\dagger_I$ no matter which type of relaxation is chosen by the researcher. 

To show this result, pick an arbitrary $\tilde{A}\in \mathscr{A}_R$. Define $\delta^*: A\to [0, 1]$ as follows: $\delta^*(a) = 0$ if $a\in \tilde{A}$, and $\delta^*(a) = 1$ if $a\notin \tilde{A}$. By construction, $\Theta_I(A(\delta^*)) = \Theta_I(\tilde{A})$ only contains one element. Therefore, we must have $\delta^*\in FF$, because there cannot exist some $\delta < \delta^*$ with $\Theta_I(A(\delta))\neq \emptyset$ and $\Theta_I(A(\delta)) \subsetneq \Theta_I(A(\delta^*))$. Because this holds for any $\tilde{A}\in \mathscr{A}_R$, we know $\Theta^*_I \subseteq \Theta^\dagger_I$.
\end{appendix}

%




\end{document}